\documentclass[sort&compress]{elsarticle}
\usepackage{amsmath,mathtools,amssymb}
\usepackage{graphicx}   
\usepackage{epstopdf}
\usepackage{subfigure}
\usepackage{color}
\usepackage{verbatim} 
\usepackage{booktabs}
\usepackage{bm}   
\usepackage{dutchcal}
\usepackage[left=3cm,right=3cm,top=2.5cm,bottom=2.5cm]{geometry}  
\usepackage{lineno,hyperref}
\usepackage{tablists} 
\usepackage{makecell} 
\usepackage{tikz}
\usepackage{hyperref}
\usepackage{pstricks}
\usepackage[titletoc]{appendix}
\usepackage{pstricks-add}
\usepackage{float}
\usepackage{mathrsfs} 
\usepackage{amsthm}
\hypersetup{
	colorlinks=ture,
	linkcolor=blue,
	filecolor=gray,
	urlcolor=blue,
	citecolor=blue}
\numberwithin{equation}{section}
\newtheorem{definition}{Definition}[section]

\newtheorem{proposition}{Proposition}[section]
\allowdisplaybreaks[4]

\begin{document}
	\begin{frontmatter}
		\title{Generalized higher connections and Yang-Mills} 
		\author{Danhua Song}
		\ead{danhua_song@163.com}
		\author{Kai Lou}
		\author{Ke Wu}
		\ead{wuke@cnu.edu.cn}
		\author{Jie Yang}
		\address{School of Mathematical Sciences, Capital Normal University, Beijing 100048, China}
		\date{}
		
		\begin{abstract} 
			We first extend Generalized Differential Calculus (GDC) to higher structures and create generalized $G$-invariant bilinear forms. In addition, we also focus on developing generalized 2- and 3-connection theories in the framework of GDC. Then, we derive the higher Bianchi-Identities and study the gauge transformations for those generalized higher connections.  Finally, we establish the generalized 2- and 3-form Yang-Mills theories based on GDC and obtain the corresponding fields equations.
			
		\end{abstract}
		
	\end{frontmatter}

	\tableofcontents

	\section{Introduction}
	In resent years, generalized differential forms have been discussed in a number of literatures and they are employed in many different geometrically and physically contexts. The algebra and calculus of ordinary differential forms are extended to an algebra and calculus of different type of generalized differential forms, namely GDC.  In this paper, we consider the extensions of the generalized differential forms to higher structures. We develop the GDC based on higher groups and construct generalized higher connection theories to develop higher Yang-Mills.
	 
	 The generalized differential forms are based on forms of negative degree which were first introduced by Sparling \cite{Sparling1, Sparling2} in order to extend twistor results to half-flat space-times and to associate an abstract twistor space with general analytic solutions of Einstein's vacuum field equations. 
	 Subsequently, Nurowski and Robinson developed the generalized forms by using the minus $1$-forms \cite{PNR, pnr}.  They established Hodge star operator, co-differential and Laplacian operators on the space of generalized differential forms. These definitions are similar to those for ordinary forms, but there are some significant differences.  
	  Based on those structures above, they found that the generalized forms can be applied in mechanical and physical field theories, such as generalized form of Hamiltonian systems, scalar fields,  Einstein's vacuum field equations, Maxwell and Yang-Mills fields. It is understood that generalized forms can
provide a natural formalism within which to study potentials and gauge
conditions and the extension of zero rest mass systems to systems with nonzero rest mass occurs naturally.
	  Besides, Guo et al. introduced and studied the generalized topological field theories \cite{GHY} in 2002.  They found that a direct relation between Chern-Simons theory and BF theory, as two typical sorts of  topological field theories, is presented by using GDC. They established the generalized Chern-Weil homomorphism for generalized curvature invariant polynomials, and showed that BF gauge theory can be obtained from the generalized second Chern class.  
	  In 2003, they presented a general approach to construct a class of generalized topological field theories with constraints  \cite{ghy}. They found that the ordinary BF formulations of general relativity, Yang-Mills theories, and $\mathscr{N}=1, 2$ chiral super-gravity can  be  reformulations as generalized topological fields with constraints. 
	   Recently, the generalized differential forms were developed by Robinson in a number of works  \cite{Robbinson0, Robinson2, Robbinson4, Robbinson5, Robinson1, Robinson3}.
	   In the mean time, analogous ideas have been explored in the dual context, and generalized vector fields are introduced and studied \cite{Robinson6, gv}.

	By using the GDC, the connection theory on the principle bundle $P(M, G)$ with gauge group $G$ is extended to the generalized connection theory \cite{Robinson1, Robinson2, Robinson3, Robbinson4, Robbinson5}, where the ordinary $\mathcal{g}$-valued connection $1$-form and a curvature $2$-form  are extended to a generalized $\mathcal{g}$-valued connection $1$-form  and a curvature $2$-form.
	It has been shown in \cite{GHY} that the generalized curvature $2$-form satisfies the generalized Bianchi-Identity. As one of the applications of generalized connection theory,  a generalized form version of the Yang-Mills theory, which is the simplest nonabelian gauge theory, is presented in \cite{pnr}. Namely, the ordinary Lagrangian of Yang-Mills theory may be simply written as a generalized term. 
	
	Recently, there are some strong connections coming from physics and from higher category theory,  where 2-gauge theory  \cite{Baez.2010, Baez2005HigherGT, Bar, JBUS} and 3-gauge theory \cite{Faria_Martins_2011, Saemann:2013pca, doi:10.1063/1.4870640, Fiorenza:2012mr} have been developed. Roughly speaking, the 2-gauge theory is about 2-connections on 2-bundles and 3-gauge theory is about 3-connections on 3-bundles. 
	In the related work \cite{2002hep.th....6130B}, Baez generalize the Yang-Mills theory to the 2-form Yang-Mills theory based on the 2-connection theory.
	Based on our previous paper \cite{sdh}, we extend the Yang-Mills theory to the 3-form Yang-Mills theory based on the 3-connection theory.
	In this paper, we intend to develop generalized higher connection theories from GDC view to reformulate the higher Yang-Mills theories. 
	
	What we do in our article can be summarized in the following outline:
		\begin{enumerate}[1. ]
			\item In the section \ref{sec2}, we review the extension of the ordinary forms to the generalized forms. In addition, we introduce the definitions of the generalized exterior product, generalized exterior derivative, hodge star operator and integration for generalized forms. Finally, we give a new integration definition to be used in this paper.
			
			\item  In the section \ref{sec3}, we generalize the Lie algebra valued differential forms to the generalized case. To make sense of all calculations in the generalized higher connection theories, we shall need to establish some new regulations and give properties for the Lie algebra valued generalized differential forms.
			
			\item In the section \ref{sec4}, we shall give the definition of generalized $G$-invariant forms, and give some new generalized maps for the Lie algebra valued generalized differential forms.
			
			\item In the section \ref{sec5}, we develop the generalized higher connections as extensions of ordinary higher connections, and introduce generalized gauge transformations. Finally, we get generalized higher Bianchi-Identities based on GDC.
			
			\item In the section \ref{sec6}, we generalize the higher Yang-Mills to the generalized higher Yang-Mills.
			
		\end{enumerate}
 
In order to make the calculations and arguments of the generalized forms  reasonable, a selection of salient results from previous papers and the detailed calculation processes are given in Appendixes.


	   \section{Generalized differential forms}\label{sec2}
	   In this section, we recall generalized differential forms of type $N$, which is represented by a number of different ways, on an $n$ dimensional manifold $M$ \cite{PNR, Robinson1, Robinson3}.
	   However, we will focus on the type $1$ generalized forms in this paper.
	   	A generalized  $p$-form of type $N=1$ on $M$ is defined by an ordered pair of an ordinary $p$-form and an ordinary $p+1$-form on the same manifold. Correspondingly, the ordinary exterior product $\wedge$ is replaced by the generalized exterior product $\boldsymbol{\wedge}$, and the ordinary exterior differential $d$ is replaced by the generalized exterior differential $\textbf{d}$ that  is also nilpotent.
	   A Hodge star operation $\star$ and the notion of duality for generalized forms are defined first in \cite{pnr}, where the star operation is used in the definition of an inner product of, and co-differential and Laplacian operators for, generalized p-forms. 
	   The integration of type $N=1$ generalized $p$-form over poly-chains is defined in \cite{Robinson3}, and another definition of the integration of type $N=1$ generalized $p$-form is presented in \cite{GHY}.
	   Although these definitions are similar for ordinary forms, there are some significant differences.
	   There are a small number of obvious changes of notation from  reference \cite{pnr}, but most of the conventions are retained.
	   Ordinary forms are denoted by roman letters, and the degree of a form is indicated above it. In particular generalized forms are denoted by  script roman letters to distinguish them from ordinary forms.
	   
       Before that, let's introduce another representation. A basis for type $N$ generalized forms consists of any basis for ordinary forms on $M$ augmented by $N$ linearly independent minus $1$-forms \{$\xi^i$\}, ($i = 1, 2, ... , N$).  Minus $1$-forms are assumed to satisfy the ordinary distributive and associative laws of exterior algebra and the exterior product rule
  $$\xi^i \xi^j = - \xi^j \xi^i, \ \ \  \overset{p}{a} \, \xi^i = (-1)^p  \xi^i  \,\overset{p}{a} ,$$
  together with a condition of linear independence, $\xi^1 \xi^2 ... \xi^N \neq 0 $ and $\xi ^2 =0$, and where $\overset{p}{a}$ is any ordinary $p$-form.
  In order to ensure that their exterior derivatives are zero-forms and that $d^2 = 0$, they are required to satisfy the condition $d \, \xi^i = k^i$, where $k^i$ is a constant for any $i=1, ..., N$.
	    Consequently, a generalized $p$-form of type $N$ will be written as
	    \begin{align}\label{gpf}
	    \mathscr{M} = \overset{p}{a} +  \overset{p+1}{a}_{i_1} \, \xi^{i_1} + \dfrac{1}{2!} \overset{p+2}{a}_{i_1 i_2} \, \xi^{i_1}\xi^{i_2} + ... + \dfrac{1}{j!} \overset{p+j}{a}_{i_1 ... i_j} \, \xi^{i_1} ... \xi^{i_j} + \dfrac{1}{N!} \overset{p+N}{a}_{i_1 ... i_N} \, \xi^{i_1} ... \xi^{i_N}.
	    \end{align}
 Here $\overset{p}{a}$, $\overset{p+1}{a}_{i_1}$, ..., $\overset{p+j}{a}_{i_1 ... i_j} = \overset{p+j}{a}_{[i_1 ... i_j]}$, ..., $\overset{p+N}{a}_{i_1 ... i_N} $ are ,respectively, ordinary $p$-, $(p+1)$-, ..., $(p+j)$-, ..., $(p+N)$- ordinary forms; $j = 1, 2, ..., N$ and $i_1, ..., i_j, ..., i_N$ range and sum over $1, 2, ..., N$. It then follows that generalized forms satisfy the usual distributive and associative laws of exterior algebra, together with the product rule $\mathscr{M}\boldsymbol{\wedge} \mathscr{N} = (-1)^{pq} \mathscr{N} \boldsymbol{\wedge} \mathscr{M}$, where $\mathscr{N}$ is a generalized $q$-form of type $N$.
  From $\eqref{gpf}$, we can see that a generalized $p$-form of type $N = 0$ is an ordinary differential $p$-form. 
  Further discussion of type $N$ generalized forms can be found in \cite{Robinson1}.
  
  In our work, we consider the type $N=1$ generalized $p$-form $\mathscr{M}$. Let
           \begin{align}
              \mathscr{M} = \overset{p}{a} +  \overset{p+1}{a}\xi,
          \end{align}
where $\overset{p}{a}$ and $\overset{p+1}{a}$ are, respectively, ordinary $p$- and $(p+1)$-forms and $p$ can take integer values from $-1$ to $n$. When $p=-1$, $\overset{-1}{a} = 0$.
Hence, given a minus $1$-form $\xi$, the generalized $p$-form of type $N=1$ $\mathscr{M}$ can be written as an ordered pair of ordinary $p$- and $p+1$-forms, that is
	   \begin{align}
	    \mathscr{M} = ( \overset{p}{a} ,  \overset{p+1}{a}) \in \Lambda^p (M)\times \Lambda^{p+1}(M),
	   \end{align}
where $\Lambda^p(M)$ denotes the module of ordinary $p$-forms on the differentiable manifold $M$. A minus $1$-form is defined to be an ordered pair
       \begin{align}
      p = -1: \ (0, \overset{0}{a}),
       \end{align}   
where $\overset{0}{a}$ is a function on $M$.  A function on $M$ will be identified with the generalized $0$-form 
        \begin{align}
     p =0:\ (\overset{0}{a}, 0).
       \end{align}

	  In this paper, we will represent the generalized forms by the ordered pairs.  If $ \mathscr{M} = ( \overset{p}{a} ,  \overset{p+1}{a})$ and $\mathscr{N}= ( \overset{q}{b} ,  \overset{q+1}{b})$, then the generalized exterior product and the generalized exterior derivative are rewritten as
	  \begin{align}
	  	 \mathscr{M} \boldsymbol{\wedge} \mathscr{N} = (\overset{p}{a} \wedge \overset{q}{b},  \overset{p}{a} \wedge \overset{q+1}{b} + (-1)^q \,\overset{p+1}{a} \wedge  \overset{q}{b}),
	  \end{align}
	   \begin{align}
	   	\textbf{d} \mathscr{M} = (d \, \overset{p}{a} + (-1)^{p+1} k \, \overset{p+1}{a},  d \, \overset{p+1}{a}),
	   \end{align}
where $k$ is a constant which is assumed to be non-zero.	   
	   
	   A hodge star operator for generalized forms \cite{pnr} is defined by
	   \begin{align}\label{star}
	   \star \mathscr{M} = ((-1)^{n - p -1} \ast \overset{p+1}{a}, \ast \overset{p}{a}),
	   \end{align}
	   where $\ast$ is  the hodge star operator for ordinary forms.
	   Then the duality of a generalized $p$-form is a $(n-p-1)$-form
 rather than a $(n-p)$-form.
	   Hence,
	   \begin{align}
	   	\mathscr{M}\boldsymbol{\wedge} \star \mathscr{N}= ((-1)^{n - q -1} \, \overset{p}{a} \wedge \ast \overset{q+1}{b}, \overset{p}{a} \wedge \ast \overset{q}{b} + \overset{p+1}{a} \wedge \ast \overset{q+1}{b}).
	   \end{align}
	   
	   An integration of  the type $N$ generalized form is defined on a $p$-poly-chain of type $N$ \cite{Robinson3, Robbinson5}. The type $N=1$ poly-chains are defined to be ordered pairs of ordinary $p$- and $p+1$-chains in $M$
	   \begin{align}
	  \bm{ \tilde{c}}_p = (c_p, c_{p+1}).
	   \end{align}
	   For a generalized $p$-form $ \mathscr{M} = ( \overset{p}{a} ,  \overset{p+1}{a})$, the integration
on  $\bm{ \tilde{c}}_p = (c_p, c_{p+1})$ can be defined as usual by
	   \begin{align}\label{in1}
	   	   \int_{\bm{ \tilde{c}}_p} \mathscr{M} =  \int_{\bm{ \tilde{c}}_p}( \overset{p}{a} ,  \overset{p+1}{a}) = \int_{c_p}  \overset{p}{a} + \int_{c_{p+1}}  \overset{p+1}{a}.
	   \end{align}
   In related work \cite{GHY}, there is a different definition
      \begin{align}\label{in2}
   	\int_{\bm{ \tilde{c}}_p} \mathscr{M} =  \int_{\bm{ \tilde{c}}_p}( \overset{p}{a} ,  \overset{p+1}{a}) = \int_{c_p}  \overset{p}{a}.
   \end{align}
As we shall see, the following definition will be naturally associated to \eqref{in1} and similar to \eqref{in2}
    \begin{align}\label{inter}
   	\int_{\bm{ \tilde{c}}_p} \mathscr{M} =  \int_{\bm{ \tilde{c}}_p}( \overset{p}{a} ,  \overset{p+1}{a}) = \int_{c_{p+1}}  \overset{p+1}{a}.
   \end{align}
 The above definition opens up a possibility to use it as a way of integrating the generalized differential forms on the $n$ dimensional manifold $M$.  
Given a generalized $(n-1)$-form $\mathscr{M}= (\overset{n-1}{a}, \overset{n}{a})$ on $M$,  we  define its integration on $(\partial M, M)$ as
\begin{align}\label{in3}
\int_{(\partial M, M)} \mathscr{M} = \int_{(\partial M, M)} (\overset{n-1}{a}, \overset{n}{a}) =  \int_{M} \overset{n}{a},
\end{align}
hopefully ending up with some desired results in this paper.

The method of calculation used in this work can be extended to the case of higher type, but the computations will be increasingly lengthy, and detail of this case are beyond the scope of the present article.
Besides, it is clear that these results can extend straightforwardly to the case where the generalized forms have values in the Lie algebra of a Lie group. In the next section, we shall construct the generalized forms with values in differential (2-)crossed module, which is also one of the main results of this paper.
 It can easily be seen directly that generalized higher connections can be constructed by replacing ordinary forms with generalized forms valued in the algebras of higher groups in the usual definitions.
In this paper, we restrict our attention to Lie $2$-algebras and Lie $3$-algebras, which are essentially the same as differential crossed modules and differential $2$-crossed modules respectively \cite{Faria_Martins_2011, 2002hep.th....6130B, Crans} and references therein.  

\section{Generalized forms valued in differential crossed module and 2-crossed module}\label{sec3}
Let us consider a Lie crossed module $(H,G; \alpha,\vartriangleright)$ \cite{ Beaz, Crans, Brown}, which  is an algebraic structure specified by two Lie groups $G$ and $H$ , together with the homomorphism $\alpha$ and an action $\vartriangleright$ of the group $G$ on $H$, and $(\mathcal{h},\mathcal{g}; \alpha,\vartriangleright)$ is the
associated differential crossed module. 
In addition, we consider a Lie 2-crossed module $(L,H,G; \beta, \alpha, \vartriangleright, \left\{,\right\})$ \cite{Radenkovic:2019qme, Martins:2010ry, TRMV1, YHMNRY, Mutlu1998FREENESSCF}, together with the homomorphisms $\alpha$ and $\beta$, where $\vartriangleright$ is a smooth action of the Lie group $G$ on all three Lie groups $G$, $H$ and $L$ by automorphisms and $\{, \} : H \times H \longrightarrow L$ is a $G$-equivariant map, and $(\mathcal{l},\mathcal{h},\mathcal{g};\beta,\alpha,\vartriangleright,\left\{,\right\})$ is the associated differential 2-crossed module.  See Appendix \ref{CM} for more details.
The definitions and properties of the ordinary differential forms with values in differential (2-)crossed module have been introduced in our previous paper \cite{SDH} and  related work \cite{doi:10.1063/1.4870640}. We shall give a small taste of these reviews to the literature, seeing more details in Appendix \ref{LAVDF}.

In this section, we shall develop the type $N=1$ generalized forms valued in differential (2-)crossed module, which will be given by the same formulas as in ordinary forms. All related maps in differential (2-)crossed module, such as $\alpha$, $\beta$, and so on, can be first extended in a obvious way to the generalized case. Hence, the generalized calculations shall be developed in this case. 

We can see from the introduction in the previous section that a Lie algebra $\mathcal{g}$-valued generalized $p$-form $\mathscr{A}$ is defined as a $\mathcal{g}$-valued ordered pair of a $\mathcal{g}$-valued ordinary $p$-form  $\overset{p}{A} \in \Lambda^p(M, \mathcal{g})$ and a $\mathcal{g}$-valued ordinary $p+1$-form  $\overset{p+1}{A} \in \Lambda^{p+1}(M, \mathcal{g})$
\begin{align}
\mathscr{A}= (\overset{p}{A}, \overset{p+1}{A}).
\end{align}
 Let $\Lambda^p (M, \mathcal{g},N)$ be a vector space of $\mathcal{g}$-valued generalized differential $p$-forms of type $N$. Note that, $\Lambda^p(M, \mathcal{g},0) =\Lambda^p(M, \mathcal{g})$ is the ordinary case. In addition, it is a straightforward matter to show that there is an identity
 \begin{align}
 	\mathscr{A} \boldsymbol{\wedge}^{\bm{\left[,\right]}} \mathscr{A'} = (\overset{p}{A}\wedge^{[,]}\overset{q}{A'}, \overset{p}{A}\wedge^{[,]}\overset{q+1}{A'} + (-1)^q \overset{p+1}{A}\wedge^{[,]}\overset{q}{A'}),
 \end{align}
 due to $\mathscr{A} \boldsymbol{\wedge}^{\bm{\left[,\right]}} \mathscr{A'}  =  \mathscr{A} \boldsymbol{\wedge}\mathscr{A'} + (-1)^{pq+1}\mathscr{A'} \boldsymbol{\wedge}\mathscr{A} $ and where $\mathscr{A'}=(\overset{q}{A'}, \overset{q+1}{A'})$.
 
 \begin{definition}
Given a crossed module $(H,G; \alpha,\vartriangleright)$ with the underlying differential structure $(\mathcal{h},\mathcal{g}; \alpha,\vartriangleright)$, generalized maps  $\bm{\alpha}: \Lambda^t(M, \mathcal{h}, 1) \longrightarrow  \Lambda^t(M, \mathcal{g}, 1) $ and $\blacktriangleright : \Lambda^p(M, \mathcal{g}, 1) \times \Lambda^t(M, \mathcal{h}, 1) \longrightarrow   \Lambda^{p+t}(M, \mathcal{h}, 1) $ are defined to be 
\begin{align}\label{alpha}
	\bm{\alpha}(\mathscr{B}) := (\alpha(\overset{t}{B}), \alpha(\overset{t+1}{B})),
\end{align}
 and 
\begin{align}\label{action1}
	\mathscr{A} \boldsymbol{\wedge}^{\blacktriangleright} \mathscr{B}  := (\overset{p}{A} \wedge^{\vartriangleright} \overset{t}{B}, \overset{p}{A}\wedge^{\vartriangleright} \overset{t+1}{B} + (-1)^t \overset{p+1}{A}\wedge^{\vartriangleright} \overset{t}{B}),
\end{align}
where $\mathscr{A}= (\overset{p}{A}, \overset{p+1}{A}) \in  \Lambda^p (M, \mathcal{g}, 1) $ and $\mathscr{B} = (\overset{t}{B}, \overset{t+1}{B}) \in  \Lambda^t (M, \mathcal{h}, 1) $.
 \end{definition}
 
 \begin{definition}
 	Given a 2-crossed module $(L,H,G; \beta, \alpha, \vartriangleright, \left\{,\right\})$ with the underlying differential structure $(\mathcal{l},\mathcal{h},\mathcal{g};\beta,\alpha,\vartriangleright,\left\{,\right\})$, generalized maps $\bm{\alpha}$ and $\blacktriangleright$ are, respectively, same as \eqref{alpha} and \eqref{action1}, besides there are generalized maps $\bm{\beta}: \Lambda^q(M, \mathcal{l}, 1) \longrightarrow  \Lambda^q(M, \mathcal{h}, 1) $, $\blacktriangleright : \Lambda^p(M, \mathcal{g}, 1) \times \Lambda^q(M, \mathcal{l}, 1) \longrightarrow   \Lambda^{p+q}(M, \mathcal{l}, 1) $ and $\bm{\{, \}}: \Lambda^{t_1}(M, \mathcal{h}, 1) \times \Lambda^{t_2}(M, \mathcal{h}, 1) \longrightarrow   \Lambda^{t_1 + t_2}(M, \mathcal{l}, 1)$  defined as follows
 	\begin{align}
 		\bm{\beta}(\mathscr{C}) := (\beta(\overset{q}{C}), \beta(\overset{q+1}{C})),
 	\end{align}
 \begin{align}
 	\mathscr{A} \boldsymbol{\wedge}^{\blacktriangleright} \mathscr{C}  := (\overset{p}{A} \wedge^{\vartriangleright} \overset{q}{C}, \overset{p}{A}\wedge^{\vartriangleright} \overset{q+1}{C} + (-1)^q \overset{p+1}{A}\wedge^{\vartriangleright} \overset{q}{C}),
 \end{align}
and 
\begin{align}
	\mathscr{B}_1\boldsymbol{\wedge}^{\bm{\{, \}}}\mathscr{B}_2:= (\overset{t_1}{B_1}\wedge^{\{,\}}\overset{t_2}{B_2}, \overset{t_1}{B_1}\wedge^{\{,\}}\overset{t_2+1}{B_2}+ (-1)^{t_2}\overset{t_1+1}{B_1}\wedge^{\{,\}}\overset{t_2}{B_2}),
\end{align}
where $\mathscr{A}= (\overset{p}{A}, \overset{p+1}{A}) \in  \Lambda^p (M, \mathcal{g}, 1) $, $\mathscr{C} = (\overset{q}{C}, \overset{q+1}{C}) \in  \Lambda^q (M, \mathcal{l}, 1) $ and $\mathscr{B}_i = (\overset{t_i}{B_i}, \overset{t_i +1}{B_i}) \in \Lambda^{t_i} (M, \mathcal{h}, 1)(i = 1, 2)$.
 \end{definition}

The following propositions shall give some properties for those generalized maps which are similar to the corresponding ordinary maps in the formalism.

\begin{proposition}
	For $\mathscr{A} \in  \Lambda^k (M, \mathcal{g}, 1)$, $\mathscr{A}_1\in  \Lambda^{k_1} (M, \mathcal{g}, 1)$, $\mathscr{A}_2\in  \Lambda^{k_2} (M, \mathcal{g}, 1)$, $\mathscr{B} \in \Lambda^t (M, \mathcal{h},1)$, $\mathscr{B}_1 \in \Lambda^{t_1} (M, \mathcal{h}, 1)$ and $ \mathscr{B}_2 \in \Lambda^{t_2} (M, \mathcal{h}, 1)$, have
	\begin{enumerate}[1)]
		\setlength{\itemsep}{6pt}
		\item $\bm{ \alpha} (\mathscr{A}\boldsymbol{\wedge}^{\blacktriangleright} \mathscr{B})=\mathscr{A}\boldsymbol{\wedge}^{\bm{\left[,\right]}}\bm{ \alpha}(\mathscr{B})$;
		
		\setlength{\itemsep}{6pt}
		\item $\bm{ \alpha} (\mathscr{B}_1)\boldsymbol{\wedge}^\blacktriangleright \mathscr{B}_2= \mathscr{B}_1\boldsymbol{\wedge}^{\bm{\left[,\right]}}\mathscr{B}_2$;
		
		\setlength{\itemsep}{6pt}
		\item $\mathscr{A}\boldsymbol{\wedge}^{\blacktriangleright} (\mathscr{B}_1\boldsymbol{\wedge}^{\bm{\left[,\right]}} \mathscr{B}_2)=(\mathscr{A}\boldsymbol{\wedge}^{\blacktriangleright} \mathscr{B}_1)\boldsymbol{\wedge}^{\bm{\left[,\right]}}\mathscr{B}_2+ (-1)^{kt_1}\mathscr{B}_1\boldsymbol{\wedge}^{\bm{\left[,\right]}}(\mathscr{A}\boldsymbol{\wedge}^{\blacktriangleright} \mathscr{B}_2)$;
		
		\setlength{\itemsep}{6pt}
		\item $(\mathscr{A}_1\boldsymbol{\wedge}^{\bm{\left[,\right]}}\mathscr{A}_2)\boldsymbol{\wedge}^{\blacktriangleright} \mathscr{B}=\mathscr{A}_1 \boldsymbol{\wedge}^{\blacktriangleright}(\mathscr{A}_2\boldsymbol{\wedge}^{\blacktriangleright} \mathscr{B})+(-1)^{k_1 k_2 +1} \mathscr{A}_2\boldsymbol{\wedge}^{\blacktriangleright}(\mathscr{A}_1\boldsymbol{\wedge}^{\blacktriangleright} \mathscr{B})$.
	\end{enumerate}
\end{proposition}
              \begin{proof}
                  	\begin{enumerate}[1)]
		               	\setlength{\itemsep}{6pt}
		                \item 	Let $\mathscr{A}= (\overset{k}{A}, \overset{k+1}{A})$, and $\mathscr{B} = (\overset{t}{B}, \overset{t+1}{B})$, then
		                \begin{align*}
		               \bm{ \alpha} (\mathscr{A}\boldsymbol{\wedge}^{\blacktriangleright} \mathscr{B}) & =   \bm{ \alpha}(\overset{k}{A}\wedge^{\vartriangleright}\overset{t}{B}, \overset{k}{A}\wedge^{\vartriangleright}\overset{t+1}{B} + (-1)^t \overset{k+1}{A}\wedge^{\vartriangleright}\overset{t}{B})\\
		               &= (\alpha(\overset{k}{A}\wedge^{\vartriangleright}\overset{t}{B}), \alpha(\overset{k}{A}\wedge^{\vartriangleright}\overset{t+1}{B} + (-1)^t \overset{k+1}{A}\wedge^{\vartriangleright}\overset{t}{B}))\\
		               &= (\overset{k}{A}\wedge^{[,]}\alpha(\overset{t}{B}), \overset{k}{A}\wedge^{[,]}\alpha(\overset{t+1}{B}) + (-1)^t \overset{k+1}{A}\wedge^{[,]}\alpha(\overset{t}{B}))\\
		               & = \mathscr{A}\boldsymbol{\wedge}^{\bm{\left[,\right]}}\bm{ \alpha}(\mathscr{B}),
		                \end{align*}
		by using (\ref{C1}).
                 \item 	Let $\mathscr{B}_1 = (\overset{t_1}{B_1}, \overset{t_1+1}{B_1})$ and $\mathscr{B}_2 = (\overset{t_2}{B_2}, \overset{t_2+1}{B_2})$, then
                \begin{align*}
                	\bm{ \alpha} (\mathscr{B}_1)\boldsymbol{\wedge}^\blacktriangleright \mathscr{B}_2 & =  	\bm{ \alpha} (\overset{t_1}{B_1}, \overset{t_1+1}{B_1})\wedge^\blacktriangleright (\overset{t_2}{B_2}, \overset{t_2+1}{B_2})\\
                	&= (\alpha (\overset{t_1}{B_1})\wedge^{\vartriangleright} \overset{t_2}{B_2}, \alpha (\overset{t_1}{B_1})\wedge^{\vartriangleright} \overset{t_2+1}{B_2} + (-1)^{t_2}\alpha(\overset{t_1+1}{B_1}) \wedge^{\vartriangleright}\overset{t_2}{B_2})\\
                	&= (\overset{t_1}{B_1}\wedge^{[,]}\overset{t_2}{B_2},  \overset{t_1}{B_1}\wedge^{[,]}\overset{t_2+1}{B_2} + (-1)^{t_2}\overset{t_1+1}{B_1} \wedge^{[,]}\overset{t_2}{B_2})\\
                	&= \mathscr{B}_1\boldsymbol{\wedge}^{\bm{\left[,\right]}}\mathscr{B}_2,
                \end{align*}
                by using (\ref{C2}).
                \item Let $\mathscr{A}$, $\mathscr{B}_1$ and $\mathscr{B}_2$ be as above, then
                \begin{align*}
                \mathscr{A}\boldsymbol{\wedge}^{\blacktriangleright} (\mathscr{B}_1\boldsymbol{\wedge}^{\bm{\left[,\right]}} \mathscr{B}_2)&= (\overset{k}{A}, \overset{k+1}{A})\boldsymbol{\wedge}^{\blacktriangleright}(\overset{t_1}{B_1}\wedge^{[,]}\overset{t_2}{B_2},  \overset{t_1}{B_1}\wedge^{[,]}\overset{t_2+1}{B_2} + (-1)^{t_2}\overset{t_1+1}{B_1} \wedge^{[,]}\overset{t_2}{B_2})\\
                &=( \overset{k}{A}\wedge^{\vartriangleright}(\overset{t_1}{B_1}\wedge^{[,]}\overset{t_2}{B_2}), \overset{k}{A}\wedge^{\vartriangleright}(\overset{t_1}{B_1}\wedge^{[,]}\overset{t_2+1}{B_2} +(-1)^{t_2}\overset{t_1+1}{B_1}\wedge^{[,]}\overset{t_2}{B_2}) \\
                &+ (-1)^{t_1 + t_2} \overset{k+1}{A}\wedge^{\vartriangleright}(\overset{t_1}{B_1}\wedge^{[,]}\overset{t_2}{B_2}))\\
                &= ((\overset{k}{A}\wedge^{\vartriangleright} \overset{t_1}{B_1})\wedge^{[,]}\overset{t_2}{B_2} +(-1)^{k t_1}\overset{t_1}{B_1}\wedge^{[,]}(\overset{k}{A}\wedge^{\vartriangleright}\overset{t_2}{B_2}),  (\overset{k}{A}\wedge^{\vartriangleright} \overset{t_1}{B_1})\wedge^{[,]}\overset{t_2+1}{B_2} \\
                &+(-1)^{k t_1} \overset{t_1}{B_1}\wedge^{[,]}(\overset{k}{A}\wedge^{\vartriangleright}\overset{t_2+1}{B_2}) + (-1)^{t_2}(\overset{k}{A}\wedge^{\vartriangleright}\overset{t_1 +1}{B_1})\wedge^{[,]}\overset{t_2}{B_2} \\
                &  + (-1)^{k(t_1 +1)+ t_2} \overset{t_1+1}{B_1}\wedge^{[,]}(\overset{k}{A}\wedge^{\vartriangleright}\overset{t_2}{B_2})+ (-1)^{t_1 +t_2}(\overset{k+1}{A}\wedge^{\vartriangleright}\overset{t_1}{B_1}) \wedge^{[,]}\overset{t_2}{B_2} \\
                & + (-1)^{t_1(k+1) + t_1 + t_2}\overset{t_1}{B}\wedge^{[,]}(\overset{k+1}{A}\wedge^{\vartriangleright}\overset{t_2}{B_2})   )\\
                &= (\mathscr{A}\boldsymbol{\wedge}^{\blacktriangleright} \mathscr{B}_1)\boldsymbol{\wedge}^{\bm{\left[,\right]}}\mathscr{B}_2+ (-1)^{kt_1}\mathscr{B}_1\boldsymbol{\wedge}^{\bm{\left[,\right]}}(\mathscr{A}\boldsymbol{\wedge}^{\blacktriangleright} \mathscr{B}_2),
                \end{align*} 
         by using (\ref{C3}).
                
                  \item 
                  Let $\mathscr{A}_1 = (\overset{k_1}{A_1}, \overset{k_1+1}{A_1})$ and $\mathscr{A}_2 = (\overset{k_2}{A_2}, \overset{k_2+1}{A_2})$ and $\mathscr{B}$ be as above, then
                \begin{align*}
                (\mathscr{A}_1\boldsymbol{\wedge}^{\bm{\left[,\right]}}\mathscr{A}_2)\boldsymbol{\wedge}^{\blacktriangleright} \mathscr{B}&= (\overset{k_1}{A_1}\wedge^{[,]}\overset{k_2}{A_2}, \overset{k_1}{A_1}\wedge^{[,]}\overset{k_2 +1}{A_2}+ (-1)^{k_2}\overset{k_1+1}{A_1}\wedge^{[,]}\overset{k_2}{A_2})\boldsymbol{\wedge}^{\blacktriangleright}(\overset{t}{B}, \overset{t+1}{B})\\
                &= ( (\overset{k_1}{A_1}\wedge^{[,]}\overset{k_2}{A_2} )\wedge^{\vartriangleright}\overset{t}{B},  (\overset{k_1}{A_1}\wedge^{[,]}\overset{k_2}{A_2} )\wedge^{\vartriangleright}\overset{t+1}{B}+(-1)^t (\overset{k_1}{A_1}\wedge^{[,]}\overset{k_2+1}{A_2} )\wedge^{\vartriangleright}\overset{t}{B} \\
                &+ (-1)^{t+k_2}(\overset{k_1+1}{A_1}\wedge^{[,]}\overset{k_2}{A_2} )\wedge^{\vartriangleright}\overset{t}{B})\\
                &=(\overset{k_1}{A_1}\wedge^{\vartriangleright}(\overset{k_2}{A_2} \wedge^{\vartriangleright}\overset{t}{B}) + (-1)^{k_1 k_2 +1} \overset{k_2}{A_2}\wedge^{\vartriangleright}(\overset{k_1}{A_1}\wedge^{\vartriangleright}\overset{t}{B}),
                \overset{k_1}{A_1}\wedge^{\vartriangleright}(\overset{k_2}{A_2} \wedge^{\vartriangleright}\overset{t+1}{B}) \\
                &+ (-1)^{k_1 k_2 +1} \overset{k_2}{A_2}\wedge^{\vartriangleright}(\overset{k_1}{A_1}\wedge^{\vartriangleright}\overset{t+1}{B})+ (-1)^t (\overset{k_1}{A_1}\wedge^{\vartriangleright}(\overset{k_2 +1}{A_2} \wedge^{\vartriangleright}\overset{t}{B}) \\
                &+ (-1)^{k_1 (k_2+1) +1} \overset{k_2+1}{A_2}\wedge^{\vartriangleright}(\overset{k_1}{A_1}\wedge^{\vartriangleright}\overset{t}{B})) + (-1)^{t + k_2}(\overset{k_1+1}{A_1}\wedge^{\vartriangleright}(\overset{k_2}{A_2} \wedge^{\vartriangleright}\overset{t}{B}) \\
                & + (-1)^{(k_1+1 )k_2 +1} \overset{k_2}{A_2}\wedge^{\vartriangleright}(\overset{k_1+1}{A_1}\wedge^{\vartriangleright}\overset{t}{B})))\\
                &= \mathscr{A}_1 \boldsymbol{\wedge}^{\blacktriangleright}(\mathscr{A}_2\boldsymbol{\wedge}^{\blacktriangleright} \mathscr{B})+(-1)^{k_1 k_2 +1} \mathscr{A}_2\boldsymbol{\wedge}^{\blacktriangleright}(\mathscr{A}_1\boldsymbol{\wedge}^{\blacktriangleright} \mathscr{B}),
                \end{align*}  
                by using (\ref{C4}).

            \end{enumerate}
            \end{proof}
        
\begin{proposition}
		\begin{enumerate}[1)]
		\setlength{\itemsep}{6pt}
		\item 
	For $\mathscr{A} \in \Lambda^k(M, \mathcal{g}, 1)$, $\mathscr{A}' \in \Lambda^{k'}(M, \mathcal{g}, 1)$ and $\mathscr{C} \in \Lambda^q(M, \mathcal{l}, 1)$, then
	\begin{align}
		& \bm{\beta}(\mathscr{A} \boldsymbol{\wedge}^{\blacktriangleright}\mathscr{C}) = \mathscr{A} \boldsymbol{\wedge}^{\blacktriangleright}\bm{\beta}(\mathscr{C});\\
		& \mathscr{A} \boldsymbol{\wedge}^{\blacktriangleright}\mathscr{A}' = \mathscr{A}\boldsymbol{\wedge} \mathscr{A}' + (-1)^{k k' +1}\mathscr{A}' \boldsymbol{\wedge} \mathscr{A}.
	\end{align}
	    \item
	For $\mathscr{A} \in \Lambda^k(M, \mathcal{g}, 1)$, $\mathscr{B}_1 \in \Lambda^{t_1} (M, \mathcal{h}, 1)$, $ \mathscr{B}_2 \in \Lambda^{t_2} (M, \mathcal{h}, 1)$ and $\mathscr{W}\in \Lambda^q(M, \mathcal{w}, 1)$ $(\mathcal{w}=\mathcal{g}, \mathcal{h}, \mathcal{l})$, then
	\begin{align}
		& \bm{d}(\mathscr{A}\boldsymbol{\wedge}^{\blacktriangleright}\mathscr{W}) =d \mathscr{A} \boldsymbol{\wedge}^{\blacktriangleright} \mathscr{W} +(-1)^k \mathscr{A}\boldsymbol{\wedge}^{\blacktriangleright}d \mathscr{W};\\
		& \bm{d} (\mathscr{B}_1 \boldsymbol{\wedge}^{\bm{\{,\}}}\mathscr{B}_2) =  d \mathscr{B}_1 \boldsymbol{\wedge}^{\bm{\{,\}}}\mathscr{B}_2 + (-1)^{t_1}\mathscr{B}_1 \boldsymbol{\wedge}^{\bm{\{,\}}}d \mathscr{B}_2;\\
		& \mathscr{A}\boldsymbol{\wedge}^{\blacktriangleright}(\mathscr{B}_1 \boldsymbol{\wedge}^{\bm{\{,\}}}\mathscr{B}_2)= (\mathscr{A}\boldsymbol{\wedge}^{\blacktriangleright} \mathscr{B}_1 )\boldsymbol{\wedge}^{\bm{\{,\}}}\mathscr{B}_2  +(-1)^{k t_1}\mathscr{B}_1\boldsymbol{\wedge}^{\bm{\{,\}}}(\mathscr{A}\boldsymbol{\wedge}^{\blacktriangleright} \mathscr{B}_2).
	\end{align}
	  \begin{proof}
		\begin{enumerate}[1)]
			\setlength{\itemsep}{6pt}
			\item 	
			Let $\mathscr{A} = (\overset{k}{A}, \overset{k+1}{A})$, $\mathscr{A'} = (\overset{k'}{A}, \overset{k'+1}{A})$ and $\mathscr{C} = (\overset{q}{C}, \overset{q+1}{C})$, then
			\begin{align*}
			\bm{\beta}(\mathscr{A} \boldsymbol{\wedge}^{\blacktriangleright}\mathscr{C})&=	\bm{\beta} (\overset{k}{A} \wedge^{\vartriangleright} \overset{q}{C}, \overset{k}{A}\wedge^{\vartriangleright} \overset{q+1}{C} + (-1)^q \overset{k+1}{A}\wedge^{\vartriangleright} \overset{q}{C})\\
			&= (\beta(\overset{k}{A} \wedge^{\vartriangleright} \overset{q}{C}), \beta(\overset{k}{A}\wedge^{\vartriangleright} \overset{q+1}{C}) + (-1)^q\beta (\overset{k+1}{A}\wedge^{\vartriangleright} \overset{q}{C} )\\
			&=(\overset{k}{A} \wedge^{\vartriangleright} \beta(\overset{q}{C}), \overset{k}{A}\wedge^{\vartriangleright} \beta(\overset{q+1}{C}) + (-1)^q\overset{k+1}{A}\wedge^{\vartriangleright} \beta (\overset{q}{C} ))\\
			&=\mathscr{A} \boldsymbol{\wedge}^{\blacktriangleright}\bm{\beta}(\mathscr{C}),
				\end{align*}
			by using (\ref{C5}), and
			\begin{align*}
			\mathscr{A} \boldsymbol{\wedge}^{\blacktriangleright}\mathscr{A}' &= (\overset{k}{A}\wedge^{\vartriangleright}\overset{k'}{A}, \overset{k}{A}\wedge^{\vartriangleright}\overset{k'+1}{A}+ (-1)^{k'}\overset{k+1}{A}\wedge^{\vartriangleright}\overset{k'}{A})\\
			&=( \overset{k}{A}\wedge \overset{k'}{A} + (-1)^{k k' +1}  \overset{k'}{A}\wedge \overset{k}{A},\overset{k}{A}\wedge \overset{k'+1}{A}+ (-1)^{k (k' +1)+1}  \overset{k'+1}{A}\wedge \overset{k}{A}\\
			&+ (-1)^{k'}\overset{k+1}{A}\wedge \overset{k'}{A} +(-1)^{k'(k +1)+k'+1} \overset{k'}{A}\wedge \overset{k+1}{A})\\
			&= \mathscr{A}\boldsymbol{\wedge} \mathscr{A}' + (-1)^{k k' +1}\mathscr{A}' \boldsymbol{\wedge} \mathscr{A},
		\end{align*}
		by using (\ref{C6}).
		 \item
		 Let $\mathscr{A} $ be as above, and  $\mathscr{B}_1 = (\overset{t_1}{B_1}, \overset{t_1+1}{B_1})$, $\mathscr{B}_2 = (\overset{t_2}{B_2}, \overset{t_2+1}{B_2})$ and $\mathscr{W}=(\overset{q}{W},\overset{q+1}{W})$, then
		 \begin{align*}
		 \bm{d} (\mathscr{A}\boldsymbol{\wedge}^{\blacktriangleright}\mathscr{W}) &= \bm{d}(\overset{k}{A} \wedge^{\vartriangleright}\overset{q}{W}, \overset{k}{A} \wedge^{\vartriangleright}\overset{q+1}{W}+ (-1)^q \overset{k+1}{A} \wedge^{\vartriangleright}\overset{q}{W})\\
		 &=(d (\overset{k}{A} \wedge^{\vartriangleright}\overset{q}{W})+ (-1)^{k+q+1}k \overset{k}{A} \wedge^{\vartriangleright}\overset{q+1}{W} + (-1)^{k+1}k \overset{k+1}{A} \wedge^{\vartriangleright}\overset{q}{W}, d(\overset{k}{A} \wedge^{\vartriangleright}\overset{q+1}{W}) \\
		 & + (-1)^q d (\overset{k+1}{A} \wedge^{\vartriangleright}\overset{q}{W}))\\
		 &= (d \overset{k}{A} \wedge^{\vartriangleright}\overset{q}{W} +(-1)^k \overset{k}{A} \wedge^{\vartriangleright}d \overset{q}{W} +(-1)^{k+q+1}k \overset{k}{A} \wedge^{\vartriangleright}\overset{q+1}{W}+(-1)^{k+1}k \overset{k+1}{A} \wedge^{\vartriangleright}\overset{q}{W}, \\
		 & d\overset{k}{A} \wedge^{\vartriangleright}\overset{q+1}{W} +(-1)^k \overset{k}{A} \wedge^{\vartriangleright}d \overset{q+1}{W} +(-1)^q d \overset{k+1}{A} \wedge^{\vartriangleright}\overset{q}{W} +(-1)^{q + k +1}\overset{k+1}{A} \wedge^{\vartriangleright}d \overset{q}{W})\\
		 &=\bm{d} \mathscr{A} \boldsymbol{\wedge}^{\blacktriangleright} \mathscr{W} +(-1)^k \mathscr{A}\boldsymbol{\wedge}^{\blacktriangleright}\bm{d} \mathscr{W},
		 \end{align*}
		by using (\ref{C7}).
			\begin{align*}
		\bm{d} (\mathscr{B}_1 \boldsymbol{\wedge}^{\bm{\{,\}}}\mathscr{B}_2) &=  \bm{d} (\overset{t_1}{B_1}\wedge^{\{,\}}\overset{t_2}{B_2}, \overset{t_1}{B_1}\wedge^{\{,\}}\overset{t_2+1}{B_2}+ (-1)^{t_2}\overset{t_1+1}{B_1}\wedge^{\{,\}}\overset{t_2}{B_2})\\
		&= (d( \overset{t_1}{B_1}\wedge^{\{,\}}\overset{t_2}{B_2}) + (-1)^{t_1 +t_2 +1} k \overset{t_1}{B_1}\wedge^{\{,\}}\overset{t_2+1}{B_2} + (-1)^{t_1 +1}k \overset{t_1+1}{B_1}\wedge^{\{,\}}\overset{t_2}{B_2}, \\
		&d (\overset{t_1}{B_1}\wedge^{\{,\}}\overset{t_2+1}{B_2})+ (-1)^{t_2}d (\overset{t_1+1}{B_1}\wedge^{\{,\}}\overset{t_2}{B_2}))\\
		&=( d \overset{t_1}{B_1}\wedge^{\{,\}}\overset{t_2}{B_2}  + (-1)^{t_1} \overset{t_1}{B_1}\wedge^{\{,\}}d \overset{t_2}{B_2}+ (-1)^{t_1 +t_2 +1} k \overset{t_1}{B_1}\wedge^{\{,\}}\overset{t_2+1}{B_2} \\
		& + (-1)^{t_1 +1}k  \overset{t_1+1}{B_1}\wedge^{\{,\}}\overset{t_2}{B_2},  d \overset{t_1}{B_1}\wedge^{\{,\}}\overset{t_2+1}{B_2}+ (-1)^{t_1} \overset{t_1}{B_1}\wedge^{\{,\}}d \overset{t_2+1}{B_2} \\
		& + (-1)^{t_2}d \overset{t_1+1}{B_1} \wedge^{\{,\}}\overset{t_2}{B_2} +  (-1)^{t_1+ t_2+1} \overset{t_1+1}{B_1}\wedge^{\{,\}}d\overset{t_2}{B_2})\\
		&=\bm{d} \mathscr{B}_1 \boldsymbol{\wedge}^{\bm{\{,\}}}\mathscr{B}_2 + (-1)^{t_1}\mathscr{B}_1 \boldsymbol{\wedge}^{\bm{\{,\}}}\bm{d} \mathscr{B}_2,
			\end{align*}
	by using (\ref{C8}).
	\begin{align*}
		\mathscr{A}\boldsymbol{\wedge}^{\blacktriangleright}(\mathscr{B}_1 \boldsymbol{\wedge}^{\bm{\{,\}}}\mathscr{B}_2)&= (\overset{k}{A}, \overset{k+1}{A})\boldsymbol{\wedge}^{\blacktriangleright} (\overset{t_1}{B_1}\wedge^{\{,\}}\overset{t_2}{B_2}, \overset{t_1}{B_1}\wedge^{\{,\}}\overset{t_2+1}{B_2}+ (-1)^{t_2}\overset{t_1+1}{B_1}\wedge^{\{,\}}\overset{t_2}{B_2})\\
		&= (\overset{k}{A}\wedge^{\vartriangleright}( \overset{t_1}{B_1}\wedge^{\{,\}}\overset{t_2}{B_2}), \overset{k}{A}\wedge^{\vartriangleright} (\overset{t_1}{B_1}\wedge^{\{,\}}\overset{t_2+1}{B_2})+ (-1)^{t_2} \overset{k}{A}\wedge^{\vartriangleright} (\overset{t_1+1}{B_1}\wedge^{\{,\}}\overset{t_2}{B_2})\\
		&+ (-1)^{t_1 +t_2 }\overset{k+1}{A}\wedge^{\vartriangleright}( \overset{t_1}{B_1}\wedge^{\{,\}}\overset{t_2}{B_2}) )\\
		&=( (\overset{k}{A}\wedge^{\vartriangleright} \overset{t_1}{B_1})\wedge^{\{,\}}\overset{t_2}{B_2} +(-1)^{k t_1}\overset{t_1}{B_1}\wedge^{\{,\}}(\overset{k}{A}\wedge^{\vartriangleright}\overset{t_2}{B_2}),  (\overset{k}{A}\wedge^{\vartriangleright} \overset{t_1}{B_1})\wedge^{\{,\}}\overset{t_2+1}{B_2} \\
		&  +(-1)^{k t_1} \overset{t_1}{B_1}\wedge^{\{,\}}(\overset{k}{A}\wedge^{\vartriangleright}\overset{t_2+1}{B_2}) + (-1)^{t_2}(\overset{k}{A}\wedge^{\vartriangleright}\overset{t_1 +1}{B_1})\wedge^{\{,\}}\overset{t_2}{B_2} \\
		& + (-1)^{k(t_1 +1)+ t_2} \overset{t_1+1}{B_1}\wedge^{\{,\}}(\overset{k}{A}\wedge^{\vartriangleright}\overset{t_2}{B_2}) + (-1)^{t_1 +t_2}(\overset{k+1}{A}\wedge^{\vartriangleright}\overset{t_1}{B_1}) \wedge^{\{,\}}\overset{t_2}{B_2} \\
		&+ (-1)^{t_1(k+1) + t_1 + t_2}\overset{t_1}{B}\wedge^{\{,\}}(\overset{k+1}{A}\wedge^{\vartriangleright}\overset{t_2}{B_2}) )\\
		&= (\mathscr{A}\boldsymbol{\wedge}^{\blacktriangleright} \mathscr{B}_1 )\boldsymbol{\wedge}^{\bm{\{,\}}}\mathscr{B}_2  +(-1)^{k t_1}\mathscr{B}_1\boldsymbol{\wedge}^{\bm{\{,\}}}(\mathscr{A}\boldsymbol{\wedge}^{\blacktriangleright} \mathscr{B}_2),
	\end{align*}
	by using (\ref{C9}).
	   \end{enumerate}
\end{proof}
	
	   \end{enumerate}

\end{proposition}

\section{Generalized  $G$-invariant bilinear forms}\label{sec4}
In order to construct generalized higher Yang-Mills theories, the corresponding generalized actions need to be developed. In this paper,  there is a natural expression of the corresponding lagrangian using generalized  $G$-invariant bilinear forms in the framework of GDC.
In this section, we first recall some of the definitions and propositions which
apply to ordinary non-degenerate $G$-invariant bilinear forms in Lie algebras of a Lie ($2$-)crossed module. Then, analogous results, applicable to the type $1$ generalized forms, are
formulated. We shall do those starting with a appropriate definition of a generalized $G$-invariant form.

Symmetric non-degenerate $G$-invariant bilinear forms in  $(\mathcal{h},\mathcal{g}; \alpha,\vartriangleright)$ and $(\mathcal{l},\mathcal{h},\mathcal{g};\beta,\alpha,\vartriangleright,\left\{,\right\})$ are given by a triple of symmetric non-degenerate  bilinear forms $\langle - , - \rangle_\mathcal{g}$ in $\mathcal{g}$, $\langle - , - \rangle_\mathcal{h}$ in $\mathcal{h}$ and $\langle - , - \rangle_\mathcal{l}$ in $\mathcal{l}$, satisfying
\begin{enumerate}
	\item  $\langle - , - \rangle_\mathcal{g}$ is $G$-invariant, i.e.
	$$ \langle gXg^{-1}, gX'g^{-1} \rangle_\mathcal{g}= \langle X,X'\rangle_\mathcal{g}, \ \ \ \ \ \ \ \ \forall g \in G, X, X' \in \mathcal{g};$$
	\item  $\langle - , - \rangle_\mathcal{h}$ is $G$-invariant, i.e.
	$$ \langle g\vartriangleright Y,  g\vartriangleright Y' \rangle_\mathcal{h}= \langle Y,Y'\rangle_\mathcal{h}, \ \ \ \ \ \ \ \forall g\in G,Y,Y'\in \mathcal{h};$$
	\item  $\langle - , - \rangle_\mathcal{l}$ is $G$-invariant, i.e.
	$$ \langle g\vartriangleright Z,  g\vartriangleright Z' \rangle_\mathcal{l}= \langle Z,Z'\rangle_\mathcal{l}, \ \ \ \ \ \ \ \forall g\in G,Z,Z'\in \mathcal{l}.$$
\end{enumerate}
We will pay our attention to the case
when the base group $G$ is compact in the real case, or has a compact real form in
the complex case, so that $G$-invariant non-degenerate bilinear forms
in $\mathcal{g}$, $\mathcal{h}$ and $\mathcal{l}$ will always exist \cite{SDH, Martins:2010ry}.
From these invariance conditions, one immediately sees that 
\begin{align}\label{XXX}
\langle [X,X'], X''\rangle_\mathcal{g}=-\langle X',[X,X'']\rangle_\mathcal{g},
\end{align}
\begin{align}\label{YYY}
\langle [Y,Y'], Y''\rangle_\mathcal{h}=-\langle Y',[Y,Y'']\rangle_\mathcal{h},
\end{align}
\begin{align}\label{ZZZ}
\langle [Z,Z'],Z''\rangle_\mathcal{l}=-\langle Z',[Z,Z'']\rangle_\mathcal{l}.
\end{align}

One can define bilinear anti-symmetric maps $\sigma:\mathcal{h} \times \mathcal{h} \longrightarrow \mathcal{g}$ by the rule:
\begin{align}
	\langle \sigma(Y,Y'), X\rangle_\mathcal{g}=-\langle Y, X\vartriangleright Y'\rangle_\mathcal{h},
\end{align}
for $X \in \mathcal{g}$, $Y,Y' \in \mathcal{h}$, and $\kappa:\mathcal{l} \times \mathcal{l} \longrightarrow \mathcal{g}$ by the rule:
\begin{align}
	\langle \kappa(Z,Z'), X\rangle_\mathcal{g}=-\langle Z, X\vartriangleright Z'\rangle_\mathcal{l},
\end{align}
for $X \in \mathcal{g}$, $Z,Z' \in \mathcal{l}$.
From the anti-symmetric properties of $\sigma$ and $\kappa$, one immediately sees that
\begin{align}\label{YXY'}
	\langle Y, X\vartriangleright Y'\rangle_\mathcal{h}=-\langle Y', X\vartriangleright Y \rangle_\mathcal{h}=-\langle X\vartriangleright Y, Y' \rangle_\mathcal{h},
\end{align}
\begin{align}\label{ZXZ'}
	\langle Z, X\vartriangleright Z'\rangle_\mathcal{l}=-\langle Z', X\vartriangleright Z \rangle_\mathcal{l}=-\langle X\vartriangleright Z, Z' \rangle_\mathcal{l}.
\end{align}

Further, there are two bilinear maps $\eta_1: \mathcal{l} \times \mathcal{h} \longrightarrow \mathcal{h}$ and $\eta_2: \mathcal{l} \times \mathcal{h} \longrightarrow \mathcal{h}$ by the rule:
\begin{align}\label{YY'Z}
	\langle \left\{ Y, Y' \right\}, Z \rangle_\mathcal{l} = -\langle Y', \eta_1(Z, Y) \rangle_\mathcal{h} = -\langle Y, \eta_2(Z, Y') \rangle_\mathcal{h} ,
\end{align}
for each $Y$, $Y' \in \mathcal{h}$, and $Z \in \mathcal{l}$.

Finally, one can define maps $\alpha^*:\mathcal{g} \longrightarrow \mathcal{h} $ defined by the rule:
\begin{align}
	\langle Y, \alpha^*(X)\rangle_\mathcal{h}=\langle \alpha(Y), X\rangle_\mathcal{g}, \ \ \ \ \ \ \ \forall X \in \mathcal{g}, Y \in \mathcal{h},
\end{align}
and $\beta^*:\mathcal{h} \longrightarrow \mathcal{l}$ defined by the rule:
\begin{align}
	\langle Z, \beta^*(Y)\rangle_\mathcal{l}=\langle \beta(Z), Y\rangle_\mathcal{h}, \ \ \ \ \ \ \ \forall Y \in \mathcal{h}, Z \in \mathcal{l}.
\end{align}
See  \cite{Radenkovic:2019qme} for more properties of those maps.
And most of the conventions of reference \cite{SDH} are retained.

To define generalized higher Yang-Mills actions, one has to first define the generalized symmetric non-degenerate $G$-invariant bilinear form. Encouraged by these constructions of generalized differential forms, we develop a generalized form of a Lie algebra $\mathcal{g}$ denoted by
\begin{align}
(X_1, X_2) \in \mathcal{g} \times \mathcal{g},\ \ \ \ \ \ \ X_i \in \mathcal{g}, i=1, 2.
\end{align}
Then we construct the generalized bilinear forms $\ll -, - \gg : ( \mathcal{g} \times \mathcal{g})\times ( \mathcal{g} \times \mathcal{g}) \longrightarrow \mathbb{C}\times  \mathbb{C}$ by 
\begin{align}\label{gbf}
\ll (X_1, X_2), (X_1', X_2') \gg := (\langle X_1, X_1'\rangle, \langle X_1, X_2' \rangle + \langle X_2, X_1' \rangle ),
\end{align}
where $\mathbb{C}$ is a number field.
From  these properties of the symmetric non-degenerate $G$-invariant bilinear form in $\mathcal{g}$, one can apparently note that $\ll -, - \gg$ is also a symmetric non-degenerate $G$-invariant bilinear form.
This structure will lead us to some unexpected and novel results.

Starting with \eqref{gbf}, we define 
\begin{align}
	\ll\mathscr{A}_1 ,\mathscr{A}_2\gg := \ll ( \overset{p}{A_1} ,  \overset{p+1}{A_1}), ( \overset{q}{A_2} ,  \overset{q+1}{A_2}) \gg= (\langle \overset{p}{A_1}, \overset{q}{A_2} \rangle, \langle \overset{p}{A_1}, \overset{q+1}{A_2} \rangle + (-1)^q \langle \overset{p+1}{A_1},  \overset{q}{A_2} \rangle),
\end{align}
where  $\mathscr{A}_1 \in \Lambda^p(M, \mathcal{g}, 1)$ and $\mathscr{A}_2\in \Lambda^q(M, \mathcal{g}, 1)$.
There is an identity
\begin{align}
	\ll\mathscr{A}_1 , \mathscr{A}_2\gg = (-1)^{pq} \ll \mathscr{A}_2, \mathscr{A}_1 \gg.
\end{align}
In a similar way, we can define the generalized forms of Lie algebra $\mathcal{h}$ and $\mathcal{l}$.

Given generalized $G$-invariant symmetric non-degenerate bilinear forms in the generalized forms of Lie algebras $\mathcal{g}$, $\mathcal{h}$ and $\mathcal{l}$, one can define maps $\bm{\overline{\sigma}}: \Lambda^{t_1} (M, \mathcal{h}, 1) \times \Lambda^{t_2} (M, \mathcal{h}, 1)\longrightarrow  \Lambda^{t_1+t_2}(M, \mathcal{g}, 1)$ by the rule
\begin{align}
	\bm{\overline{\sigma}}(\mathscr{B}_1, \mathscr{B}_2):= (\overline{\sigma} (\overset{t_1}{B_1}, \overset{t_2}{B_2}), \overline{\sigma} (\overset{t_1}{B_1}, \overset{t_2+1}{B_2})+ (-1)^{t_2}\overline{\sigma} (\overset{t_1+1}{B_1}, \overset{t_2}{B_2})),
\end{align}
and $\bm{\overline{\kappa}}:\Lambda^{q_1} (M, \mathcal{l}, 1) \times \Lambda^{q_2} (M, \mathcal{l}, 1)\longrightarrow   \Lambda^{q_1+q_2}(M, \mathcal{g}, 1)$ by the rule
\begin{align}
	\bm{\overline{\kappa}}(\mathscr{C}_1, \mathscr{C}_2):= (\overline{\kappa}(\overset{q_1}{C_1}, \overset{q_2}{C_2}),\overline{\kappa} (\overset{q_1}{C_1}, \overset{q_2+1}{C_2})+ (-1)^{q_2}\overline{\kappa} (\overset{q_1+1}{C_1}, \overset{q_2}{C_2})),
\end{align}
where $ \mathscr{A} = (\overset{k}{A}, \overset{k+1}{A})$, $\mathscr{B}_i = (\overset{t_i}{B_i}, \overset{t_i+1}{B_i})$ and $\mathscr{C}_i= (\overset{q_i}{C_i}, \overset{q_i+1}{C_i})(i = 1, 2)$, and the map $\overline{\sigma} $ is a bilinear map $ \Lambda^{t_1}(M, \mathcal{h}) \times \Lambda^{t_2}(M, \mathcal{h}) \longrightarrow \Lambda^{t_1+t_2}(M, \mathcal{g})$ 
and $\overline{\kappa}$  is a bilinear map $\Lambda^{q_1}(M, \mathcal{l}) \times \Lambda^{q_2}(M, \mathcal{l}) \longrightarrow \Lambda^{q_1+q_2}(M, \mathcal{g})$, seeing Appendix \ref{LAVDF} for more details.

In addition, one define maps $\bm{\overline{\eta_i}}: \Lambda^{q} (M, \mathcal{l}, 1) \times \Lambda^{t} (M, \mathcal{h}, 1)\longrightarrow   \Lambda^{q+t}(M, \mathcal{h}, 1)$, $i = 1, 2$, by the rule
\begin{align}
	\bm{\overline{\eta_i}}(\mathscr{C}, \mathscr{B}):= (\overline{\eta_i} (\overset{q}{C}, \overset{t}{B}), \overline{\eta_i}(\overset{q}{C}, \overset{t+1}{B})+ (-1)^{t}\overline{\eta_i} (\overset{q+1}{C}, \overset{t}{B})),
\end{align}
where $\mathscr{C}= (\overset{q}{C}, \overset{q+1}{C})$ , $\mathscr{B}= (\overset{t}{B}, \overset{t+1}{B})$ and the maps $\eta_i$, $i = 1, 2$, are bilinear maps $\Lambda^q (M,\mathcal{l}) \times \Lambda^t(M, \mathcal{h}) \longrightarrow \Lambda^{q+t}(M,\mathcal{h})$ in Appendix \ref{LAVDF}. 

Finally, one define maps $\bm{\alpha^{\ast}}: \Lambda^{k} (M, \mathcal{g}, 1) \longrightarrow   \Lambda^{k}(M, \mathcal{h}, 1)$ by 
\begin{align}
	\bm{\alpha^{\ast}}(\mathscr{A}):=(\alpha^{\ast}(\overset{k}{A}), \alpha^{\ast}(\overset{k+1}{A})),
\end{align}
and $\bm{\beta^{\ast}}: \Lambda^{t} (M, \mathcal{h}, 1) \longrightarrow  \Lambda^{t}(M, \mathcal{l}, 1)$ by 
\begin{align}
	\bm{\beta^{\ast}}(\mathscr{B}):=(\beta^{\ast}(\overset{t}{B}), \beta^{\ast}(\overset{t+1}{B})),
\end{align}
where $\mathscr{A}$ and $\mathscr{B}$ are as above, and see  Appendix \ref{LAVDF}  for $\alpha^{\ast}$ and $\beta^{\ast}$.

The ordinary non-degenerate $G$-invariant forms induce the  properties of the generalized $G$-invariant forms embodied in the following propositions.
\begin{proposition}
For $\mathscr{W}_1  \in \Lambda^{q_1} (M, \mathcal{w}, 1)$, $ \mathscr{W}_2 \in \Lambda^{q_2} (M, \mathcal{w}, 1)$ and $\mathscr{W}_3 \in \Lambda^{q_3}(M, \mathcal{w}, 1)$ $(\mathcal{w}=\mathcal{g}, \mathcal{h}, \mathcal{l})$, then
\begin{align}
\ll \mathscr{W}_1 \boldsymbol{\wedge}^{\bm{[,]}}\mathscr{W}_2, \mathscr{W}_3 \gg = (-1)^{q_1 q_2 +1} \ll \mathscr{W}_2, \mathscr{W}_1 \boldsymbol{\wedge}^{\bm{[,]}}\mathscr{W}_3 \gg.
\end{align}
\end{proposition}
  \begin{proof}
  	Let $\mathscr{W}_1 =( \overset{q_1}{W_1}, \overset{q_1 +1}{W_1})$, $ \mathscr{W}_2= (\overset{q_2}{W_2}, \overset{q_2 +1}{W_2}) $ and $\mathscr{W}_3=( \overset{q_3}{W_3}, \overset{q_3 +1}{W_3}) $, then
  	\begin{align*}
	\ll \mathscr{W}_1 \boldsymbol{\wedge}^{\bm{[,]}}\mathscr{W}_2, \mathscr{W}_3 \gg &=\ll (\overset{q_1}{W_1}\wedge^{[,]}\overset{q_2}{W_2}, \overset{q_1}{W_1}\wedge^{[,]}\overset{q_2 +1}{W_2} + (-1)^{q_2} \overset{q_1 +1}{W_1}\wedge^{[,]}\overset{q_2}{W_2}),( \overset{q_3}{W_3}, \overset{q_3 +1}{W_3})\gg\\
	&= ( \langle \overset{q_1}{W_1}\wedge^{[,]}\overset{q_2}{W_2}, \overset{q_3}{W_3}\rangle, \langle \overset{q_1}{W_1}\wedge^{[,]}\overset{q_2}{W_2}, \overset{q_3 +1}{W_3}\rangle + (-1)^{q_3}\langle \overset{q_1}{W_1}\wedge^{[,]}\overset{q_2 +1}{W_2}, \overset{q_3}{W_3}\rangle \\
	&+ (-1)^{q_2 + q_3} \langle \overset{q_1 +1}{W_1}\wedge^{[,]}\overset{q_2}{W_2}, \overset{q_3}{W_3}\rangle )\\
	&= ( (-1)^{q_1 q_2 +1}\langle \overset{q_2}{W_2}, \overset{q_1}{W_1}\wedge^{[,]}\overset{q_3}{W_3}\rangle, (-1)^{q_1 q_2 +1}\langle \overset{q_2}{W_2}, \overset{q_1}{W_1}\wedge^{[,]} \overset{q_3 +1}{W_1}\rangle \\
	&+ (-1)^{q_3 + q_1 (q_2 +1)+1} \langle \overset{q_2 +1}{W_2}, \overset{q_1}{W_1}\wedge^{[,]}\overset{q_3}{W_3}\rangle + (-1)^{ q_3 + q_1q_2 +1} \langle \overset{q_2}{W_2},\overset{q_1 +1}{W_1}\wedge^{[,]} \overset{q_3}{W_3}\rangle ) \\
	&= (-1)^{q_1 q_2 +1} \ll \mathscr{W}_2, \mathscr{W}_1 \boldsymbol{\wedge}^{\bm{[,]}}\mathscr{W}_3 \gg,
  	\end{align*}
  by using (\ref{C10}).
  \end{proof}

\begin{proposition}
For $\mathscr{A} \in \Lambda^k(M, \mathcal{g}, 1)$, $\mathscr{W}_1  \in \Lambda^{q_1} (M, \mathcal{w}, 1)$ and $ \mathscr{W}_2 \in \Lambda^{q_2} (M, \mathcal{w}, 1)$, we have
\begin{align}
	\ll \mathscr{W}_1, \mathscr{A} \boldsymbol{\wedge}^{\blacktriangleright}\mathscr{W}_2\gg = (-1)^{q_2 (k + q_1) + k q_1 + 1}\ll \mathscr{W}_2, \mathscr{A}\boldsymbol{\wedge}^{\blacktriangleright}\mathscr{W}_1\gg =(-1)^{k q_1 +1}\ll \mathscr{A} \boldsymbol{\wedge}^{\blacktriangleright} \mathscr{W}_1, \mathscr{W}_2 \gg 
\end{align}
\end{proposition}
\begin{proof}
   Let $ \mathscr{A}=(\overset{k}{A}, \overset{k+1}{A})$, $\mathscr{W}_1 =( \overset{q_1}{W_1}, \overset{q_1 +1}{W_1})$ and $ \mathscr{W}_2= (\overset{q_2}{W_2}, \overset{q_2 +1}{W_2}) $, then
   \begin{align*}
 	\ll \mathscr{W}_1, \mathscr{A} \boldsymbol{\wedge}^{\blacktriangleright}\mathscr{W}_2\gg &= \ll ( \overset{q_1}{W_1}, \overset{q_1 +1}{W_1}), (\overset{k}{A}\wedge^{\vartriangleright}\overset{q_2}{W_2}, \overset{k}{A}\wedge^{\vartriangleright}\overset{q_2+1}{W_2} +(-1)^{q_2}\overset{k+1}{A}\wedge^{\vartriangleright}\overset{q_2}{W_2} )\gg\\
 	&= (\langle \overset{q_1}{W_1}, \overset{k}{A}\wedge^{\vartriangleright}\overset{q_2}{W_2}\rangle, \langle \overset{q_1}{W_1} , \overset{k}{A}\wedge^{\vartriangleright}\overset{q_2+1}{W_2}+ (-1)^{q_2}\overset{k+1}{A}\wedge^{\vartriangleright} \overset{q_2}{W_2}\rangle \\
 	& + (-1)^{k +q_2}\langle \overset{q_1+1}{W_1}, \overset{k}{A}\wedge^{\vartriangleright}\overset{q_2}{W_2}\rangle )\\
 	&=((-1)^{q_2(k + q_1)+ k q_1 +1}\langle \overset{q_2}{W_2}, \overset{k}{A}\wedge^{\vartriangleright}\overset{q_1}{W_1}\rangle, (-1)^{(q_2 +1)(k + q_1)+ k q_1 +1}\langle \overset{q_2+1}{W_2}, \overset{k}{A}\wedge^{\vartriangleright}\overset{q_1}{W_1}\rangle \\
 	&+ (-1)^{q_2(k + q_1)+ (k+1)q_1 +1 }\langle \overset{q_2}{W_2}, \overset{k+1}{A}\wedge^{\vartriangleright}\overset{q_1}{W_1}\rangle  + (-1)^{q_2 (k+q_1 )+ kq_1+1} \langle \overset{q_2}{W_2}, \overset{k}{A}\wedge^{\vartriangleright}\overset{q_1 +1}{W_1}\rangle )\\
 	&= (-1)^{q_2 (k + q_1) + k q_1 + 1}\ll \mathscr{W}_2, \mathscr{A}\boldsymbol{\wedge}^{\blacktriangleright}\mathscr{W}_1\gg,\\
 		\ll \mathscr{W}_1, \mathscr{A} \boldsymbol{\wedge}^{\blacktriangleright}\mathscr{W}_2\gg &= \ll ( \overset{q_1}{W_1}, \overset{q_1 +1}{W_1}), (\overset{k}{A}\wedge^{\vartriangleright}\overset{q_2}{W_2}, \overset{k}{A}\wedge^{\vartriangleright}\overset{q_2+1}{W_2} +(-1)^{q_2}\overset{k+1}{A}\wedge^{\vartriangleright}\overset{q_2}{W_2} )\gg\\
 	&= (\langle \overset{q_1}{W_1}, \overset{k}{A}\wedge^{\vartriangleright}\overset{q_2}{W_2}\rangle, \langle \overset{q_1}{W_1} , \overset{k}{A}\wedge^{\vartriangleright}\overset{q_2+1}{W_2}+ (-1)^{q_2}\overset{k+1}{A}\wedge^{\vartriangleright} \overset{q_2}{W_2}\rangle \\
 	&+ (-1)^{k +q_2}\langle \overset{q_1+1}{W_1}, \overset{k}{A}\wedge^{\vartriangleright}\overset{q_2}{W_2}\rangle )\\
 	&= ((-1)^{k q_1 +1} \langle \overset{k}{A}\wedge^{\vartriangleright}\overset{q_1}{W_1}, \overset{q_2}{W_2}\rangle, (-1)^{k q_1 +1}\langle  \overset{k}{A}\wedge^{\vartriangleright}\overset{q_1}{W_1}, \overset{q_2 +1}{W_2}\rangle  \\
 	&+(-1)^{q_2 + (k+1)q_1 +1}\langle \overset{k+1}{A}\wedge^{\vartriangleright}\overset{q_1}{W_1},  \overset{q_2}{W_2}\rangle  +  (-1)^{k+q_2 + k(q_1 +1)+1}\langle \overset{k}{A}\wedge^{\vartriangleright}\overset{q_1 +1}{W_1}, \overset{q_2}{W_2}\rangle)\\
 	&= (-1)^{k q_1 +1}\ll \mathscr{A} \boldsymbol{\wedge}^{\blacktriangleright} \mathscr{W}_1, \mathscr{W}_2 \gg,
   \end{align*}
by using (\ref{C11}).
\end{proof}

\begin{proposition}
	Given $G$-invariant symmetric non-degenerate bilinear forms in $\mathcal{g}$, $\mathcal{h}$ and $\mathcal{l}$, the following identities hold
	\begin{align}
		&\ll \bm{\overline{\sigma}}(\mathscr{B}_1, \mathscr{B}_2), \mathscr{A}\gg=(-1)^{k t_2 +1}\ll \mathscr{B}_1, \mathscr{A}\boldsymbol{\wedge}^{\blacktriangleright}\mathscr{B}_2\gg;\\
		&\ll \bm{\mathscr{A}, \overline{\sigma}}(\mathscr{B}_1, \mathscr{B}_2)\gg=(-1)^{t_1 t_2 +1}\ll  \mathscr{A}\boldsymbol{\wedge}^{\blacktriangleright}\mathscr{B}_2, \mathscr{B}_1 \gg;\\
		&\ll \bm{\overline{\kappa}}(\mathscr{C}_1, \mathscr{C}_2), \mathscr{A}\gg=(-1)^{k q_2 +1}\ll \mathscr{C}_1, \mathscr{A}\boldsymbol{\wedge}^{\blacktriangleright}\mathscr{C}_2\gg;\\
	&\ll \bm{\mathscr{A}, \overline{\kappa}}(\mathscr{C}_1, \mathscr{C}_2)\gg=(-1)^{q_1 q_2 +1}\ll  \mathscr{A}\boldsymbol{\wedge}^{\blacktriangleright}\mathscr{C}_2, \mathscr{C}_1 \gg;\\
&\ll \mathscr{B}_1 \boldsymbol{\wedge}^{\bm{\{,\}}}\mathscr{B}_2, \mathscr{C} \gg = (-1)^{t_1 (t_2 +q) +1} \ll \mathscr{B}_2, \bm{\overline{\eta_1}}(\mathscr{C}, \mathscr{B}_1) \gg;\\
&\ll \mathscr{B}_1 \boldsymbol{\wedge}^{\bm{\{,\}}}\mathscr{B}_2, \mathscr{C} \gg = (-1)^{t_2 q +1}\ll \mathscr{B}_1, \bm{\overline{\eta_2}}(\mathscr{C}, \mathscr{B}_2)\gg;\\
	&\ll \mathscr{B}, \bm{\alpha^{\ast}}(\mathscr{A})\gg = \ll\bm{\alpha }(\mathscr{B}), \mathscr{A} \gg;\\
&\ll \mathscr{C}, \bm{\beta^{\ast}}(\mathscr{B}) \gg = \ll\bm{\beta}(\mathscr{C}), \mathscr{B}\gg,
	\end{align}
where $\mathscr{A} \in \Lambda^{k} (M, \mathcal{g}, 1)$,  $\mathscr{B} \in \Lambda^{t} (M, \mathcal{h}, 1)$, $\mathscr{C}\in \Lambda^q(M, \mathcal{l}, 1)$, $\mathscr{B}_i \in \Lambda^{t_i} (M, \mathcal{h}, 1)$ and $\mathscr{C}_i \in \Lambda^{q_i} (M, \mathcal{l}, 1)(i = 1, 2)$.
\end{proposition}
\begin{proof}
These identities can be easily given by using (\ref{C12}, \ref{C13}, \ref{C14},  \ref{C15}, \ref{C16}, \ref{C17}, \ref{C18}).
\end{proof}

	   \section{Generalized (higher) connections}\label{sec5}
	   In this section, we define generalized higher connections as extensions of ordinary higher connections in the frame of GDC. In related works  \cite{GHY, ghy, Robbinson5}, the generalized connection theory has been defined by using type $1$ generalized differential forms. The generalized connection forms of type $N=2$ can be found in \cite{Robbinson4, Robinson2}. In addition, Robinson has defined the type $N$ generalized connections in \cite{Robinson1, Robinson3}.
	   We will restrict ourselves only to the generalized higher connection forms of type $N=1$. There may be some similar results for type $N$ generalized higher connections, which may be more complex.
	   
	    Let us consider a principle bundle P($M$, $G$), where $M$ is the $n$-dimensional base manifold and $G$ is the structure group whose Lie algebra is denoted by $\mathcal{g}$.  A generalized connection $1$-form $\mathscr{A}$ valued in the Lie algebra $\mathcal{g}$ is defined as 
	    \begin{align}
	    	\mathscr{A} = (\overset{1}{A}, \overset{2}{A}),
	    \end{align}
	   where $\overset{1}{A}$ is an ordinary $\mathcal{g}$-valued connection $1$-form and $\overset{2}{A}$ is an ordinary $\mathcal{g}$-valued $2$-form. Then the generalized curvature $2$-form $	\mathscr{F}$ is
	   \begin{align}
	   	\mathscr{F}& = \textbf{d} \mathscr{A} + \mathscr{A} \boldsymbol{\wedge}\mathscr{A} \notag\\
	   	                   & =(d \overset{1}{A} + \overset{1}{A} \wedge \overset{1}{A}  + k \overset{2}{A} , d \overset{1}{A}  + \overset{1}{A}  \wedge \overset{2}{A} - \overset{2}{A} \wedge \overset{1}{A}) \notag\\
	   	                   & = (F + k \overset{2}{A} , D \overset{2}{A}),
   		   \end{align}
	   where $F = d \overset{1}{A} + \overset{1}{A} \wedge \overset{1}{A}$ is  the ordinary curvature $2$-form, and $D$ is the covariant derivative with respect to the connection $\overset{1}{A}$. It is easy to show that the generalized curvature $2$-form $\mathscr{F}$ satisfies the generalized Bianchi Identity
	   \begin{align}\label{BI1}
	   	\mathscr{D} \mathscr{F} := \textbf{d} \mathscr{F} + \mathscr{A} \boldsymbol{\wedge}\mathscr{F} - \mathscr{F} \boldsymbol{\wedge} \mathscr{A} \equiv 0,
	   \end{align}
   where $\mathscr{D}$ is the generalized covariant derivative with respect to the generalized connection $\mathscr{A}$.
   
   In \cite{Robbinson5}, a generalized gauge transformation is introduced for generalized connections. Let $\bm{G}=\{(g, mg)|g \in G, m \in \Lambda^1(M, \mathcal{g})\}$ which is a group under group and exterior multiplication. Calculation shows that the product of $\bm{g_1} = (g_1, m_1g_1)$ and $\bm{g_2} = (g_2, m_2g_2) \in \bm{G}$ is the element of $\bm{G}$ given by
   \begin{align}
   \bm{g_1}\bm{g_2}=(g_1, m_1g_1)(g_2, m_2g_2)=(g_1 g_2, g_1 m_2 g_2 + m_1 g_1 g_2).
   \end{align}
 The identity of $\bm{G}$ is $(1, 0)$, where $1$ is the identity of $G$, and the inverse  of $\bm{g}$ is given by $(g^{-1}, -g^{-1}m)$.
 Consider generalized gauge transformation 
 \begin{align}
\mathscr{A'}&= \bm{g^{-1}}\mathscr{A} \bm{g} + \bm{g^{-1}}\textbf{d}\bm{g}\notag \\
&=(g^{-1}(\overset{1}{A} - k m)g + g^{-1}dg, g^{-1}(dm - k m m + \overset{1}{A}m + m\overset{1}{A} + \overset{2}{A})g),
 \end{align}
 and the generalized curvature changes as 
 \begin{align}
 	\mathscr{F'}= \bm{g^{-1}}\mathscr{F}\bm{g}= (g^{-1}(F + k \overset{2}{A})g, g^{-1}D \overset{2}{A} g +g^{-1} (F + k \overset{2}{A}) \wedge^{[, ]}mg).
 \end{align}
   We note that the generalized gauge transformation is same as the ordinary gauge transformation in the expression. 
	
	\subsection{Generalized 2-connections}
	Assume we are given a Lie crossed module $(G, H; \alpha, \vartriangleright)$ (namely $2$-group) with the corresponding differential crossed module $(\mathcal{g}, \mathcal{h}; \alpha, \vartriangleright)$, we consider an ordinary 2-connection $(\overset{1}{A}, \overset{2}{B})$ on a principle $2$-bundle over $M$
\begin{align}
 \overset{1}{A}\in \Lambda^1(M, \mathcal{g}),\ \ \ \overset{2}{B} \in \Lambda^2(M, \mathcal{h}),
\end{align}
and the corresponding 2-curvature $(\overset{2}{\mathbb{F}}, \overset{3}{\mathbb{G}})$ is given by
\begin{align}
 \overset{2}{\mathbb{F}} = d \overset{1}{A} + \overset{1}{A}\wedge \overset{1}{A} - \alpha(\overset{2}{B}),
 \ \ \ \overset{3}{\mathbb{G}} = d \overset{2}{B} + \overset{1}{A}\wedge^{\vartriangleright} \overset{2}{B}.
\end{align}
The 2-Bianchi-Identities are as follows
\begin{align}
&d \overset{2}{\mathbb{F}} + \overset{1}{A} \wedge^{[,]} \overset{2}{\mathbb{F}} = - \alpha(\overset{3}{\mathbb{G}}),\label{2BI}\\
&d \overset{3}{\mathbb{G}} + \overset{1}{A} \wedge^{\vartriangleright} \overset{3}{\mathbb{G}} = (\overset{2}{\mathbb{F}} + \alpha (\overset{2}{B}))\wedge^{\vartriangleright }\overset{2}{B}\label{22BI}.
\end{align}
	We consider gauge transformations given in \cite{Martins:2010ry}:
	\begin{itemize}
		\item [1)] Thin: 
		\begin{align}\label{thin}
	 \overset{1}{A'}= g^{-1}  \overset{1}{A} g + g^{-1} dg,\ \ \ \ \ \overset{2}{B'} = g^{-1} \vartriangleright \overset{2}{B},
		\end{align}
		then the curvatures change as 
		$$\overset{2}{\mathbb{F'}}=g^{-1}\overset{2}{\mathbb{F}}g,\ \ \ \ \  \overset{3}{\mathbb{G'}}=g^{-1}\vartriangleright \overset{3}{\mathbb{G}},$$
		where $g \in G$.
			\item [2)] Fat: 
			\begin{align}\label{fat}
		 \overset{1}{A'} =  \overset{1}{A} + \alpha(\overset{1}{\phi}),\ \ \ \ \ \overset{2}{B'} = \overset{2}{B} + d \overset{1}{\phi} +  \overset{1}{A}  \wedge^{\vartriangleright} \overset{1}{\phi}+ \overset{1}{\phi} \wedge \overset{1}{\phi},
			\end{align}
			then the curvatures change as 
			$$\overset{2}{\mathbb{F'}}= \overset{2}{\mathbb{F}},\ \ \ \ \  \overset{3}{\mathbb{G'}}=\overset{3}{\mathbb{G}}+ \overset{2}{\mathbb{F}}\wedge^{\vartriangleright}\overset{1}{\phi},$$
			where $\overset{1}{\phi} \in \Lambda^1(M, \mathcal{h})$.
	\end{itemize}
	
We first define a generalized 2-connection $(\mathscr{A}, \mathscr{B})$ by
\begin{align}
\mathscr{A} = (\overset{1}{A}, \overset{2}{A}),\ \ \ \ \ 
\mathscr{B} = (\overset{2}{B}, \overset{3}{B}),
\end{align}
where $(\overset{1}{A}, \overset{2}{B})$ is the ordinary 2-connection, and $\overset{2}{A} \in \Lambda^2(M, \mathcal{g})$ and $\overset{3}{B}\in \Lambda^3 (M, \mathcal{h})$.
And the generalized 2-curvature $(\mathscr{F}, \mathscr{G})$ is given by
\begin{align*}
\mathscr{F} = \textbf{d} \mathscr{A} + \mathscr{A} \boldsymbol{\wedge}\mathscr{A} -  \bm{\alpha}(\mathscr{B}) = (\overset{2}{\mathbb{F}} + k \overset{2}{A}, D \overset{2}{A}- \alpha(\overset{3}{B})),\\
\mathscr{G} = \textbf{d} \mathscr{B} + \mathscr{A} \boldsymbol{\wedge}^{\blacktriangleright}\mathscr{B} = (\overset{3}{\mathbb{G}} - k \overset{3}{B}, D \overset{3}{B} + \overset{2}{A} \wedge^{\vartriangleright} \overset{2}{B}).
\end{align*}
	
	\begin{proposition}
		If $(\mathscr{A}, \mathscr{B})$ is any generalized $2$-connection on a principle 2-bundle over $M$, then its generalized 2-curvature $(\mathscr{F}, \mathscr{G})$ satisfies the generalized  $2$-Bianchi-Identities:
			\begin{align}
			&\textbf{d} \mathscr{F} + \mathscr{A} \boldsymbol{\wedge}^{\bm{[,]}} \mathscr{F} = - \bm{\alpha}(\mathscr{G});\\
			&\textbf{d} \mathscr{G} + \mathscr{A}\boldsymbol{\wedge}^{\blacktriangleright}\mathscr{G} = (\mathscr{F}+ \bm{\alpha}(\mathscr{B}))\boldsymbol{\wedge}^{\blacktriangleright} \mathscr{B}.
		\end{align}
	\end{proposition}
\begin{proof}
\begin{align*}
\textbf{d} \mathscr{F} + \mathscr{A} \boldsymbol{\wedge}^{\bm{[,]}} \mathscr{F} &= \textbf{d}(\overset{2}{\mathbb{F}} + k \overset{2}{A}, D \overset{2}{A} - \alpha(\overset{3}{B})) + (\overset{1}{A}, \overset{2}{A})\boldsymbol{\wedge} (\overset{2}{\mathbb{F}} + k \overset{2}{A}, D \overset{2}{A} - \alpha (\overset{3}{B})) \\
&- (\overset{2}{\mathbb{F}} + k \overset{2}{A}, D \overset{2}{A} - \alpha (\overset{3}{B}))  \boldsymbol{\wedge} (\overset{1}{A}, \overset{2}{A})\\
&=(d \overset{2}{\mathbb{F}} + k d \overset{2}{A} - k D \overset{2}{A} + k \alpha(\overset{3}{B}), d D \overset{2}{A} - d \alpha (\overset{3}{B})) + (\overset{1}{A} \wedge \overset{2}{\mathbb{F}} + k \overset{1}{A} \wedge \overset{2}{A}, \overset{1}{A} \wedge D \overset{2}{A} \\
&-\overset{1}{A}\wedge \alpha(\overset{3}{B})   + \overset{2}{A} \wedge \overset{2}{\mathbb{F}} + k \overset{2}{A} \wedge \overset{2}{A} ) + (\overset{2}{\mathbb{F}} \wedge \overset{1}{A} + k \overset{2}{A} \wedge \overset{1}{A}, \overset{2}{\mathbb{F}} \wedge \overset{2}{A} + k \overset{2}{A}\wedge \overset{2}{A} - D \overset{2}{A} \wedge \overset{1}{A} \\
&+ \alpha(\overset{3}{B})\wedge \overset{1}{A})\\
&=(D \overset{2}{\mathbb{F}} + k \alpha(\overset{3}{B}), - \overset{2}{A}\wedge^{[,]} \alpha(\overset{2}{B}) - D \alpha(\overset{3}{B}))\\
&= (- \alpha(d \overset{2}{B})- \alpha(\overset{1}{A}\wedge^{\vartriangleright} \overset{2}{B}) + k \alpha(\overset{3}{B}), -\alpha(\overset{2}{A}\wedge^{\vartriangleright}\overset{2}{B} + D \overset{3}{B}))\\
&=(-\alpha(d\overset{2}{B}+ \overset{1}{A}\wedge^{\vartriangleright}\overset{2}{B} - k \overset{3}{B}), -\alpha(\overset{2}{A}\wedge^{\vartriangleright}\overset{2}{B}+ D \overset{3}{B}))\\
&= - \bm{\alpha}(\mathscr{G}),\\
\textbf{d} \mathscr{G} + \mathscr{A}\boldsymbol{\wedge}^{\blacktriangleright}\mathscr{G} &= \textbf{d} (D \overset{2}{B} - k \overset{3}{B}, D \overset{3}{B} + \overset{2}{A} \wedge^{\vartriangleright} \overset{2}{B}) + (\overset{1}{A}\wedge \overset{2}{A} )\boldsymbol{\wedge}^{\vartriangleright}(D \overset{2}{B} - k \overset{3}{B}, D \overset{3}{B} + \overset{2}{A}\wedge^{\vartriangleright}\overset{2}{B})\\
&=(d D \overset{2}{B} - k d \overset{3}{B} + k D \overset{3}{B} + k \overset{2}{A} \wedge^{\vartriangleright}\overset{2}{B}, d D \overset{3}{B} + d \overset{2}{A} \wedge^{\vartriangleright} \overset{2}{B} + \overset{2}{A} \wedge^{\vartriangleright} d \overset{2}{B})+ (\overset{1}{A}\wedge^{\vartriangleright} D \overset{2}{B}\\
&- k \overset{1}{A}\wedge^{\vartriangleright}\overset{3}{B}, \overset{1}{A}\wedge^{\vartriangleright}D \overset{3}{B} +(\overset{1}{A}\wedge \overset{2}{A})\wedge^{\vartriangleright}\overset{2}{B} - \overset{2}{A}\wedge^{\vartriangleright}D \overset{2}{B} + k \overset{2}{A}\wedge^{\vartriangleright}\overset{3}{B})\\
&= ((d \overset{1}{A} + \overset{1}{A}\wedge \overset{1}{A})\wedge^{\vartriangleright}\overset{2}{B} + k \overset{2}{A}\wedge^{\vartriangleright}\overset{2}{B}, (d\overset{1}{A} + \overset{1}{A}\wedge \overset{1}{A})\wedge^{\vartriangleright}\overset{3}{B} + (d \overset{2}{A} + \overset{1}{A} \wedge \overset{2}{A} \\
&-  \overset{2}{A} \wedge  \overset{1}{A})\wedge^{\vartriangleright} \overset{2}{B} + k \overset{2}{A}\wedge^{\vartriangleright} \overset{3}{B})\\
&=(\overset{2}{\mathbb{F}}\wedge^{\vartriangleright}\overset{2}{B} + \alpha(\overset{2}{B})\wedge^{\vartriangleright}\overset{2}{B}+k \overset{2}{A}\wedge^{\vartriangleright}\overset{2}{B}, \overset{2}{\mathbb{F}} \wedge^{\vartriangleright}\overset{3}{B} + \alpha(\overset{2}{B})\wedge^{\vartriangleright}\overset{3}{B} + D \overset{2}{A}\wedge^{\vartriangleright}\overset{2}{B} + k \overset{2}{A}\wedge^{\vartriangleright} \overset{3}{B})\\
&=(\overset{2}{\mathbb{F}} + k \overset{2}{A} + \alpha(\overset{2}{B}), D \overset{2}{A}) \wedge^{\vartriangleright}(\overset{2}{B}, \overset{3}{B})\\
&= (\mathscr{F}+ \bm{\alpha}(\mathscr{B}))\boldsymbol{\wedge}^{\blacktriangleright} \mathscr{B},
\end{align*}
by using \eqref{2BI} and \eqref{22BI}.
\end{proof}
	
	We construct a generalized mixed relation 
	\begin{align}
		\bm{g}\blacktriangleright \mathscr{B}=(g, mg)\blacktriangleright (\overset{2}{B}, \overset{3}{B}):=(g \vartriangleright\overset{2}{B}, g \vartriangleright \overset{3}{B}+m\wedge^{\vartriangleright}(g \vartriangleright \overset{2}{B})),
	\end{align}
	where $\bm{g}$ and $m$ are as above. Then it follows that
	\begin{align}
		\bm{\alpha}(\bm{g}\blacktriangleright \mathscr{B})=\bm{g}\alpha(\mathscr{B})\bm{g^{-1}}.
	\end{align}
	It is clear that these gauge transformations \eqref{thin} and \eqref{fat} can be extended to the generalized case. Thus, we define corresponding generalized gauge transformations:
		\begin{itemize}
		\item [1)] Thin: 
		\begin{align}\label{gthin}
	   \mathscr{A'}&= \bm{g^{-1}}\mathscr{A} \bm{g} + \bm{g^{-1}}\textbf{d}\bm{g}\notag \\ &=(g^{-1}(\overset{1}{A} - k m)g + g^{-1}dg, g^{-1}(Dm - k m\wedge m + \overset{2}{A})g),\\
	   \mathscr{B'}&= \bm{g^{-1}}\blacktriangleright \mathscr{B}\notag\\
	   &=(g^{-1}\vartriangleright \overset{2}{B}, g^{-1}\vartriangleright (\overset{3}{B}- m\wedge^{\vartriangleright}\overset{2}{B})),
		\end{align}
		then the generalized curvatures change as 
		\begin{align}
		\mathscr{F'}&=\bm{g^{-1}}\mathscr{F}\bm{g} \notag \\
		&= (g^{-1}(\overset{2}{\mathbb{F}} + k \overset{2}{A})g, g^{-1}(D \overset{2}{A}- \alpha(\overset{3}{B}) + (\overset{2}{\mathbb{F}} + k \overset{2}{A}) \wedge^{[, ]}m)g), \\
		 \mathscr{G'}&=\bm{g^{-1}\blacktriangleright}\mathscr{G}\notag\\
		 &=(g^{-1}\vartriangleright(\overset{3}{\mathbb{G}} - k \overset{3}{B}), g^{-1}\vartriangleright(D \overset{3}{B} + \overset{2}{A}\wedge^{\vartriangleright}\overset{2}{B}) + (g^{-1}m)\wedge^{\vartriangleright}(D\overset{2}{B}-k \overset{3}{B})),
	 		\end{align}
		where $\bm{g} \in \bm{G}$.
		\item [2)] Fat: 
		\begin{align}\label{gfat}
		&\mathscr{A'}=\mathscr{A}+\bm{\alpha}(\bm{\phi}),\\
		&\mathscr{B'}=\mathscr{B}+ \textbf{d}\bm{\phi}+ \mathscr{A} \boldsymbol{\wedge}^{\blacktriangleright}\bm{\phi} + \bm{\phi} \boldsymbol{\wedge}\bm{\phi},
		\end{align}
		then the generalized curvatures change as 
		\begin{align}
	  	\mathscr{F'}&= \mathscr{F}\notag\\
	  	&=(\overset{2}{\mathbb{F}} + k \overset{2}{A}, D \overset{2}{A}- \alpha(\overset{3}{B})),\\
	  \mathscr{G'}&=\mathscr{G}+ \mathscr{F}\boldsymbol{\wedge}^{\blacktriangleright}\bm{\phi}\notag\\
  &= (\overset{3}{\mathbb{G}} - k \overset{3}{B} + (\overset{2}{\mathbb{F}} + k \overset{2}{A})\wedge^{\vartriangleright}\overset{1}{\phi}, D \overset{3}{B} + \overset{2}{A} \wedge^{\vartriangleright} \overset{2}{B} + (\overset{2}{\mathbb{F}} + k \overset{2}{A})\wedge^{\vartriangleright}\overset{2}{\phi}-(D \overset{2}{A}-\alpha(\overset{3}{B}))\wedge^{\vartriangleright}\overset{1}{\phi}),
		\end{align}
where $\bm{\phi}=(\overset{1}{\phi}, \overset{2}{\phi})\in \Lambda^1(M, \mathcal{h},1)$.
	\end{itemize}

	\subsection{Generalized 3-connections}
	Review the notion of ordinary 3-connection in \cite{SDH}. Let $(L, H, G;\beta, \alpha, \vartriangleright, \left\{,\right\})$  be a Lie 2-crossed module. The corresponding differential 2-crossed module is denoted as $(\mathcal{l}, \mathcal{h}, \mathcal{g};\beta, \alpha, \vartriangleright, \left\{,\right\})$. A 3-connection $(\overset{1}{A}, \overset{2}{B}, \overset{3}{C})$ on a principle $3$-bundle is given by
	\begin{align}
		\overset{1}{A}\in \Lambda^1(M, \mathcal{g}),\ \ \ \ \overset{2}{B} \in \Lambda^2(M, \mathcal{h}),\ \ \ \ \overset{3}{C} \in \Lambda^3(M, \mathcal{l}),
	\end{align}
	and the corresponding 3-curvature $(\overset{2}{\mathbb{F}}, \overset{3}{\mathbb{G}}, \overset{4}{\mathbb{H}})$ is given by
	\begin{align}
		\overset{2}{\mathbb{F}} = d \overset{1}{A} + \overset{1}{A}\wedge \overset{1}{A} - \alpha(\overset{2}{B}),
		\ \ \ \overset{3}{\mathbb{G}} = d \overset{2}{B} + \overset{1}{A}\wedge^{\vartriangleright} \overset{2}{B} - \beta (\overset{3}{C}),\ \ \ \overset{4}{\mathbb{H}} = d \overset{3}{C} + \overset{1}{A}\wedge^{\vartriangleright} \overset{3}{C}+ \overset{2}{B} \wedge^{\{,\}} \overset{2}{B}.
	\end{align}
	The 3-Bianchi-Identities are as follows
	\begin{align}
		&d \overset{2}{\mathbb{F}} + \overset{1}{A}\wedge^{[,]}\overset{2}{\mathbb{F}} = - \alpha(\overset{3}{\mathbb{G}} + \beta(\overset{3}{C})), \label{3BI1}  \\
	&	d \overset{3}{\mathbb{G}} + \overset{1}{A}\wedge^{\vartriangleright}\overset{3}{\mathbb{G}} = (\overset{2}{\mathbb{F}} + \alpha(\overset{2}{B}))\wedge^{\vartriangleright} \overset{2}{B} - \beta (\overset{4}{\mathbb{H}} - \overset{2}{B}\wedge^{\{, \}}\overset{2}{B}), \label{3BI2} \\
		& d \overset{4}{\mathbb{H}} + \overset{1}{A}\wedge^{\vartriangleright}\overset{4}{\mathbb{H}} = (\overset{2}{\mathbb{F}} + \alpha(\overset{2}{B})) \wedge^{\vartriangleright} \overset{3}{C} + (\overset{3}{\mathbb{G}} + \beta(\overset{3}{C}))\wedge^{\{, \}}\overset{2}{B} + \overset{2}{B}\wedge^{\{, \}} (\overset{3}{\mathbb{G}} + \beta (\overset{3}{C})).\label{3BI3}
	\end{align}
	
	We consider three types of gauge transformations given in \cite{TRMV1, SDH}.
	\begin{itemize}
		\item [1)] $G$-gauge transformation: 
	\begin{align}\label{Ggt}
	\overset{1}{A'}= g^{-1}  \overset{1}{A} g + g^{-1} dg,\ \ \ \ \ \overset{2}{B'} = g^{-1} \vartriangleright \overset{2}{B}, \ \ \ \ \ \overset{3}{C'} = g^{-1} \vartriangleright \overset{3}{C},
\end{align}
then the curvatures change as 
$$\overset{2}{\mathbb{F'}}=g^{-1}\overset{2}{\mathbb{F}}g,\ \ \ \ \  \overset{3}{\mathbb{G'}}=g^{-1}\vartriangleright \overset{3}{\mathbb{G}},\ \ \ \ \  \overset{4}{\mathbb{H'}}=g^{-1}\vartriangleright \overset{4}{\mathbb{H}},$$
where $g \in G$.
	\item [2)] $H$-gauge transformation: 
		\begin{align}\label{Hgt}
		\overset{1}{A'} =  \overset{1}{A} - \alpha(\overset{1}{\phi}),\ \ \overset{2}{B'} = \overset{2}{B} - d \overset{1}{\phi} -  \overset{1}{A'}  \wedge^{\vartriangleright} \overset{1}{\phi} - \overset{1}{\phi} \wedge \overset{1}{\phi},\ \ \overset{3}{C'} = \overset{3}{C} + \overset{2}{B'}\wedge^{\{, \}}\overset{1}{\phi} + \overset{1}{\phi}\wedge^{\{, \}}\overset{2}{B},
	\end{align}
	then the curvatures change as 
	$$\overset{2}{\mathbb{F'}}= \overset{2}{\mathbb{F}},\ \ \ \ \  \overset{3}{\mathbb{G'}}=\overset{3}{\mathbb{G}} - \overset{2}{\mathbb{F}}\wedge^{\vartriangleright}\overset{1}{\phi},\ \ \ \ \  \overset{4}{\mathbb{H'}}= \overset{4}{\mathbb{H}}+\overset{3}{\mathbb{G'}}\wedge^{\{, \}}\overset{1}{\phi}-\overset{1}{\phi}\wedge^{\{, \}}\overset{3}{\mathbb{G}},$$
	where $\overset{1}{\phi} \in \Lambda^1(M, \mathcal{h})$.
		\item [3)] $L$-gauge transformation: 
			\begin{align}\label{Lgt}
			\overset{1}{A'} =  \overset{1}{A},\ \ \overset{2}{B'} = \overset{2}{B} + \beta(\overset{2}{\psi}),\ \ \overset{3}{C'} = \overset{3}{C} + d \overset{2}{\psi} + \overset{1}{A}\wedge^{\vartriangleright}\overset{2}{\psi},
		\end{align}
		then the curvatures change as 
		$$\overset{2}{\mathbb{F'}}= \overset{2}{\mathbb{F}},\ \ \ \ \  \overset{3}{\mathbb{G'}}=\overset{3}{\mathbb{G}} ,\ \ \ \ \  \overset{4}{\mathbb{H'}}=\overset{4}{\mathbb{H}} +\overset{2}{\mathbb{F}}\wedge^{\vartriangleright}\overset{2}{\psi},$$
		where $\overset{2}{\psi} \in \Lambda^2(M, \mathcal{l})$.
	\end{itemize}

	Then we first define the generalized 3-connection $(\mathscr{A}, \mathscr{B}, \mathscr{C})$
	\begin{align}
		\mathscr{A} = (\overset{1}{A}, \overset{2}{A}),\ \ \ \ 
		\mathscr{B} = (\overset{2}{B}, \overset{3}{B}),\ \ \ \ 
		\mathscr{C} = (\overset{3}{C}, \overset{4}{C}),
	\end{align}
	where $(\overset{1}{A}, \overset{2}{B}, \overset{3}{C})$ is the ordinary 3-connection, and $\overset{2}{A} \in \Lambda^2(M, \mathcal{g})$, $\overset{3}{B}\in \Lambda^3 (M, \mathcal{h})$ and $\overset{4}{C} \in \Lambda^4(M, \mathcal{l})$.
	The generalized 3-curvature $(\mathscr{F}, \mathscr{G}, \mathscr{H} )$ is given by
	\begin{align*}
	&\mathscr{F} = \textbf{d} \mathscr{A} + \mathscr{A}\boldsymbol{\wedge} \mathscr{A} - \bm{\alpha}(\mathscr{B})= (\overset{2}{\mathbb{F}} + k \overset{2}{A}, D \overset{2}{A}- \alpha(\overset{3}{B})),\\
	&\mathscr{G} = \textbf{d} \mathscr{B} + \mathscr{A}\boldsymbol{\wedge}^{\blacktriangleright}\mathscr{B} - \bm{\beta}(\mathscr{C}) = (\overset{3}{\mathbb{G}} - k \overset{3}{B}, D \overset{3}{B} + \overset{2}{A} \wedge^{\vartriangleright} \overset{2}{B} - \beta(\overset{4}{C})),\\
	&\mathscr{H} = \textbf{d} \mathscr{C} + \mathscr{A}\boldsymbol{\wedge}^{\blacktriangleright} \mathscr{C} + \mathscr{B}\boldsymbol{\wedge}^{\bm{\{, \}}} \mathscr{B} =(\overset{4}{\mathbb{H}} + k \overset{4}{C}, D \overset{4}{C} - \overset{2}{A} \wedge^{\vartriangleright}\overset{3}{C} + \overset{2}{B}\wedge^{\{, \}}\overset{3}{B} +  \overset{3}{B}\wedge^{\{, \}}\overset{2}{B}).	\end{align*}

	\begin{proposition}
If $(\mathscr{A}, \mathscr{B}, \mathscr{C})$ is any generalized $3$-connection on a principle 3-bundle over $M$, then its generalized 3-curvature $(\mathscr{F}, \mathscr{G}, \mathscr{H})$ satisfies the generalized  $3$-Bianchi-Identities
		\begin{align}
		&d \mathscr{F} + \mathscr{A}\boldsymbol{\wedge}^{\bm{[,]}}\mathscr{F} = - \bm{\alpha}(\mathscr{G} + \bm{\beta}(\mathscr{C}));\\
		&d \mathscr{G} + \mathscr{A}\boldsymbol{\wedge}^{\blacktriangleright}\mathscr{G} = (\mathscr{F} + \bm{ \alpha}(\mathscr{B}))\boldsymbol{\wedge}^{\blacktriangleright} \mathscr{B} - \bm{\beta}(\mathscr{H}- \mathscr{B}\boldsymbol{\wedge}^{\bm{\{, \}}}\mathscr{B});\\
		& d \mathscr{H} + \mathscr{A}\boldsymbol{\wedge}^{\blacktriangleright}\mathscr{H} = (\mathscr{F} + \bm{\alpha}(\mathscr{B}))\boldsymbol{\wedge}^{\blacktriangleright}\mathscr{C} +(\mathscr{G} + \bm{\beta}(\mathscr{C}))\boldsymbol{\wedge}^{\bm{\{, \}}}\mathscr{B} + \mathscr{B}\boldsymbol{\wedge}^{\bm{\{, \}}} (\mathscr{G} + \bm{\beta}(\mathscr{C})).
	\end{align}
\end{proposition}
\begin{proof}
	\begin{align*}
	\textbf{d} \mathscr{F} + \mathscr{A}\boldsymbol{\wedge}^{\bm{[,]}}\mathscr{F} &=  \textbf{d}(\overset{2}{\mathbb{F}} + k \overset{2}{A}, D \overset{2}{A} - \alpha(\overset{3}{B}) )+ (\overset{1}{A}, \overset{2}{A}) \boldsymbol{\wedge}^{[, ]}(\overset{2}{\mathbb{F}} + k \overset{2}{A}, D \overset{2}{A} - \alpha(\overset{3}{B}))\\
	&= (d (\overset{2}{\mathbb{F}} +  \overset{2}{A}) - k (D \overset{2}{A} - \alpha(\overset{3}{B})), d(D \overset{2}{A} - \alpha(\overset{3}{B})))+ (\overset{1}{A} \wedge^{[,]}(\overset{2}{\mathbb{F}} + k \overset{2}{A}), \overset{1}{A}\wedge^{[, ]}(D \overset{2}{A} \\
	&- \alpha(\overset{3}{B}))  + \overset{2}{A} \wedge^{[, ]}(\overset{2}{\mathbb{F}} + k \overset{2}{A}))\\
	&=(d \overset{2}{\mathbb{F}} + \overset{1 }{A} \wedge^{[, ]} \overset{2}{\mathbb{F}} + k \alpha (\overset{3}{B}), - \alpha(d \overset{3}{B} + \overset{1}{A}\wedge^{\vartriangleright} \overset{3}{B} + \overset{2}{A}\wedge^{\vartriangleright} \overset{2}{B}))\\
	&= - \bm{\alpha}(\mathscr{G} + \bm{\beta}(\mathscr{C})),\\
	\textbf{d}\mathscr{G} + \mathscr{A}\boldsymbol{\wedge}^{\blacktriangleright}\mathscr{G} &=  \textbf{d}(\overset{3}{\mathbb{G}} - k \overset{3}{B}, D \overset{3}{B} + \overset{2}{A}\wedge^{\vartriangleright}\overset{2}{B} - \beta(\overset{4}{C})) + (\overset{1}{A}, \overset{2}{A})\boldsymbol{\wedge}^{\blacktriangleright} (\bm{G} - k \overset{3}{B}, D \overset{3}{B} + \overset{2}{A}\wedge^{\vartriangleright}\overset{2}{B} - \beta(\overset{4}{C}))\\
	&= (d \overset{3}{\mathbb{G}} + \overset{1}{A}\wedge^{\vartriangleright}\overset{3}{\mathbb{G}} + k (\overset{2}{A} \wedge^{\vartriangleright}\overset{2}{B} - \beta(\overset{4}{C})), (d \overset{1}{A} + \overset{1}{A}\wedge \overset{1}{A})\wedge^{\vartriangleright} \overset{3}{B} + d \overset{2}{A}\wedge^{\vartriangleright}\overset{2}{B}\\
	&+ (\overset{1}{A}\wedge \overset{2}{A})\wedge^{\vartriangleright} \overset{2}{B} - d (\beta(\overset{4}{C})) - \overset{1}{A} \wedge^{\vartriangleright} (\beta(\overset{4}{C})) - \overset{2}{A}\wedge^{\vartriangleright}(\overset{1}{A} \wedge^{\vartriangleright}\overset{2}{B} - \beta(\overset{3}{C})- k \overset{3}{B}))\\
	&= ( (\overset{2}{\mathbb{F}} + \alpha(\overset{2}{B}))\wedge^{\vartriangleright}\overset{2}{B} - \beta(\overset{4}{\mathbb{H}}- \overset{2}{B}\wedge^{\{, \}}\overset{2}{B}) + k \overset{2}{A}\wedge^{\vartriangleright}\overset{2}{B} - \beta(k \overset{4}{C}), (\overset{2}{\mathbb{F}} + k \overset{2}{A} + \alpha(\overset{2}{B})) \wedge^{\vartriangleright} \overset{3}{B} \\
	&+ D \overset{2}{A}\wedge^{\vartriangleright}\overset{2}{B} - \beta(D \overset{4}{C} - \overset{2}{A}\wedge^{\vartriangleright} \overset{3}{C}))\\
	&= (\mathscr{F} + \bm{ \alpha}(\mathscr{B}))\boldsymbol{\wedge}^{\blacktriangleright} \mathscr{B} - \bm{\beta}(\mathscr{H}- \mathscr{B}\boldsymbol{\wedge}^{\bm{\{, \}}}\mathscr{B}),\\
	 \textbf{d} \mathscr{H} + \mathscr{A}\wedge^{\blacktriangleright}\mathscr{H} &= \textbf{d}(\overset{4}{\mathbb{H}} + k \overset{4}{C}, D \overset{4}{C} - \overset{2}{A}\wedge^{\vartriangleright}\overset{3}{C} + \overset{2}{B} \wedge^{\{, \}}\overset{3}{B}+ \overset{3}{B} \wedge^{\{, \}}\overset{2}{B})  + (\overset{1}{A}, \overset{2}{A})\boldsymbol{\wedge}^{\blacktriangleright} (\overset{4}{\mathbb{H}} + k \overset{4}{C}, D \overset{4}{C} \\
	 &- \overset{2}{A}\wedge^{\vartriangleright}\overset{3}{C} + \overset{2}{B} \wedge^{\{, \}}\overset{3}{B}+ \overset{3}{B} \wedge^{\{, \}}\overset{2}{B})\\
	 &= (d \overset{4}{\mathbb{H}} + \overset{1}{A}\wedge^{\vartriangleright}\overset{4}{\mathbb{H}} + k \overset{2}{A}\wedge^{\vartriangleright}\overset{3}{C} - k (\overset{2}{B} \wedge^{\{, \}}\overset{3}{B}+ \overset{3}{B} \wedge^{\{, \}}\overset{2}{B}), d (D \overset{4}{C} -  \overset{2}{A}\wedge^{\vartriangleright}\overset{3}{C} + \overset{2}{B} \wedge^{\{, \}}\overset{3}{B}\\
	 &+ \overset{3}{B} \wedge^{\{, \}}\overset{2}{B}) + \overset{1}{A}\wedge^{\vartriangleright}(D \overset{4}{C} - \overset{2}{A}\wedge^{\vartriangleright}\overset{3}{C} + \overset{2}{B} \wedge^{\{, \}}\overset{3}{B} +  \overset{3}{B} \wedge^{\{, \}}\overset{2}{B} ) + \overset{2}{A}\wedge^{\vartriangleright}(\overset{4}{\mathbb{H}} + k \overset{4}{C})) \\
	 &= ((\overset{2}{\mathbb{F}} + k \overset{2}{A} + \alpha(\overset{2}{B}))\wedge^{\vartriangleright} \overset{3}{C} + ( \overset{3}{\mathbb{G}} - k \overset{3}{B} + \beta(\overset{3}{C}))\wedge^{\{, \}}\overset{2}{B} + \overset{2}{B}\wedge^{\{, \}}( \overset{3}{\mathbb{G}} - k \overset{3}{B} + \beta(\overset{3}{C})), (\overset{2}{\mathbb{F}} \\
	 &+ k \overset{2}{A} + \alpha(\overset{2}{B}))\wedge^{\vartriangleright} \overset{4}{C} - D \overset{2}{A}\wedge^{\vartriangleright}\overset{3}{C} + (\overset{3}{\mathbb{G}} - k \overset{3}{B} + \beta(\overset{3}{C}))\wedge^{\{, \}}\overset{3}{B} + (D \overset{3}{B} + \overset{2}{A}\wedge^{\vartriangleright}\overset{2}{B}) \wedge^{\{, \}} \overset{2}{B} \\
	 &+ \overset{2}{B}\wedge^{\{, \}} (D \overset{3}{B} + \overset{2}{A}\wedge^{\vartriangleright}\overset{2}{B}) - \overset{3}{B}\wedge^{\{, \}}(\overset{3}{\mathbb{G}} - k \overset{3}{B} + \beta(\overset{3}{C})))\\
	 &=  (\mathscr{F} + \bm{\alpha}(\mathscr{B}))\boldsymbol{\wedge}^{\blacktriangleright}\mathscr{C} +(\mathscr{G} + \bm{\beta}(\mathscr{C})) \boldsymbol{\wedge}^{\bm{\{, \}}}\mathscr{B} + \mathscr{B}\boldsymbol{\wedge}^{\bm{\{, \}}} (\mathscr{G} + \bm{\beta}(\mathscr{C})),
	\end{align*}
by using \eqref{3BI1}, \eqref{3BI2} and \eqref{3BI3}.
\end{proof}

		We construct a generalized mixed relation 
	\begin{align}
		\bm{g}\blacktriangleright \mathscr{C}=(g, mg)\blacktriangleright (\overset{3}{C}, \overset{4}{C}):=(g \vartriangleright\overset{3}{C}, g \vartriangleright \overset{4}{C}+m\wedge^{\vartriangleright}(g \vartriangleright \overset{3}{C})),
	\end{align}
	where $\bm{g}$ and $m$ are as above. Then it follows that
	\begin{align}
		\bm{\beta}(\bm{g}\blacktriangleright \mathscr{C})=\bm{g}\blacktriangleright\beta(\mathscr{C}).
	\end{align}

		It is clear that these gauge transformations \eqref{Ggt}, \eqref{Hgt}and \eqref{Lgt} can be extended to the generalized case. Thus, we define generalized gauge transformations:
		\begin{itemize}
		\item [1)] generalized $G$-gauge transformation: 
		\begin{align}\label{gGgt}
		 \mathscr{A'}&= \bm{g^{-1}}\mathscr{A} \bm{g} + \bm{g^{-1}}\textbf{d}\bm{g}\notag \\ &=(g^{-1}(\overset{1}{A} - k m)g + g^{-1}dg, g^{-1}(Dm - k m\wedge m + \overset{2}{A})g),\\
		\mathscr{B'}&= \bm{g^{-1}}\blacktriangleright \mathscr{B}\notag\\
		&=(g^{-1}\vartriangleright \overset{2}{B}, g^{-1}\vartriangleright (\overset{3}{B}- m\wedge^{\vartriangleright}\overset{2}{B})),\\
		 \mathscr{C'}&=\bm{g^{-1}}\blacktriangleright \mathscr{C}\notag\\
		 &=(g^{-1}\vartriangleright \overset{3}{C}, g^{-1}\vartriangleright(\overset{4}{C} + m\wedge^{\vartriangleright}\overset{3}{C})),
		\end{align}
		then the generalized curvatures change as 
	\begin{align}
		\mathscr{F'}&=\bm{g^{-1}}\mathscr{F}\bm{g} \notag \\
		&= (g^{-1}(\overset{2}{\mathbb{F}} + k \overset{2}{A})g, g^{-1}(D \overset{2}{A}- \alpha(\overset{3}{B}) + (\overset{2}{\mathbb{F}} + k \overset{2}{A}) \wedge^{[, ]}m)g), \\
		\mathscr{G'}&=\bm{g^{-1}\blacktriangleright}\mathscr{G}\notag\\
		&=(g^{-1}\vartriangleright(\overset{3}{\mathbb{G}} - k \overset{3}{B}), g^{-1}\vartriangleright(D \overset{3}{B} + \overset{2}{A}\wedge^{\vartriangleright}\overset{2}{B}) + (g^{-1}m)\wedge^{\vartriangleright}(D\overset{2}{B}-k \overset{3}{B})),\\
		\mathscr{H'}&=\bm{g^{-1}\blacktriangleright}\mathscr{H}\notag\\
		&=(g^{-1}\vartriangleright(\overset{4}{\mathbb{H}} + k \overset{4}{C}), g^{-1}\vartriangleright (D \overset{4}{C}-\overset{2}{A}\wedge^{\vartriangleright}\overset{3}{C}+ \overset{2}{B}\wedge^{\{, \}}\overset{3}{B} + \overset{3}{B}\wedge^{\{, \}}\overset{2}{B} \notag\\
		&- m \wedge^{\vartriangleright}(D\overset{3}{C}+\overset{2}{B}\wedge^{\{, \}}\overset{2}{B} + k \overset{4}{C}))).
	\end{align}
		\item [2)] generalized $H$-gauge transformation: 
		\begin{align}\label{gHgt}
		 \mathscr{A'}&= \mathscr{A}-\bm{\alpha}(\bm{\phi})\notag\\
		 &=(\overset{1}{A}-\alpha(\overset{1}{\phi}), \overset{2}{A}-\alpha(\overset{2}{\phi})),\\
		  \mathscr{B'}&= \mathscr{B}-\textbf{d}\bm{\phi}-\mathscr{A'}\boldsymbol{\wedge}^{\blacktriangleright}\bm{\phi}-\bm{\phi}\boldsymbol{\wedge}\bm{\phi}\notag\\
		  &=(\overset{2}{B}-D\overset{1}{\phi}-\overset{1}{\phi}\wedge \overset{1}{\phi}-k\overset{2}{\phi}+\alpha(\overset{1}{\phi})\wedge^{\vartriangleright}\overset{1}{\phi}, \overset{3}{B}-D\overset{2}{\phi}+\alpha(\overset{1}{\phi})\wedge^{\vartriangleright}\overset{2}{\phi}\notag\\
		  &+(\overset{2}{A}-\alpha(\overset{2}{\phi}))\wedge^{\vartriangleright}\overset{1}{\phi}-\overset{1}{\phi}\wedge^{[, ]}\overset{2}{\phi}),\\
		  \mathscr{C'}&=\mathscr{C}+\mathscr{B'}\boldsymbol{\wedge}^{\bm{\{, \}}}\bm{\phi} + \bm{\phi}\boldsymbol{\wedge}^{\bm{\{, \}}}\mathscr{B}\notag\\
		  &=(\overset{3}{C}+(\overset{2}{B}-D\overset{1}{\phi}-k\overset{2}{\phi}+\alpha(\overset{1}{\phi})\wedge^{\vartriangleright}\overset{1}{\phi}-\overset{1}{\phi}\wedge\overset{1}{\phi})\wedge^{\{, \}}\overset{1}{\phi}+\overset{1}{\phi}\wedge^{\{, \}}\overset{2}{B}, \overset{4}{C}+(\overset{2}{B}\notag\\
		  &-D\overset{1}{\phi}-k\overset{2}{\phi}+\alpha(\overset{1}{\phi})\wedge^{\vartriangleright}\overset{1}{\phi}-\overset{1}{\phi}\wedge\overset{1}{\phi})\wedge^{\{, \}}\overset{2}{\phi}-(\overset{3}{B}- D\overset{2}{\phi}+\alpha(\overset{1}{\phi})\wedge^{\vartriangleright}\overset{2}{\phi}+(\overset{2}{A}\notag\\
		  &-\alpha(\overset{2}{\phi}))\wedge^{\vartriangleright}\overset{1}{\phi}- \overset{1}{\phi}\wedge^{[, ]}\overset{2}{\phi})\wedge^{\{, \}}\overset{1}{\phi}+\overset{1}{\phi}\wedge^{\{, \}}\overset{3}{B}+\overset{2}{\phi}\wedge^{\{, \}}\overset{2}{B}),
		\end{align}
	where $\bm{\phi}=(\overset{1}{\phi}, \overset{2}{\phi})\in \Lambda^1(M, \mathcal{g},1)$, then the generalized curvatures change as 
		\begin{align}
		\mathscr{F'}&=\mathscr{F} \notag \\
		&= (\overset{2}{\mathbb{F}} + k \overset{2}{A}, D \overset{2}{A}- \alpha(\overset{3}{B})), \\
		\mathscr{G'}&=\mathscr{G}-\mathscr{F}\boldsymbol{\wedge}^{\blacktriangleright}\bm{\phi}\notag\\
		&=(\overset{3}{\mathbb{G}}-k \overset{3}{B}-(\overset{2}{\mathbb{F}}+k\overset{2}{A})\wedge^{\vartriangleright}\overset{1}{\phi}, D \overset{3}{B}-\beta(\overset{4}{C})+\overset{2}{A}\wedge^{\vartriangleright}\overset{2}{B}-(\overset{2}{\mathbb{F}}+k \overset{2}{A})\wedge^{\vartriangleright}\overset{2}{\phi}\notag\\
		&+(D\overset{2}{A}-\alpha(\overset{3}{B}))\wedge^{\vartriangleright}\overset{1}{\phi}),\\
		\mathscr{H'}&=\mathscr{H}+\mathscr{G'}\boldsymbol{\wedge}^{\bm{\{, \}}}\bm{\phi} - \bm{\phi}\boldsymbol{\wedge}^{\bm{\{, \}}}\mathscr{G}\notag\\
		&=(\overset{4}{\mathbb{H}}+k\overset{4}{C} + (\overset{3}{\mathbb{G}}-k\overset{3}{B}-(\overset{2}{\mathbb{F}}+k\overset{2}{A})\wedge^{\vartriangleright}\overset{1}{\phi})\wedge^{\{, \}}\overset{1}{\phi}-\overset{1}{\phi}\wedge^{\{, \}}(\overset{3}{\mathbb{G}}-k\overset{3}{B}), D\overset{4}{C}-\overset{2}{A}\wedge^{\vartriangleright}\overset{3}{C}\notag\\
		&+\overset{2}{B}\wedge^{\{, \}}\overset{3}{B}+\overset{3}{B}\wedge^{\{, \}}\overset{2}{B}+ (\overset{3}{\mathbb{G}}-k\overset{3}{B}-(\overset{2}{\mathbb{F}}+k\overset{2}{A})\wedge^{\vartriangleright}\overset{1}{\phi})\wedge^{\{, \}}\overset{2}{\phi}- (D \overset{3}{B}-\beta(\overset{4}{C})+\overset{2}{A}\wedge^{\vartriangleright}\overset{2}{B}\notag\\
		&-(\overset{2}{\mathbb{F}}+k \overset{2}{A})\wedge^{\vartriangleright}\overset{2}{\phi}+(D\overset{2}{A}-\alpha(\overset{3}{B}))\wedge^{\vartriangleright}\overset{1}{\phi})\wedge^{\{, \}}\overset{1}{\phi}- \overset{1}{\phi}\wedge^{\{, \}}(D \overset{3}{B} + \overset{2}{A} \wedge^{\vartriangleright} \overset{2}{B} - \beta(\overset{4}{C})\notag\\
		&+\overset{2}{\phi}\wedge^{\{, \}}(\overset{3}{\mathbb{G}}-k\overset{3}{B})).
	\end{align}
		\item [3)] $L$-gauge transformation: 
		\begin{align}\label{gLgt}
			\mathscr{A'}&=\mathscr{A}=(\overset{1}{A}, \overset{2}{A}),\\
			\mathscr{B'}&=\mathscr{B}+\bm{\beta}(\bm{\psi})=(\overset{2}{B}+\beta(\overset{2}{\psi}), \overset{3}{B}+\beta(\overset{3}{\psi})),\\ \mathscr{C'}&=\mathscr{C}+\textbf{d}\bm{\psi}+\mathscr{A}\boldsymbol{\wedge}^{\blacktriangleright}\bm{\psi}=(\overset{3}{C}+D\overset{2}{\psi}-k\overset{3}{\psi}, \overset{4}{C}+D\overset{3}{\psi}+\overset{2}{A}\wedge^{\vartriangleright}\overset{2}{\psi}),
		\end{align}
		then the generalized curvatures change as 
		\begin{align}
		\mathscr{F'}&=\mathscr{F}= (\overset{2}{\mathbb{F}} + k \overset{2}{A}, D \overset{2}{A}- \alpha(\overset{3}{B})), \\
		\mathscr{G'}&=\mathscr{G}=(\overset{3}{\mathbb{G}} - k \overset{3}{B}, D \overset{3}{B} + \overset{2}{A} \wedge^{\vartriangleright} \overset{2}{B} - \beta(\overset{4}{C})),\\ \mathscr{H'}&=\mathscr{H}+\mathscr{F}\boldsymbol{\wedge}^{\blacktriangleright}\bm{\psi}\notag\\
		&=(\overset{4}{\mathbb{H}} + k \overset{4}{C} + (\overset{2}{\mathbb{F}}+k\overset{2}{A})\wedge^{\vartriangleright}\overset{2}{\psi}, D \overset{4}{C} - \overset{2}{A} \wedge^{\vartriangleright}\overset{3}{C} + \overset{2}{B}\wedge^{\{, \}}\overset{3}{B} +  \overset{3}{B}\wedge^{\{, \}}\overset{2}{B}\notag\\
		&+(\overset{2}{\mathbb{F}}+k\overset{2}{A})\wedge^{\vartriangleright}\overset{3}{\psi}+(D\overset{2}{A}-\alpha(\overset{3}{B}))\wedge^{\vartriangleright}\overset{2}{\psi}),
		\end{align}
		where $\bm{\psi}=(\overset{2}{\psi}, \overset{3}{\psi}) \in  \Lambda^2(M, \mathcal{l},1)$.
	\end{itemize}

\section{Generalized (higher) Yang-Mills}\label{sec6}
	
	\subsection{Generalized Yang-Mills}
The following calculations follow the integration rules in \eqref{in3}. As we shall see, there is a generalized desired formalism of Yang-Mills equation, staying same with the ordinary Yang-Mills equation in form. As a remark, we note that all calculations are carried out using the GDC.  
	
Choosing a Lie group $G$ with Lie algebra $\mathcal{g}$ and $\bm{\tilde{M}}= (\partial M, M)$,  where $M$ is an $n$ dimensional manifold with metric and $\partial M$ is its boundary, let $\mathscr{A} = (\overset{1}{A}, \overset{2}{A}) \in \Lambda^1(M, \mathcal{g}, 1)$ be a generalized connection $1$-form, namely generalized Yang-Mills potential. Then the corresponding generalized curvature $2$-form $\mathscr{F}$, namely generalized Yang-Mills field, can be given by 
 \begin{align}
	\mathscr{F}= (F + k \overset{2}{A} , D \overset{2}{A}).
\end{align}

We can construct the generalized Yang-Mills action  
\begin{align}\label{GYM1}
	\bm{S_{GYM}}= \int_{\bm{\tilde{M}}} \ll \mathscr{F}, \star \mathscr{F} \gg.
\end{align}
Take the variational derivative of \eqref{GYM1}, and the detailed calculation process is shown in the Appendix \ref{YME}. 
  One can see that varying the action with respect to the variables $\overset{1}{A}$ and  $\overset{2}{A}$, one
obtains the generalized Yang-Mills equations 
\begin{align}\label{GYM11}
 D \ast F + k D \ast \overset{2}{A} + \overset{2}{A}\wedge^{[,]}\ast D \overset{2}{A}= 0, \notag \\
k \ast F + k^2 \ast \overset{2}{A} - D \ast D \overset{2}{A} = 0.
\end{align}
We find that the above equations are equal to 
\begin{align}\label{YM}
\mathscr{D} \star \mathscr{F} = \textbf{d} \star \mathscr{F} + \mathscr{A} \boldsymbol{\wedge}\star \mathscr{F} + (-1)^n \star \mathscr{F} \boldsymbol{\wedge} \mathscr{A} = 0,
\end{align}
which is the generalized Yang-Mills equations in \cite{pnr}.  In this notation, it is well-known that the generalized Yang-Mills equation \eqref{YM} has a very similar appearance to the generalized Bianchi-Identity \eqref{BI1}.
The reference \cite{pnr} restricts only to the reconstruction of the generalized form of Yang-Mills equation, but we reconstruct the generalized Yang-Mills action obtaining equation  \eqref{YM}.
Especially, when the connection 2-form $\overset{2}{A} = 0$, the action \eqref{GYM1} will be from the generalized case to the ordinary case, and there is an ordinary Yang-Mills equation $D \ast F =0 $.
On the other hand, one can note that the connection 1-form $\overset{1}{A}$ is flat in this case, i.e. $F =0$ from \eqref{GYM11}.

	\subsection{Generalized 2-form Yang-Mills}
	The related study of ordinary higher Yang-Mills theory is originally proposed in \cite{2002hep.th....6130B} in which Baez constructs the 2-form Yang-Mills  based on a 2-group. Encouraged by his idea, we develop the 3-form Yang-Mills based on a 3-group in our previous article \cite{sdh}. In this section, we shall do the same for the concept of generalized higher connections, establishing a generalized higher Yang-Mills theory. At first, let us develop generalized 2-form Yang-Mills in this subsection.
	
	Given a Lie crossed module $(G, H; \alpha, \vartriangleright)$ with the corresponding differential crossed module $(\mathcal{g}, \mathcal{h}; \alpha, \vartriangleright)$, let $(\mathscr{A}, \mathscr{B})$ be a generalized 2-connection, namely generalized 2-form Yang-Mills potential, where $\mathscr{A} = (\overset{1}{A}, \overset{2}{A}) \in \Lambda^1(M, \mathcal{g}, 1)$ and $\mathscr{B} = (\overset{2}{B}, \overset{3}{B}) \in \Lambda^2(M, \mathcal{h}, 1)$. Then the generalized 2-curvature $(\mathscr{F}, \mathscr{G})$, namely generalized 2-form Yang-Mills field, is given by
	\begin{align*}
		\mathscr{F} = (\overset{2}{\mathbb{F}} + k \overset{2}{A}, D \overset{2}{A}- \alpha(\overset{3}{B})),\\
		\mathscr{G} = (\overset{3}{\mathbb{G}} - k \overset{3}{B}, D \overset{3}{B} + \overset{2}{A} \wedge^{\vartriangleright} \overset{2}{B}),
	\end{align*}
where $\overset{2}{\mathbb{F}}= d \overset{1}{A} + \overset{1}{A} \wedge \overset{1}{A} - \alpha(\overset{2}{B} )$, and  $\overset{3}{\mathbb{G}}= d \overset{2}{B} + \overset{1}{A} \wedge^{\vartriangleright} \overset{2}{B}$.
	
	According to the construction of ordinary 2-form Yang-Mills action \cite{SDH}, we establish generalized 2-form Yang-Mills action 
	\begin{align}\label{GYM2}
		\bm{S_{G2YM}}= \int_{\bm{\tilde{M}}} \ll \mathscr{F}, \star \mathscr{F} \gg +  \ll \mathscr{G}, \star \mathscr{G} \gg.
	\end{align}
Take the variational derivative of \eqref{GYM2}, and see the Appendix \ref{YME} for details.
One can see that varying the action with respect to the variables $\overset{1}{A}$, $\overset{2}{A}$, $\overset{2}{B}$, and $\overset{3}{B}$, one
obtains generalized 2-form Yang-Mills equations 
\begin{align}\label{GYM22}
&D \ast (\overset{2}{\mathbb{F}} + k \overset{2}{A})+\overset{2}{A}\wedge^{[,]}\ast (D \overset{2}{A}- \alpha (\overset{3}{B}))- \overline{\sigma}(\ast(\overset{3}{\mathbb{G}}- k \overset{3}{B}), \overset{2}{B}) + (-1)^{n+1}\overline{\sigma}(\ast(D \overset{3}{B} + \overset{2}{A}\wedge^{\vartriangleright}\overset{2}{B}), \overset{3}{B})=0; \notag\\
&D \ast (D \overset{2}{A} - \alpha (\overset{3}{B}) )- k \ast (\overset{2}{\mathbb{F}} + k \overset{2}{A}) + \overline{\sigma} (\ast (D \overset{3}{B}+ \overset{2}{A}\wedge^{\vartriangleright} \overset{2}{B}), \overset{2}{B})=0;\notag\\
&D \ast (\overset{3}{\mathbb{G}}- k \overset{3}{B} )+ \alpha^{\ast}\ast(\overset{2}{\mathbb{F}} + k \overset{2}{A}) + \overset{2}{A}\wedge^{\vartriangleright}\ast (D \overset{3}{B}+ \overset{2}{A}\wedge^{\vartriangleright}\overset{2}{B})=0;\notag\\
 &D \ast (D \overset{3}{B} + \overset{2}{A}\wedge^{\vartriangleright} \overset{2}{B})- \alpha^{\ast} \ast (D \overset{2}{A} - \alpha (\overset{3}{B})) -k \ast (\overset{3}{\mathbb{G}} - k \overset{3}{B}) =0.
\end{align}
We find that the above equations are equal to 
\begin{align}\label{YM2}
	d \star \mathscr{F} + \mathscr{A} \wedge^{[,]}\star \mathscr{F} = \overline{\sigma}(\star \mathscr{G}, \mathscr{B}); \notag\\
	d \star \mathscr{G} + \mathscr{A}\wedge^{\vartriangleright}\star \mathscr{G}= - \alpha^{\ast}(\star \mathscr{F}).
\end{align}
	   In addition, we can also apply the results of section \ref{sec3} and section \ref{sec4} to compute directly the variational derivative of the generalized action to obtain the equations \eqref{YM2}. 
	  In other wards, all derivation process in \cite{SDH} can be extended straightforwardly to the generalized case where the GDC works merely.
	  Especially, when $\overset{2}{A} = 0$ and $\overset{3}{B}=0$, the action \eqref{GYM2} will become the ordinary case, and there are ordinary 2-form Yang-Mills equations
	  \begin{align}
	  	D \ast \overset{2}{\mathbb{F}}=\overline{\sigma}(\ast \overset{3}{\mathbb{G}}, \overset{2}{B}), \ \ \ \  
	  	D \ast \overset{3}{\mathbb{G}} =-\alpha^*(\ast \overset{2}{\mathbb{F}}).
	  \end{align}
	   Besides, it can easily be seen directly that the 2-connection 1-form $\overset{1}{A}$ and 2-connection 2-form $\overset{2}{B}$ are flat from \eqref{GYM22}, i.e. $\overset{2}{\mathbb{F}}=0$ and $\overset{3}{\mathbb{G}}=0$ in this case.
	  
	  	\subsection{Generalized 3-form Yang-Mills}
	  	Given a Lie 2-crossed module $(L, H, G;\beta, \alpha, \vartriangleright, \left\{,\right\})$ with the differential structure $(\mathcal{l}, \mathcal{h}, \mathcal{g};\beta, \alpha, \vartriangleright, \left\{,\right\})$, we consider a  generalized 3-connection $(\mathscr{A}, \mathscr{B, \mathscr{C}})$, namely generalized 3-form Yang-Mills potential, where  	$\mathscr{A} = (\overset{1}{A}, \overset{2}{A})\in \Lambda^1(M, \mathcal{g}, 1)$, $\mathscr{B} = (\overset{2}{B}, \overset{3}{B}) \in \Lambda^2(M, \mathcal{h}, 1)$, and $\mathscr{C} = (\overset{3}{C}, \overset{4}{C})\in \Lambda^3(M, \mathcal{l}, 1)$. The generalized 3-curvature $(\mathscr{F}, \mathscr{G}, \mathscr{H})$, namely generalized 3-form Yang-Mills field, is given by
	  		\begin{align*}
	  		&	\mathscr{F} = (\overset{2}{\mathbb{F}} + k \overset{2}{A}, D \overset{2}{A}- \alpha(\overset{3}{B}));\\
	  		&	\mathscr{G} = (\overset{3}{\mathbb{G}} - k \overset{3}{B}, D \overset{3}{B} + \overset{2}{A} \wedge^{\vartriangleright} \overset{2}{B} - \beta(\overset{4}{C}));\\
	  		&	\mathscr{H} = (\overset{4}{\mathbb{H}} + k \overset{4}{C}, D \overset{4}{C} - \overset{2}{A} \wedge^{\vartriangleright}\overset{3}{C} + \overset{2}{B}\wedge^{\{, \}}\overset{3}{B} +  \overset{3}{B}\wedge^{\{, \}}\overset{2}{B}).
	  	\end{align*}
	  	
	  	According to the construction of ordinary 3-form Yang-Mills action \cite{sdh}, we establish generalized 3-form Yang-Mills action 
	 \begin{align}\label{GYM3}
	 	\bm{S_{G3YM}}= \int_{\bm{\tilde{M}}} \ll \mathscr{F}, \star \mathscr{F} \gg +  \ll \mathscr{G}, \star \mathscr{G} \gg + \ll \mathscr{H}, \star \mathscr{H} \gg.
	 \end{align}
	 Take the variational derivative $\delta \bm{S_{G3YM}}$ and see the appendix \ref{YME} for details. One can see that varying the action with respect to the variables $\overset{1}{A}$, $\overset{2}{A}$, $\overset{2}{B}$, $\overset{3}{B}$, $\overset{3}{C}$ and $\overset{4}{C}$, one
	 obtains generalized 3-form Yang-Mills equations 
   \begin{align}\label{GYM33}
   	&D \ast (\overset{2}{\mathbb{F}} + k \overset{2}{A}) + \overset{2}{A}\wedge^{[, ]}\ast (D \overset{2}{A} - \alpha(\overset{2}{B}) ) - \overline{\sigma}(\ast (\overset{3}{\mathbb{G}} - k \overset{3}{B}), \overset{2}{B}) + (-1)^{n+1} \overline{\sigma}(\ast (D \overset{3}{B} + \overset{2}{A}\wedge^{\vartriangleright}\overset{2}{B} - \beta(\overset{4}{C})), \overset{3}{B}) \notag\\
   	& + (-1)^{n+1}\overline{\kappa}(\ast (\overset{4}{\mathbb{H}}+ k \overset{4}{C}), \overset{3}{C}) - \overline{\kappa}(\ast (D \overset{4}{C} - \overset{2}{A}\wedge^{\vartriangleright}\overset{3}{C} + \overset{2}{B}\wedge^{\{, \}}\overset{3}{B} + \overset{3}{B}\wedge^{\{, \}}\overset{2}{B}) , \overset{4}{C})  = 0;\notag \\
	&k \ast (\overset{2}{\mathbb{F}} + k \overset{2}{A}) - D \ast (D \overset{2}{A} - \alpha(\overset{3}{B})) - \overline{\sigma}(\ast (D \overset{3}{B} + \overset{2}{A}\wedge^{\vartriangleright}\overset{2}{B} - \beta(\overset{4}{C})), \overset{2}{B}) + (-1)^{n+1} \overline{\kappa}(\ast (D \overset{4}{C}- \overset{2}{A}\wedge^{\vartriangleright}\overset{3}{C} \notag \\
	&+ \overset{2}{B}\wedge^{\{, \}}\overset{3}{B} + \overset{3}{B}\wedge^{\{, \}}\overset{2}{B}) , \overset{3}{C}) = 0;\notag \\
	&\alpha^{\ast}(\ast \overset{2}{\mathbb{F}} + k \ast \overset{2}{A}) + D \ast (\overset{3}{\mathbb{G}} - k \overset{3}{B}) + \overset{2}{A}\wedge^{\vartriangleright}\ast (D \overset{3}{B} + \overset{2}{A}\wedge^{\vartriangleright}\overset{2}{B} - \beta(\overset{4}{C}))+ \overline{\eta_2}(\ast (\overset{4}{\mathbb{H}} + k \overset{4}{C}), \overset{2}{B}) +  \overline{\eta_1}(\ast (\overset{4}{\mathbb{H}} \notag \\
	&+ k \overset{4}{C}), \overset{2}{B}) + (-1)^{n+1}\overline{\eta_1}(\ast (D \overset{4}{C} - \overset{2}{A}\wedge^{\vartriangleright}\overset{3}{C} + \overset{2}{B}\wedge^{\{, \}}\overset{3}{B} + \overset{3}{B}\wedge^{\{, \}}\overset{2}{B}) , \overset{3}{B}) + (-1)^{n+1}\overline{\eta_2}(\ast (D \overset{4}{C} - \overset{2}{A}\wedge^{\vartriangleright}\overset{3}{C} \notag \\
	&+ \overset{2}{B}\wedge^{\{, \}}\overset{3}{B} + \overset{3}{B}\wedge^{\{, \}}\overset{2}{B}) , \overset{3}{B}) =0; \notag \\
   &\alpha^{\ast}(\ast D \overset{2}{A} - \ast \alpha(\overset{3}{B})) + k \ast(\overset{3}{\mathbb{G}}- k \overset{3}{B}) - D\ast(D \overset{3}{B} + \overset{2}{A}\wedge^{\vartriangleright}\overset{2}{B} - \beta(\overset{4}{C})) + \overline{\eta_1}(\ast (D \overset{4}{C} -\overset{2}{A}\wedge^{\vartriangleright}\overset{3}{C} + \overset{2}{B}\wedge^{\{, \}}\overset{3}{B}\notag \\
   &+ \overset{3}{B}\wedge^{\{, \}}\overset{2}{B}) , \overset{2}{B})+ \overline{\eta_2}(\ast (D \overset{4}{C}  - \overset{2}{A}\wedge^{\vartriangleright}\overset{3}{C} + \overset{2}{B}\wedge^{\{, \}}\overset{3}{B} + \overset{3}{B}\wedge^{\{, \}}\overset{2}{B}) , \overset{2}{B})  = 0;  \notag \\
	&D \ast(\overset{4}{\mathbb{H}} + k \overset{4}{C}) + \overset{2}{A}\wedge^{\vartriangleright}\ast (D \overset{4}{C} - \overset{2}{A}\wedge^{\vartriangleright}\overset{3}{C} + \overset{2}{B}\wedge^{\{, \}}\overset{3}{B} + \overset{3}{B}\wedge^{\{, \}}\overset{2}{B}) - \beta^{\ast}(\ast \overset{3}{\mathbb{G}} - k \ast \overset{3}{B})=0;\notag \\
	&k \ast (\overset{4}{\mathbb{H}} + k \overset{4}{C})- D \ast (D \overset{4}{C} - \overset{2}{A}\wedge^{\vartriangleright}\overset{3}{C} + \overset{2}{B}\wedge^{\{, \}}\overset{3}{B} + \overset{3}{B}\wedge^{\{, \}}\overset{2}{B})  -  \beta^{\ast}\ast (D \overset{3}{B} + \overset{2}{A}\wedge^{\vartriangleright} \overset{2}{B} - \beta(\overset{4}{C}))=0.\notag \\
   \end{align}
   We find that the above equations are equal to 
   \begin{align}\label{G3YM}
&d \star \mathscr{F} + \mathscr{A} \wedge^{[, ]}\star \mathscr{F} = \overline{\sigma}(\star \mathscr{G}, \mathscr{B}) + (-1)^n \overline{\kappa}(\star \mathscr{H}, \mathscr{C}); \notag \\
&d \star \mathscr{G} + \mathscr{A}\wedge^{\blacktriangleright}\star \mathscr{G} = - \overline{\eta_1}(\star \mathscr{H}, \mathscr{B})- \overline{\eta_2}(\star \mathscr{H}, \mathscr{B})- \alpha^{\ast}(\star \mathscr{F}); \notag \\
&d \star \mathscr{H} + \mathscr{A}\wedge^{\blacktriangleright}\star \mathscr{H} = \beta^{\ast}(\star\mathscr{G}).
   \end{align}
   In this notation, equations \eqref{G3YM} are also called generalized 3-form Yang-Mills equations, which can be obtained from these results of section \ref{sec3} and section \ref{sec4} by extending in the formalism to the generalized case.
   Specially,  when $\overset{2}{A} = 0$, $\overset{3}{B}=0$ and $\overset{4}{C}=0$, the action \eqref{GYM3} will become the ordinary case, and there are ordinary 3-form Yang-Mills equations
   \begin{align}
   	&d\ast \overset{2}{\mathbb{F}}+ \overset{1}{A}\wedge^{[,]}\ast \overset{2}{\mathbb{F}} =\overline{\sigma}(\ast \overset{3}{\mathbb{G}}, \overset{2}{B})+ (-1)^{n} \overline{\kappa}(\ast\overset{4}{\mathbb{H}}, \overset{3}{C}),\\[3mm]
   	&d \ast \overset{3}{\mathbb{G}} +\overset{1}{A}\wedge^\vartriangleright \ast \overset{3}{\mathbb{G}}=-\overline{\eta_2}(\ast \overset{4}{\mathbb{H}}, \overset{2}{B}) - \overline{\eta_1}(\ast \overset{4}{\mathbb{H}}, \overset{2}{B})-\alpha^*(\ast \overset{2}{\mathbb{F}}),\\[3mm]
   	&d\ast \overset{4}{\mathbb{H}}+\overset{1}{A}\wedge^\vartriangleright \ast\mathcal{H} = \beta^*(\ast \overset{3}{\mathbb{G}}).
   	\end{align}
     Besides, it can be seen directly that the 3-connection 1-form $\overset{1}{A}$, the 3-connection 2-form $\overset{2}{B}$ and  the 3-connection 3-form $\overset{3}{C}$ are flat from \eqref{GYM33}, i.e. $\overset{2}{\mathbb{F}}=0$, $\overset{3}{\mathbb{G}}=0$ and $\overset{4}{\mathbb{H}}=0$ in this case.

	    \begin{appendices}
	    	
	  \section{ Crossed modules and 2-crossed modules}\label{CM}
	  We review related concepts of crossed module and 2-crossed module given in \cite{Beaz, Crans, Brown, Radenkovic:2019qme, Martins:2010ry, TRMV1, YHMNRY, Mutlu1998FREENESSCF}.
	  \begin{definition}{(\textbf{pre-crossed module and crossed module})}
	  	A Lie pre-crossed module $\left(H,G;\alpha,\vartriangleright \right)$ is given by a Lie group map $\alpha: H  \longrightarrow G$ together with a smooth left action $\vartriangleright$ of G on H by automorphisms such that:
	  	\begin{equation}
	  		\alpha \left(g \vartriangleright h\right) = g \alpha \left(h\right) g^{-1},
	  	\end{equation}
	  	for each $g \in G $ and $h \in H$. The Peiffer commutators in a pre-crossed module are defined as $\lbrack\lbrack \cdot , \cdot \rbrack\rbrack : H \times H \longrightarrow H $ by
	  	\begin{equation}
	  		\lbrack\lbrack h , h' \rbrack\rbrack = h h' h^{-1}\left(\alpha \left(h\right) \vartriangleright h'^{-1}\right),
	  	\end{equation}
	  	for any $h, h' \in H$.
	  A Lie pre-crossed module is said to be  {\bfseries a Lie crossed module (or a strict Lie 2-group)}, if all of its Peiffer commutators are trivial, which is to say that
	  \begin{equation}
	  	\alpha \left(h\right) \vartriangleright h' = h h' h^{-1} ,
	  \end{equation}
	  for each $h, h' \in H$. 
	    \end{definition}
  
	  \begin{definition}{(\textbf{Differential pre-crossed modules})}
	  	A differential pre-crossed module $(\mathcal{h}, \mathcal{g}; \alpha, \vartriangleright)$ is given by a Lie algebra map $\alpha : \mathcal{h} \longrightarrow \mathcal{g}$ together with a left action $\vartriangleright$ of $\mathcal{g}$ on $\mathcal{h}$ by derivations such that:
	  	\begin{equation}\label{XY}
	  		\alpha(X\vartriangleright Y ) =\left[ X, \alpha(Y) \right],
	  	\end{equation}
	  	for each $X \in \mathcal{g}$ and $ Y \in \mathcal{h}$.
	  	The Peiffer commutators in a differential pre-crossed module are defined as $\lbrack\lbrack , \rbrack\rbrack : \mathcal{h} \times \mathcal{h} \longrightarrow \mathcal{h} $ by
	  	\begin{equation}\label{jkh}
	  		\lbrack\lbrack Y , Y'\rbrack\rbrack = \left[Y,Y'\right] - \alpha(Y)\vartriangleright Y',
	  	\end{equation}
	  	for each $Y, Y' \in \mathcal{h}$.
	  A differential pre-crossed module is said to be  {\bfseries a differential crossed module (or a strict Lie 2-algebra)}, if all of its Peiffer commutators vanish, which is to say that:
	  \begin{equation}\label{ YY'}
	  	\alpha(Y)\vartriangleright Y'=\left[Y,Y'\right],
	  \end{equation}
	  for each $Y,Y'\in \mathcal{h}$.
	    \end{definition}
  
	   \begin{definition}{(\textbf{Lie 2-crossed modules})}
	   	A Lie 2-crossed module $(L, H, G;\beta, \alpha, \vartriangleright, \left\{,\right\})$ is given by a complex of Lie groups:
	   	$$L \stackrel{\beta}{\longrightarrow}H \stackrel{\alpha}{\longrightarrow}G$$
	   	together with smooth left action $\vartriangleright$ by automorphisms of $G$ on $L$ and $H$ (and on $G$ by conjugation), i.e.
	   	\begin{equation}
	   		g \vartriangleright (e_1 e_2)=(g \vartriangleright e_1)(g \vartriangleright e_2), \ \ \ (g_1 g_2)\vartriangleright e = g_1 \vartriangleright (g_2 \vartriangleright e),
	   	\end{equation}
	   	for any $g, g_1, g_2\in G, e, e_1, e_2\in H $ or $L$, and a $G$-equivariant smooth function $ \left\{, \right\} :H \times H \longrightarrow L $, the Peiffer lifting, such that
	   	\begin{equation}
	   		g \vartriangleright \left\{ h_1, h_2 \right\} = \left\{g \vartriangleright h_1, g \vartriangleright h_2\right\},
	   	\end{equation}
	   	for any $g \in G$ and $h_1,h_2\in H$. They satisfy:
	   	\begin{enumerate}[1)]
	   		\item $L \stackrel{\beta}{\longrightarrow}H \stackrel{\alpha}{\longrightarrow}G$ is a complex of $G$-modules (in other words $\beta$ and $\alpha$ are $G$-equivariant and $\beta \circ \alpha$ maps $L$ to $1_G$, the identity of G);
	   		\setlength{\itemsep}{4pt}
	   		\item $\beta \left\{h_1,h_2\right\} =\lbrack\lbrack h_1 , h_2 \rbrack\rbrack $, for each $h_1, h_2\in H$;
	   		
	   		\setlength{\itemsep}{4pt}
	   		\item $ \left[l_1, l_2\right]= \left\{ \beta (l_1), \beta (l_2)\right\}$, for each $l_1,l_2 \in L$, and here $\left[l_1,l_2\right]=l_1 l_2 l_1 ^{-1} l_2 ^{-1}$;
	   		
	   		\setlength{\itemsep}{4pt}
	   		\item $\left\{ h_1 h_2, h_3 \right\}= \left\{h_1, h_2 h_3 h_2 ^{-1}\right\} \alpha (h_1) \vartriangleright \left\{h_2,h_3\right\}$, for each $h_1,h_2, h_3 \in H$;
	   		
	   		\setlength{\itemsep}{4pt}
	   		\item $ \left\{h_1,h_2 h_3\right\}  =  \left\{ h_1 , h_2  \right\} \left\{ h_1 , h_3  \right\} \left\{ \lbrack\lbrack h_1, h_3 \rbrack\rbrack  ^{-1}, \alpha (h_1) \vartriangleright h_2 \right\}$, for each $h_1,h_2, h_3 \in H$;
	   		
	   		\setlength{\itemsep}{4pt}
	   		\item $\left\{ \beta (l) , h\right\} \left\{h, \beta (l)\right\} = l (\alpha (h)\vartriangleright l^{-1})$, for each $h \in H $ and $ l \in L$.
	   	\end{enumerate}
	   	
	   \end{definition}
	   There is a left action of $H$ on $L$ by automorphisms $\vartriangleright' $ which is defined by
	   \begin{equation}
	   	h \vartriangleright' l = l \left\{ \beta (l)^{-1}, h \right\},
	   \end{equation}
	   for each $l \in L$ and $h \in H$. This together with the homomorphism $ \beta :\L \longrightarrow H $ defines a crossed module. In particular, for any $h\in H$,
	   $$h \vartriangleright' 1_L = \left\{ 1_H, h\right\} =\left\{h, 1_H\right\}=1_L,$$
	   where $1_H$ and $1_L$ are the identity of $H$ and $L$, respectively.
	   
	   \begin{definition}{(\textbf{Differential 2-crossed modules})}
	   	A differential 2-crossed module$(\mathcal{l},\mathcal{h}, \mathcal{g}; \beta, \alpha,\vartriangleright, \left\{ , \right\})$ is given by a complex of Lie algebras:
	   	$$\mathcal{l}\stackrel{\beta}{\longrightarrow} \mathcal{h} \stackrel{\alpha}{\longrightarrow} \mathcal{g},$$
	   	together with left action $\vartriangleright $ by derivations of $\mathcal{g}$ on $\mathcal{l}, \mathcal{h}, \mathcal{g}$ (on the latter by the adjoint representation), and a $\mathcal{g}$-equivariant bilinear map $\left\{ , \right\}:\mathcal{h} \times \mathcal{h} \longrightarrow \mathcal{l}$, the Peiffer lifting, such that
	   	\begin{equation}\label{12}
	   		X\vartriangleright \left\{ Y_1,Y_2\right\} = \left\{X\vartriangleright Y_1,Y_2\right\} + \left\{Y_1, X\vartriangleright Y_2\right\},
	   	\end{equation}
	   	for each $X \in \mathcal{g}$ and $Y_1,Y_2 \in \mathcal{h}$. They satisfy:
	   	\begin{enumerate}[1)]
	   		\setlength{\itemsep}{4pt}
	   		\item 	$\mathcal{l}\stackrel{\beta}{\longrightarrow} \mathcal{h} \stackrel{\alpha}{\longrightarrow} \mathcal{g}$ is a complex of $\mathcal{g}$-modules;
	   		
	   		\setlength{\itemsep}{4pt}
	   		\item $\beta \left\{Y_1,Y_2\right\} =\lbrack\lbrack Y_1 , Y_2\rbrack\rbrack $, for each $Y_1, Y_2\in \mathcal{h}$, where $\lbrack\lbrack Y_1 , Y_2\rbrack\rbrack=\left[ Y_1, Y_2\right] - \alpha (Y_1) \vartriangleright Y_2$;
	   		
	   		\setlength{\itemsep}{4pt}
	   		\item $ \left[Z_1, Z_2\right]= \left\{ \beta (Z_1), \beta (Z_2)\right\}$, for   each $Z_1,Z_2 \in L$;
	   		
	   		\setlength{\itemsep}{4pt}
	   		\item $\left\{ \left[Y_1 ,Y_2\right], Y_3 \right\}= \alpha (Y_1) \vartriangleright \left\{Y_2,Y_3\right\}+\left\{Y_1,\left[Y_2,Y_3\right]\right\}-\alpha (Y_2)\vartriangleright \left\{Y_1,Y_3\right\}-\left\{Y_2,\left[Y_1,Y_3\right]\right\}$, for each $Y_1,Y_2, Y_3 \in \mathcal{h}$. This is the same as :
	   		$$\left\{\left[Y_1,Y_2\right],Y_3\right\}= \left\{\alpha (Y_1)\vartriangleright Y_2, Y_3\right\} - \left\{ \alpha (Y_2)\vartriangleright Y_1, Y_3\right\} - \left\{Y_1, \beta \left\{ Y_2, Y_3\right\}\right\} + \left\{Y_2, \beta\left\{Y_1,Y_2 \right\}\right\};$$
	   		
	   		\setlength{\itemsep}{4pt}
	   		\item $ \left\{Y_1,\left[Y_2,Y_3\right]\right\}= \left\{ \beta\left\{Y_1,Y_2 \right\},Y_3  \right\}-\left\{ \beta\left\{Y_1,Y_3 \right\},Y_2 \right\}$ for each $Y_1,Y_2, Y_3 \in \mathcal{h}$;
	   		
	   		\setlength{\itemsep}{4pt}
	   		\item $\left\{ \beta (Z) , Y \right\} +\left\{Y, \beta (Z)\right\} =- (\alpha (Y)\vartriangleright Z)$, for each $Y\in \mathcal{h} $ and $ Z \in \mathcal
	   		{l}$.	
	   	\end{enumerate}
	   	
	   \end{definition}
	   
	   Analogously to the Lie 2-crossed module case, there is a left action of $\mathcal{h}$ on $\mathcal{l}$ which is defined by
	   \begin{equation}\label{YZ}
	   	Y\vartriangleright'Z= -\left\{ \beta(Z),Y\right\},
	   \end{equation}
	   for each $Y\in \mathcal{h}$ and $Z \in \mathcal{l}$. This together with the homomorphism $\beta :\mathcal{l} \longrightarrow \mathcal{h}$ defines a differential crossed module $(L, H; \beta, \vartriangleright')$.

	   \section{Lie algebra valued differential forms}\label{LAVDF}
	   We review the properties of Lie algebra valued differential forms based on higher groups, seeing \cite{SDH} for details.
	   \begin{proposition}
	   	For $A\in \Lambda^k (U, \mathcal{g})$, $A_1 \in \Lambda^{k_1} (U, \mathcal{g})$, $A_2 \in \Lambda^{k_2} (U, \mathcal{g})$, $B \in \Lambda^t (U, \mathcal{h})$, $B_1 \in \Lambda^{t_1} (U, \mathcal{h})$ and $ {B_2} \in \Lambda^{t_2} (U, \mathcal{h})$, have
	   		\begin{align}\label{C1}
	   		\alpha (A\wedge^\vartriangleright B)=A\wedge^{\left[,\right]}\alpha(B);
	   		\end{align}
	
	   		\begin{align}\label{C2}
	   		\alpha (B_1)\wedge^\vartriangleright B_2 = B_1\wedge^{\left[,\right]}B_2;
	   			\end{align}
   			
	   		\begin{align}\label{C3}
	   		A\wedge^\vartriangleright (B_1\wedge^{\left[,\right]} B_2)=(A\wedge^\vartriangleright B_1)\wedge^{\left[,\right]}B_2+ (-1)^{kt_1}B_1\wedge^{\left[,\right]}(A\wedge^\vartriangleright B_2);
	   		\end{align}

	   			\begin{align}\label{C4}
	   		(A_1\wedge^{\left[,\right]}A_2)\wedge^\vartriangleright B=A_1 \wedge^\vartriangleright(A_2\wedge^\vartriangleright B)+(-1)^{k_1 k_2 +1} A_2\wedge^\vartriangleright(A_1\wedge^\vartriangleright B).
	   			\end{align}
   		
	   \end{proposition}
	   \begin{proposition}
	   	\begin{enumerate}[1)]
	   		\setlength{\itemsep}{6pt}
	   		\item For $A\in \Lambda^k (U,\mathcal{g})$,$A'\in  \Lambda^{k'}(U,\mathcal{g})$ and $C\in \Lambda^* (U,\mathcal{l})$, have
	   		\begin{align}\label{C5}
	   			\beta(A\wedge^\vartriangleright C)=A\wedge^\vartriangleright \beta(C);
	   		\end{align}
	   		\begin{align}\label{C6}
	   			A\wedge^\vartriangleright A' = A\wedge A'+(-1)^{kk'+1} A'\wedge A.
	   		\end{align}
	   		
	   		\setlength{\itemsep}{6pt}
	   		\item For $A\in \Lambda^k (U,\mathcal{g})$,$B_1 \in \Lambda^{t_1}(U,\mathcal{h})$,$B_2 \in \Lambda^{t_2} (U,\mathcal{h})$ and $W \in \Lambda^* (U,\mathcal{w})$ $(\mathcal{w}=\mathcal{g},\mathcal{h},\mathcal{l})$, have
	   		\begin{align}\label{C7}
	   			d(A\wedge^\vartriangleright W)=dA\wedge^\vartriangleright W + (-1)^k A \wedge^\vartriangleright dW;
	   		\end{align}
	   		\begin{align}\label{C8}
	   			d(B_1\wedge^{\left\{,\right\}}B_2)=dB_1\wedge^{\left\{,\right\}}B_2+(-1)^{t_1} B_1 \wedge^{\left\{,\right\}}dB_2;
	   		\end{align}
	   		\begin{align}\label{C9}
	   			A\wedge^\vartriangleright (B_1\wedge^{\left\{,\right\}}B_2)=(A\wedge^\vartriangleright B_1)\wedge^{\left\{,\right\}}B_2 + (-1)^{kt_1}B_1\wedge^{\left\{,\right\}}(A\wedge^\vartriangleright B_2),
	   		\end{align}
	   	\end{enumerate}
	   \end{proposition}
  \begin{proposition}
  	For $A_1 \in \Lambda^{k_1}(M,\mathcal{g})$, $A_2 \in \Lambda^{k_2}(M,\mathcal{g})$, $A_3 \in \Lambda^{k_3}(M,\mathcal{g})$, $B_1\in \Lambda^{t_1}(M,\mathcal{h})$, $B_2\in \Lambda^{t_2}(M,\mathcal{h})$, $B_3\in \Lambda^{t_3}(M,\mathcal{h})$, $C_1\in \Lambda^{q_1}(M,\mathcal{l})$, $C_2\in \Lambda^{q_2}(M,\mathcal{l})$, and $C_3\in \Lambda^{q_3}(M,\mathcal{l})$, we have 
  	\begin{align} \label{C10}
  		\langle A_1 \wedge^{[,]}A_2, A_3\rangle = (-1)^{k_1k_2+1}\langle A_2, A_1 \wedge^{[,]}A_3 \rangle;
  	\end{align}
  	\begin{align}
  		\langle B_1 \wedge^{[,]}B_2, B_3\rangle = (-1)^{t_1t_2+1}\langle B_2, B_1 \wedge^{[,]}B_3 \rangle;
  	\end{align}
  	\begin{align}
  		\langle C_1 \wedge^{[,]}C_2, C_3\rangle = (-1)^{q_1q_2+1}\langle C_2, C_1 \wedge^{[,]}C_3 \rangle.
  	\end{align}
  \end{proposition}
\begin{proposition}
	For $A \in \Lambda^{k}(M,\mathcal{g})$, $B_1 \in \Lambda^{t_1}(M, \mathcal{h})$, $B_2 \in \Lambda^{t_2}(M, \mathcal{h})$, $C_1 \in \Lambda^{q_1}(M, \mathcal{l})$ and $C_2 \in \Lambda^{q_2}(M, \mathcal{l})$,  we have
	\begin{align} \label{C11}
		\langle B_1, A\wedge^\vartriangleright B_2\rangle =(-1)^{t_2(k+t_1)+kt_1+1}\langle B_2, A\wedge^\vartriangleright B_1 \rangle =(-1)^{kt_1+1}\langle A\wedge^\vartriangleright B_1, B_2\rangle;
	\end{align}
	\begin{align}\label{C_1AC_2}
		\langle C_1, A\wedge^\vartriangleright C_2\rangle =(-1)^{q_2(k+q_1)+kq_1+1}\langle C_2, A\wedge^\vartriangleright C_1 \rangle =(-1)^{kq_1+1}\langle A\wedge^\vartriangleright C_1,  C_2\rangle.
	\end{align}
\end{proposition}

There is a bilinear map $\overline{\sigma}: \Lambda^{t_1}(M, \mathcal{h}) \times \Lambda^{t_2}(M, \mathcal{h}) \longrightarrow \Lambda^{t_1+t_2}(M, \mathcal{g})$ defined by the rule
\begin{align}\label{C12}
	\langle \overline{\sigma}(B_1,B_2), A\rangle = (-1)^{kt_2+1} \langle B_1, A\wedge^\vartriangleright B_2\rangle,
\end{align}
for $A\in \Lambda^k(M, \mathcal{g}), B_1 \in \Lambda^{t_1}(M, \mathcal{h})$ and $B_2 \in \Lambda^{t_2}(M, \mathcal{h})$,
and the following identity holds
\begin{align}\label{C13}
	\langle  A,\overline{\sigma}(B_1,B_2)\rangle = (-1)^{t_1t_2+1} \langle  A\wedge^\vartriangleright B_2,B_1\rangle.
\end{align}
Similarly, there is a bilinear map $\overline{\kappa}: \Lambda^{q_1}(M, \mathcal{l}) \times \Lambda^{q_2}(M, \mathcal{l}) \longrightarrow \Lambda^{q_1+q_2}(M, \mathcal{g})$ defined by the rule
\begin{align}\label{C14}
	\langle \overline{\kappa}(C_1,C_2), A\rangle = (-1)^{kq_2+1} \langle C_1, A\wedge^\vartriangleright C_2\rangle,
\end{align} 
for $A\in \Lambda^k(M, \mathcal{g}), C_1 \in \Lambda^{q_1}(M, \mathcal{l})$ and $C_2 \in \Lambda^{q_2}(M, \mathcal{l})$,
and there is an identity
\begin{align}\label{C15}
	\langle  A,\overline{\kappa}(C_1,C_2)\rangle = (-1)^{q_1q_2+1} \langle  A\wedge^\vartriangleright C_2,C_1\rangle.
\end{align}
There are bilinear maps $\overline{\eta_i}: \Lambda^q (M,\mathcal{l}) \times \Lambda^t(M, \mathcal{h}) \longrightarrow \Lambda^{q+t}(M,\mathcal{h})$ defined by 
$$\overline{\eta_i}(C, B):= \sum\limits_{a,b} C^b \wedge B^a \eta_i(Z_b, Y_a),  \ \ \ \ i = 1,2,$$
for $B= \sum\limits_{a} B^a Y_a \in \Lambda^t(M, \mathcal{h})$, $C= \sum\limits_{b} C^b Z_b \in \Lambda^q(M, \mathcal{l})$,
and there are two identities
\begin{align}\label{C16}
	\langle  B_1 \wedge^{\left\{,  \right\}} B_2, C \rangle = (-1)^{t_1(t_2 + q) +1}\langle B_2, \overline{\eta_1}(C, B_1) \rangle
	= (-1)^{t_2 q +1}\langle B_1, \overline{\eta_2}(C, B_2) \rangle,
\end{align}
by using \eqref{YY'Z}.
Besides, for $A\in \Lambda^k (M, \mathcal{g})$, $B \in \Lambda^t (M, \mathcal{h})$ and $C \in \Lambda^q(M, \mathcal{l})$, there are
\begin{align}\label{C17}
	\langle B, \alpha^*(A)\rangle= \langle\alpha(B), A\rangle,
\end{align}
\begin{align}\label{C18}
	\langle C, \beta^*(B)\rangle= \langle\beta(C), B\rangle,
\end{align}
being induced by $\langle Y, \alpha^*(X)\rangle_\mathcal{h}=\langle \alpha(Y), X\rangle_\mathcal{g}$ and $\langle Z, \beta^*(Y)\rangle_\mathcal{l}=\langle \beta(Z), Y\rangle_\mathcal{h}$.

      \section{Higher Yang-Mills equations}\label{YME}
      	\begin{enumerate}
      		\item Generalized Yang-Mills equations  
      		
      		  We consider generalized Yang-Mills action
      		\begin{align}\label{GYM}
      			\bm{S_{GYM}}= \int_{\bm{\tilde{M}}} \ll \mathscr{F}, \star \mathscr{F} \gg.
      		\end{align}
      		Take the variational derivative of  \eqref{GYM}
      		\begin{align}
      			\delta \bm{S_{GYM}}= \int_{\bm{\tilde{M}}} \ll \delta \mathscr{F}, \star \mathscr{F} \gg +  \ll  \mathscr{F}, \delta \star \mathscr{F} \gg,
      		\end{align}
      		where 
      		\begin{align}
      			\delta \mathscr{F} = (\delta F + k \delta \overset{2}{A} , \delta D \overset{2}{A})
      		\end{align}
      		and
      		\begin{align}
      			\star \mathscr{F} = ((-1)^{n-1} \ast D \overset{2}{A}, \ast F + k \ast \overset{2}{A}).
      		\end{align}
      		Then we can see that
      		\begin{align}
      			\ll \delta \mathscr{F}, \star \mathscr{F} \gg = ( \langle \delta(F + k \overset{2}{A}), (-1)^n \ast D \overset{2}{A} \rangle, \langle \delta(F + k \overset{2}{A}), \ast (F + k \overset{2}{A})\rangle + \langle \delta D \overset{2}{A}, \ast D \overset{2}{A}\rangle),
      		\end{align}
      		\begin{align}
      			\ll  \mathscr{F}, \delta \star \mathscr{F} \gg = ( \langle F + k \overset{2}{A}, (-1)^n \delta \ast D \overset{2}{A} \rangle, \langle F + k \overset{2}{A},  \delta\ast (F + k \overset{2}{A})\rangle + \langle D \overset{2}{A}, \delta \ast D \overset{2}{A}\rangle).
      		\end{align}
      		It follows that
      		\begin{align*}
      			\delta \bm{S_{GYM}}&= \int_{\bm{\tilde{M}}}  ( \langle \delta(F + k \overset{2}{A}), (-1)^n \ast D \overset{2}{A} \rangle, \langle \delta(F + k \overset{2}{A}), \ast (F + k \overset{2}{A})\rangle + \langle \delta D \overset{2}{A}, \ast D \overset{2}{A}\rangle) \\
      			&+ \int_{\bm{\tilde{M}}}  ( \langle F + k \overset{2}{A}, (-1)^n \delta \ast D \overset{2}{A} \rangle, \langle F + k \overset{2}{A},  \delta\ast (F + k \overset{2}{A})\rangle + \langle D \overset{2}{A}, \delta \ast D \overset{2}{A}\rangle)\\
      			&= \int_{M}  \langle \delta(F + k \overset{2}{A}), \ast (F + k \overset{2}{A})\rangle + \langle \delta D \overset{2}{A}, \ast D \overset{2}{A}\rangle + \int_{M} \langle F+ k \overset{2}{A},  \delta\ast (F + k \overset{2}{A})\rangle + \langle D \overset{2}{A}, \delta \ast D \overset{2}{A}\rangle\\
      			&= 2 \int_{M}\langle \delta(F + k \overset{2}{A}), \ast (F + k \overset{2}{A})\rangle + \langle \delta D \overset{2}{A}, \ast D \overset{2}{A}\rangle\\
      			&= 2\int_{M}  \langle d \delta \overset{1}{A} + \overset{1}{A} \wedge^{[,]}\delta  \overset{1}{A}  + k \delta \overset{2}{A}, \ast F + k \ast \overset{2}{A} \rangle + \langle d \delta \overset{2}{A} + \delta \overset{1}{A} \wedge^{[,]}\overset{2}{A} - \delta \overset{2}{A} \wedge^{[,]} \overset{1}{A}, \ast D \overset{2}{A} \rangle )\\
      			&= 2\int_{M} \langle \delta \overset{1}{A}, d \ast F \rangle + \langle \delta \overset{1}{A}, \overset{1}{A}\wedge^{[,]} \ast F \rangle + \langle \delta \overset{2}{A}, k \ast F \rangle + \langle \delta \overset{1}{A}, k d \ast  \overset{2}{A} \rangle + \langle \delta \overset{1}{A}, \overset{1}{A}\wedge^{[,]} k \ast \overset{2}{A} \rangle  \\
      			&+\langle \delta \overset{2}{A}, k^2 \ast \overset{2}{A} \rangle - \langle \delta \overset{2}{A}, d \ast D \overset{2}{A} \rangle + \langle \delta \overset{1}{A}, \overset{2}{A} \wedge^{[,]}\ast D \overset{2}{A} \rangle - \langle \delta \overset{2}{A}, \overset{1}{A}\wedge^{[,]} \ast D \overset{2}{A} \rangle\\
      			&= 2\int_{M} \langle \delta \overset{1}{A}, d \ast F + \overset{1}{A}\wedge^{[,]}\ast F + k d \ast \overset{2}{A} + k \overset{1}{A}\wedge^{[,]} \ast  \overset{2}{A} + \overset{2}{A}\wedge^{[,]}\ast D \overset{2}{A} \rangle + \langle \delta \overset{2}{A}, k \ast F + k^2 \overset{2}{A} \\
      			&- d \ast D \overset{2}{A} - \overset{1}{A} \wedge^{[,]}\ast D \overset{2}{A}\rangle\\
      			&= 2\int_{M} \langle \delta \overset{1}{A}, D \ast F + k D \ast \overset{2}{A} + \overset{2}{A}\wedge^{[,]}\ast D \overset{2}{A} \rangle + \langle \delta \overset{2}{A}, k \ast F + k^2 \ast \overset{2}{A} - D \ast D \overset{2}{A} \rangle,
      		\end{align*}
      		by using \eqref{in3} and \eqref{C10}.
     \item Generalized 2-form Yang-Mills equations
    
   We consider generalized 2-form Yang-Mills action
   \begin{align}\label{act2}
   	\bm{S_{G2YM}}= \int_{\bm{\tilde{M}}} \ll \mathscr{F}, \star \mathscr{F} \gg +  \ll \mathscr{G}, \star \mathscr{G} \gg.
   \end{align}
It is well known that the generalized 2-form Yang-Mills fields are 
\begin{align*}
	&\mathscr{F} = (\overset{2}{\mathbb{F}} + k \overset{2}{A}, D \overset{2}{A}- \alpha(\overset{3}{B})),\\
	&\mathscr{G} = (\overset{3}{\mathbb{G}} - k \overset{3}{B}, D \overset{3}{B} + \overset{2}{A} \wedge^{\vartriangleright} \overset{2}{B}),
\end{align*}
and their dual forms are as follows using \eqref{star}
   \begin{align}
	&\star \mathscr{F} = ((-1)^{n-1} \ast (D \overset{2}{A} - \alpha (\overset{3}{B})), \ast (\overset{2}{\mathbb{F}}+ k  \overset{2}{A})),\\
	&\star \mathscr{G} = ((-1)^n \ast (D \overset{3}{B} + \overset{2}{A}\wedge^{\vartriangleright}\overset{2}{B}), \ast (\overset{3}{\mathbb{G}}- k  \overset{3}{B})).
\end{align}
Hence, the action \eqref{act2} can be rewrite as
  \begin{align}\label{act22}
	\bm{S_{G2YM}}&= \int_{\bm{\tilde{M}}} \ll (\overset{2}{\mathbb{F}} + k \overset{2}{A}, D \overset{2}{A}- \alpha(\overset{3}{B})),  ((-1)^{n-1} \ast (D \overset{2}{A} - \alpha (\overset{3}{B})), \ast (\overset{2}{\mathbb{F}}+ k \overset{2}{A})) \gg \notag \\
	&+ \int_{\bm{\tilde{M}}} \ll  (\overset{3}{\mathbb{G}} - k \overset{3}{B}, D \overset{3}{B} + \overset{2}{A} \wedge^{\vartriangleright} \overset{2}{B}), ((-1)^n \ast (D \overset{3}{B} + \overset{2}{A}\wedge^{\vartriangleright}\overset{2}{B}), \ast (\overset{3}{\mathbb{G}}- k  \overset{3}{B}))\gg  \notag \\
	&= \int_{M} \langle \overset{2}{\mathbb{F}} + k \overset{2}{A}, \ast ( \overset{2}{\mathbb{F}} + k \overset{2}{A}) \rangle + \langle D \overset{2}{A} - \alpha(\overset{3}{B}), \ast (D \overset{2}{A} - \alpha(\overset{3}{B}))\rangle \notag\\
	&+  \int_{M} \langle \overset{3}{\mathbb{G}} - k \overset{3}{B}, \ast(\overset{3}{\mathbb{G}} - k \overset{3}{B})\rangle + \langle D \overset{3}{B} + \overset{2}{A}\wedge^{\vartriangleright}\overset{2}{B}, \ast(D \overset{3}{B} + \overset{2}{A}\wedge^{\vartriangleright}\overset{2}{B}) \rangle.
\end{align}

   Take the variational derivative $\delta \bm{S_{G2YM}}$,
   \begin{align}
   	\delta \bm{S_{G2YM}}&=2\int_{M} \langle \delta(\overset{2}{\mathbb{F}} + k \overset{2}{A}), \ast ( \overset{2}{\mathbb{F}} + k \overset{2}{A}) \rangle + \langle \delta (D \overset{2}{A} - \alpha(\overset{3}{B})), \ast (D \overset{2}{A} - \alpha(\overset{3}{B}))\rangle \notag\\
   	&+ 2\int_{M} \langle \delta( \overset{3}{\mathbb{G}} - k \overset{3}{B}), \ast(\overset{3}{\mathbb{G}} - k \overset{3}{B})\rangle + \langle \delta (D \overset{3}{B} + \overset{2}{A}\wedge^{\vartriangleright}\overset{2}{B}), \ast(D \overset{3}{B} + \overset{2}{A}\wedge^{\vartriangleright}\overset{2}{B}) \rangle \notag\\
   	&= 2 \int_{M}   \langle d \delta \overset{1}{A} + \overset{1}{A} \wedge^{[,]} \delta \overset{1}{A} - \alpha(\delta \overset{2}{B}) + k \delta \overset{2}{A}, \ast (\overset{2}{\mathbb{F}} + k \overset{2}{A})\rangle  + \langle  d \delta \overset{2}{A} + \delta \overset{1}{A}  \wedge^{[,]} \overset{2}{A} - \delta \overset{2}{A} \wedge^{[,]}  \overset{1}{A} \notag\\
   	&- \alpha (\delta \overset{3}{B}), \ast (D \overset{2}{A} - \alpha (\overset{3}{B}))\rangle + \langle  d \delta \overset{2}{B} + \delta \overset{1}{A} \wedge^{\vartriangleright} \overset{2}{B} + \overset{1}{A}\wedge^{\vartriangleright}\delta \overset{2}{B}  - k \delta \overset{3}{B}, \ast (\overset{3}{\mathbb{G}} - k \overset{3}{B}) \rangle + \langle d \delta \overset{3}{B} \notag\\
   	&+ \delta \overset{1}{A}\wedge^{\vartriangleright} \overset{3}{B} + \overset{1}{A}\wedge^{\vartriangleright}\delta  \overset{3}{B}  + \delta \overset{2}{A} \wedge^{\vartriangleright}\overset{2}{B} + \overset{2}{A} \wedge^{\vartriangleright}\delta \overset{2}{B}, \ast (D \overset{3}{B} + \overset{2}{A}\wedge^{\vartriangleright} \overset{2}{B})\rangle \notag \\
   	&= 2 \int_{M}   \langle \delta \overset{1}{A}, d \ast (\overset{2}{\mathbb{F}} + k \overset{2}{A})\rangle + \langle  \delta \overset{1}{A}, \overset{1}{A} \wedge^{[,]} \ast (\overset{2}{\mathbb{F}} + k \overset{2}{A})\rangle - \langle \delta \overset{2}{B}, \alpha^{\ast}\ast (\overset{2}{\mathbb{F}} + k \overset{2}{A}) \rangle + \langle \delta \overset{2}{A}, \ast (k \overset{2}{\mathbb{F}} \notag\\
   	&+ k^2 \overset{2}{A}) \rangle - \langle \delta \overset{2}{A}, d \ast (D \overset{2}{A} - \alpha (\overset{3}{B})) \rangle + \langle \delta \overset{1}{A}, \overset{2}{A}\wedge^{[,]}\ast (D \overset{2}{A}- \alpha (\overset{3}{B})\rangle - \langle  \delta \overset{2}{A}, \overset{1}{A}\wedge^{[,]} \ast (D \overset{2}{A} \notag \\
   	&- \alpha (\overset{3}{B})) \rangle - \langle \delta \overset{3}{B}, \alpha^{\ast}\ast (D \overset{2}{A} - \alpha (\overset{3}{B})) \rangle - \langle \delta \overset{2}{B}, d \ast (\overset{3}{\mathbb{G}} - k \overset{3}{B})\rangle - \langle \delta \overset{1}{A}, \overline{\sigma}(\ast (\overset{3}{\mathbb{G}} - k \overset{3}{B}), \overset{2}{B})\rangle \notag \\
   	&- \langle \delta \overset{2}{B},  \overset{1}{A} \wedge^{\vartriangleright} \ast (\overset{3}{\mathbb{G}} - k \overset{3}{B})\rangle- \langle \delta \overset{3}{B}, \ast (k \overset{3}{\mathbb{G}} - k^2 \overset{3}{B})\rangle + \langle \delta \overset{3}{B}, d \ast (D \overset{3}{B} + \overset{2}{A}\wedge^{\vartriangleright} \overset{2}{B})\rangle \notag\\
   	&+ (-1)^{n+1} \langle \delta \overset{1}{A}, \overline{\sigma}(\ast (D \overset{3}{B} + \overset{2}{A}\wedge^{\vartriangleright} \overset{2}{B}),  \overset{3}{B}) \rangle   +  \langle \delta \overset{3}{B}, \overset{1}{A}\wedge^{\vartriangleright}\ast (D \overset{3}{B} +\overset{2}{A}\wedge^{\vartriangleright} \overset{2}{B}) \rangle- \langle \delta \overset{2}{A}, \overline{\sigma}(\ast (D \overset{3}{B} \notag\\
   	&+ \overset{2}{A}\wedge^{\vartriangleright}\overset{2}{B}), \overset{2}{B}) \rangle - \langle \delta \overset{2}{B},  \overset{2}{A} \wedge^{\vartriangleright}\ast (D \overset{3}{B}+ \overset{2}{A}\wedge^{\vartriangleright} \overset{2}{B})\rangle \notag\\
   	&= 2\int_{M} \langle \delta \overset{1}{A}, D \ast (\overset{2}{\mathbb{F}} + k \overset{2}{A}) + \overset{2}{A}\wedge^{[,]}\ast (D \overset{2}{A}- \alpha (\overset{3}{B}))- \overline{\sigma}(\ast(\overset{3}{\mathbb{G}}- k \overset{3}{B}), \overset{2}{B}) + (-1)^{n+1}\overline{\sigma}(\ast (D \overset{3}{B} \notag\\
   	&+ \overset{2}{A}\wedge^{\vartriangleright}\overset{2}{B}), \overset{3}{B})\rangle - \langle\delta \overset{2}{A}, D \ast (D \overset{2}{A} - \alpha (\overset{3}{B}) )- k \ast (\overset{2}{\mathbb{F}} + k \overset{2}{A}) + \overline{\sigma} (\ast (D \overset{3}{B}+ \overset{2}{A}\wedge^{\vartriangleright} \overset{2}{B}), \overset{2}{B})\rangle \notag \\
   	&- \langle \delta \overset{2}{B}, D \ast (\overset{3}{\mathbb{G}}- k \overset{3}{B} )+ \alpha^{\ast}\ast(\overset{2}{\mathbb{F}} + k \overset{2}{A}) + \overset{2}{A}\wedge^{\vartriangleright}\ast (D \overset{3}{B}+ \overset{2}{A}\wedge^{\vartriangleright}\overset{2}{B})\rangle \notag \\
   	&+ \langle \delta \overset{3}{B}, D \ast (D \overset{3}{B} + \overset{2}{A}\wedge^{\vartriangleright} \overset{2}{B})- \alpha^{\ast} \ast (D \overset{2}{A} - \alpha (\overset{3}{B})) -k \ast (\overset{3}{\mathbb{G}} - k \overset{3}{B}) \rangle,
   \end{align}
by using \eqref{C10}, \eqref{C11}, \eqref{C13} and \eqref{C17}.

  \item Generalized 3-form Yang-Mills equations
  
    We consider generalized 3-form Yang-Mills action
   \begin{align}\label{G3YM1}
   	\bm{S_{G3YM}}= \int_{\bm{\tilde{M}}} \ll \mathscr{F}, \star \mathscr{F} \gg +  \ll \mathscr{G}, \star \mathscr{G} \gg + \ll \mathscr{H}, \star \mathscr{H} \gg.
   \end{align}
The generalized 3-form Yang-Mills fields are given by 
	\begin{align*}
	&	\mathscr{F} = (\overset{2}{\mathbb{F}} + k \overset{2}{A}, D \overset{2}{A}- \alpha(\overset{3}{B}));\\
	&	\mathscr{G} = (\overset{3}{\mathbb{G}} - k \overset{3}{B}, D \overset{3}{B} + \overset{2}{A} \wedge^{\vartriangleright} \overset{2}{B} - \beta(\overset{4}{C}));\\
	&	\mathscr{H} = (\overset{4}{\mathbb{H}} + k \overset{4}{C}, D \overset{4}{C} - \overset{2}{A} \wedge^{\vartriangleright}\overset{3}{C} + \overset{2}{B}\wedge^{\{, \}}\overset{3}{B} +  \overset{3}{B}\wedge^{\{, \}}\overset{2}{B}),
\end{align*}
and their dual forms are as follows using  \eqref{star}
\begin{align}
	&\star \mathscr{F} = ((-1)^{n-1} \ast (D \overset{2}{A} - \alpha (\overset{3}{B})), \ast (\overset{2}{\mathbb{F}}+ k  \overset{2}{A}));\\
	&\star \mathscr{G} = ((-1)^n \ast (D \overset{3}{B} + \overset{2}{A}\wedge^{\vartriangleright}\overset{2}{B} - \beta(\overset{4}{C})), \ast (\overset{3}{\mathbb{G}}- k  \overset{3}{B}));\\
	&\star \mathscr{H} = ((-1)^{n-1} \ast( D \overset{4}{C} - \overset{2}{A} \wedge^{\vartriangleright}\overset{3}{C} + \overset{2}{B}\wedge^{\{, \}}\overset{3}{B} +  \overset{3}{B}\wedge^{\{, \}}\overset{2}{B}), \ast( \overset{4}{\mathbb{H}} + k \overset{4}{C})).
\end{align}
Hence, the action \eqref{G3YM1} can be rewrite as
   \begin{align*}\label{G3YM1}
	\bm{S_{G3YM}}&= \int_{\bm{\tilde{M}}} \ll  (\overset{2}{\mathbb{F}} + k \overset{2}{A}, D \overset{2}{A}- \alpha(\overset{3}{B})),((-1)^{n-1} \ast (D \overset{2}{A} - \alpha (\overset{3}{B})), \ast (\overset{2}{\mathbb{F}}+ k  \overset{2}{A})) \gg \\
	&+  \ll (\overset{3}{\mathbb{G}} - k \overset{3}{B}, D \overset{3}{B} + \overset{2}{A} \wedge^{\vartriangleright} \overset{2}{B} - \beta(\overset{4}{C})), ((-1)^n \ast (D \overset{3}{B} + \overset{2}{A}\wedge^{\vartriangleright}\overset{2}{B} - \beta(\overset{4}{C})), \ast (\overset{3}{\mathbb{G}}- k  \overset{3}{B})) \gg \\
	&+ \ll (\overset{4}{\mathbb{H}} + k \overset{4}{C}, D \overset{4}{C} - \overset{2}{A} \wedge^{\vartriangleright}\overset{3}{C} + \overset{2}{B}\wedge^{\{, \}}\overset{3}{B} +  \overset{3}{B}\wedge^{\{, \}}\overset{2}{B}), ((-1)^{n-1} \ast( D \overset{4}{C} - \overset{2}{A} \wedge^{\vartriangleright}\overset{3}{C}\\
	& + \overset{2}{B}\wedge^{\{, \}}\overset{3}{B} +  \overset{3}{B}\wedge^{\{, \}}\overset{2}{B}), \ast( \overset{4}{\mathbb{H}} + k \overset{4}{C}))\gg\\
	&= \int_{M} \langle \overset{2}{\mathbb{F}} + k \overset{2}{A}, \ast (\overset{2}{\mathbb{F}} + k \overset{2}{A})\rangle + \langle D \overset{2}{A}- \alpha(\overset{3}{B}), \ast (D \overset{2}{A}- \alpha(\overset{3}{B}))\rangle \\
	&+ \langle \overset{3}{\mathbb{G}} - k \overset{3}{B}, \ast (\overset{3}{\mathbb{G}} - k \overset{3}{B})\rangle + \langle D \overset{3}{B} + \overset{2}{A} \wedge^{\vartriangleright} \overset{2}{B} - \beta(\overset{4}{C}), \ast(D \overset{3}{B} + \overset{2}{A} \wedge^{\vartriangleright} \overset{2}{B} - \beta(\overset{4}{C}))\rangle \\
	&+ \langle \overset{4}{\mathbb{H}} + k \overset{4}{C}, \ast (\overset{4}{\mathbb{H}} + k \overset{4}{C})\rangle + \langle D \overset{4}{C} - \overset{2}{A} \wedge^{\vartriangleright}\overset{3}{C} + \overset{2}{B}\wedge^{\{, \}}\overset{3}{B} +  \overset{3}{B}\wedge^{\{, \}}\overset{2}{B}, \ast  ( D\overset{4}{C} - \overset{2}{A} \wedge^{\vartriangleright}\overset{3}{C} \\
	&+ \overset{2}{B}\wedge^{\{, \}}\overset{3}{B} +  \overset{3}{B}\wedge^{\{, \}}\overset{2}{B}) \rangle.
\end{align*}

   Take the variational derivative $\delta \bm{S_{G3YM}}$,
   \begin{align*}
   	\delta \bm{S_{G3YM}}&= 2\int_{M} \langle \delta(\overset{2}{\mathbb{F}} + k \overset{2}{A}), \ast (\overset{2}{\mathbb{F}} + k \overset{2}{A})\rangle + \langle \delta (D \overset{2}{A}- \alpha(\overset{3}{B})), \ast (D \overset{2}{A}- \alpha(\overset{3}{B}))\rangle \\
   	&+ \langle \delta (\overset{3}{\mathbb{G}} - k \overset{3}{B}), \ast (\overset{3}{\mathbb{G}} - k \overset{3}{B})\rangle + \langle \delta (D \overset{3}{B} + \overset{2}{A} \wedge^{\vartriangleright} \overset{2}{B} - \beta(\overset{4}{C})), \ast(D \overset{3}{B} + \overset{2}{A} \wedge^{\vartriangleright} \overset{2}{B} \\
   	&- \beta(\overset{4}{C}))\rangle + \langle \delta (\overset{4}{\mathbb{H}} + k \overset{4}{C}), \ast (\overset{4}{\mathbb{H}} + k \overset{4}{C})\rangle + \langle \delta (D \overset{4}{C} - \overset{2}{A} \wedge^{\vartriangleright}\overset{3}{C} + \overset{2}{B}\wedge^{\{, \}}\overset{3}{B}  \\
   	&+  \overset{3}{B}\wedge^{\{, \}}\overset{2}{B}), \ast  ( D\overset{4}{C} - \overset{2}{A} \wedge^{\vartriangleright}\overset{3}{C}+ \overset{2}{B}\wedge^{\{, \}}\overset{3}{B} +  \overset{3}{B}\wedge^{\{, \}}\overset{2}{B}) \rangle \\
   	&= 2 \int_{M} \langle \delta \overset{1}{A}, d \ast (\overset{2}{\mathbb{F}} + k \overset{2}{A})\rangle  + \langle \delta \overset{1}{A},  \overset{1}{A}\wedge^{[, ]}\ast (\overset{2}{\mathbb{F}} + k \overset{2}{A})\rangle  + \langle \langle \delta \overset{2}{A}, k \ast (\overset{2}{\mathbb{F}} + k \overset{2}{A})\rangle \\
   	&- \langle \delta \overset{2}{B}, \alpha^{\ast} \ast (\overset{2}{\mathbb{F}} + k \overset{2}{A})\rangle- \langle \delta \overset{2}{A},  d \ast (D \overset{2}{A} - \alpha(\overset{3}{B}))\rangle  - \langle  \delta \overset{2}{A}, \overset{1}{A}\wedge^{[, ]}\ast (D \overset{2}{A} - \alpha(\overset{3}{B}))\rangle \\
   	&+ \langle \delta \overset{1}{A}, \overset{2}{A}\wedge^{[, ]}\ast (D \overset{2}{A} - \alpha(\overset{3}{B}))\rangle - \langle \delta \overset{3}{B}, \alpha^{\ast}(\ast (D \overset{2}{A} - \alpha(\overset{3}{B})))\rangle \\
   	&- \langle \delta \overset{2}{B}, d \ast (\overset{3}{\mathbb{G}} - k \overset{3}{B})\rangle - \langle \delta \overset{1}{A}, \overline{\sigma}(\ast (\overset{3}{\mathbb{G}} - k \overset{3}{B}), \overset{2}{B})\rangle - \langle \delta \overset{2}{B}, \overset{1}{A}\wedge^{\vartriangleright}\ast (\overset{3}{\mathbb{G}} - k \overset{3}{B})\rangle \\
   	&- \langle \delta \overset{3}{C}, \beta^{\ast}(\ast \overset{3}{\mathbb{G}} - k \ast \overset{3}{B}) \rangle - \langle \delta \overset{3}{B}, k \ast (\overset{3}{\mathbb{G}} - k \overset{3}{B})\rangle + \langle \delta \overset{3}{B}, d \ast (D \overset{3}{B} + \overset{2}{A}\wedge^{\vartriangleright}\overset{2}{B} - \beta(\overset{4}{C}))\rangle \\
   	&+ (-1)^{n+1}\langle \delta \overset{1}{A}, \overline{\sigma}(\ast(D \overset{3}{B} + \overset{2}{A}\wedge^{\vartriangleright}\overset{2}{B} - \beta(\overset{4}{C})), \overset{3}{B})\rangle + \langle \delta \overset{3}{B}, \overset{1}{A}\wedge^{\vartriangleright}\ast(D \overset{3}{B} + \overset{2}{A}\wedge^{\vartriangleright}\overset{2}{B} \\
   	&- \beta(\overset{4}{C}))\rangle - \langle \delta \overset{2}{A},  \overline{\sigma}(\ast (D \overset{3}{B} + \overset{2}{A}\wedge^{\vartriangleright}\overset{2}{B} - \beta(\overset{4}{C})), \overset{2}{B})\rangle - \langle \delta \overset{2}{B}, \overset{2}{A}\wedge^{\vartriangleright}\ast(D \overset{3}{B} + \overset{2}{A}\wedge^{\vartriangleright}\overset{2}{B}\\
   	& - \beta(\overset{4}{C}))\rangle - \langle \delta \overset{4}{C}, \beta^{\ast}(\ast(D \overset{3}{B} + \overset{2}{A} \wedge^{\vartriangleright}\overset{2}{B} - \beta(\overset{4}{C})))\rangle\\
   	&+  \langle \delta \overset{3}{C}, d \ast (\overset{4}{\mathbb{H}} + k \overset{4}{C}) \rangle + (-1)^{n+1}\langle \delta \overset{1}{A}, \overline{\kappa}(\ast (\overset{4}{\mathbb{H}} + k \overset{4}{C}), \overset{3}{C})\rangle + \langle \delta \overset{3}{C}, \overset{1}{A}\wedge^{\vartriangleright}\ast(\overset{4}{\mathbb{H}} + k \overset{4}{C})\rangle \\
   	&- \langle \delta \overset{2}{B}, \overline{\eta_2}(\ast (\overset{4}{\mathbb{H}} + k \overset{4}{C}), \overset{2}{B})\rangle - \langle \delta \overset{2}{B}, \overline{\eta_1}(\ast (\overset{4}{\mathbb{H}} + k \overset{4}{C}), \overset{2}{B})\rangle + \langle \delta \overset{4}{C}, k \ast (\overset{4}{\mathbb{H}} + k \overset{4}{C})\rangle - \langle \delta \overset{4}{C}, \\
   	&d \ast (D \overset{4}{C} - \overset{2}{A}\wedge^{\vartriangleright}\overset{3}{C} + \overset{2}{B}\wedge^{\{, \}}\overset{3}{B} + \overset{3}{B}\wedge^{\{, \}}\overset{2}{B}) \rangle - \langle \delta \overset{1}{A}, \overline{\kappa}(\ast (D \overset{4}{C} - \overset{2}{A}\wedge^{\vartriangleright}\overset{3}{C} + \overset{2}{B}\wedge^{\{, \}}\overset{3}{B} \\
   	&+ \overset{3}{B}\wedge^{\{, \}}\overset{2}{B}) , \overset{4}{C})- \langle \delta \overset{4}{C} , \overset{1}{A}\wedge^{\vartriangleright}\ast (D \overset{4}{C} - \overset{2}{A}\wedge^{\vartriangleright}\overset{3}{C} + \overset{2}{B}\wedge^{\{, \}}\overset{3}{B} + \overset{3}{B}\wedge^{\{, \}}\overset{2}{B})  \rangle \\
   	&+ (-1)^{n +1}\langle \delta\overset{2}{A}, \overline{\kappa}(\ast (D \overset{4}{C} - \overset{2}{A}\wedge^{\vartriangleright}\overset{3}{C} + \overset{2}{B}\wedge^{\{, \}}\overset{3}{B} + \overset{3}{B}\wedge^{\{, \}}\overset{2}{B}) , \overset{3}{C})\rangle  + \langle \delta \overset{3}{C}, \overset{2}{A}\wedge^{\vartriangleright}\ast (D \overset{4}{C} -\\
   	& \overset{2}{A}\wedge^{\vartriangleright}\overset{3}{C} + \overset{2}{B}\wedge^{\{, \}}\overset{3}{B} + \overset{3}{B}\wedge^{\{, \}}\overset{2}{B})  \rangle   + (-1)^n \langle \delta \overset{2}{B}, \overline{\eta_2}(\ast (D \overset{4}{C} - \overset{2}{A}\wedge^{\vartriangleright}\overset{3}{C} + \overset{2}{B}\wedge^{\{, \}}\overset{3}{B} \\
   	&+ \overset{3}{B}\wedge^{\{, \}}\overset{2}{B}) ,  \overset{3}{B})\rangle -  \langle \delta \overset{3}{B}, \overline{\eta_1}(\ast (D \overset{4}{C} - \overset{2}{A}\wedge^{\vartriangleright}\overset{3}{C} + \overset{2}{B}\wedge^{\{, \}}\overset{3}{B} + \overset{3}{B}\wedge^{\{, \}}\overset{2}{B}) , \overset{2}{B})\rangle -  \langle \delta \overset{3}{B}, \\
   	&\overline{\eta_2}(\ast (D \overset{4}{C} - \overset{2}{A}\wedge^{\vartriangleright}\overset{3}{C}  + \overset{2}{B}\wedge^{\{, \}}\overset{3}{B} + \overset{3}{B}\wedge^{\{, \}}\overset{2}{B}) , \overset{2}{B})\rangle  + (-1)^n \langle \delta \overset{2}{B}, \overline{\eta_1}(\ast (D \overset{4}{C} - \overset{2}{A}\wedge^{\vartriangleright}\overset{3}{C} \\
   	&+ \overset{2}{B}\wedge^{\{, \}}\overset{3}{B} + \overset{3}{B}\wedge^{\{, \}}\overset{2}{B}) ,  \overset{3}{B})\rangle\\
   	&= 2 \int_{M} \langle \delta \overset{1}{A}, D \ast (\overset{2}{\mathbb{F}} + k \overset{2}{A}) + \overset{2}{A}\wedge^{[, ]}\ast (D \overset{2}{A} - \alpha(\overset{2}{B}) ) - \overline{\sigma}(\ast (\overset{3}{\mathbb{G}} - k \overset{3}{B}), \overset{2}{B}) + (-1)^{n+1} \overline{\sigma}(\ast (D \overset{3}{B} \\
   	&+ \overset{2}{A}\wedge^{\vartriangleright}\overset{2}{B} - \beta(\overset{4}{C})), \overset{3}{B}) + (-1)^{n+1}\overline{\kappa}(\ast (\overset{4}{\mathbb{H}}+ k \overset{4}{C}), \overset{3}{C}) - \overline{\kappa}(\ast (D \overset{4}{C} - \overset{2}{A}\wedge^{\vartriangleright}\overset{3}{C} + \overset{2}{B}\wedge^{\{, \}}\overset{3}{B} \\
   	&+ \overset{3}{B}\wedge^{\{, \}}\overset{2}{B}) , \overset{4}{C}) \rangle+ \langle \delta \overset{2}{A}, k \ast (\overset{2}{\mathbb{F}} + k \overset{2}{A}) - D \ast (D \overset{2}{A} - \alpha(\overset{3}{B})) - \overline{\sigma}(\ast (D \overset{3}{B} + \overset{2}{A}\wedge^{\vartriangleright}\overset{2}{B} - \beta(\overset{4}{C})),  \\
   	&\overset{2}{B})+ (-1)^{n+1} \overline{\kappa}(\ast (D \overset{4}{C}- \overset{2}{A}\wedge^{\vartriangleright}\overset{3}{C} + \overset{2}{B}\wedge^{\{, \}}\overset{3}{B} + \overset{3}{B}\wedge^{\{, \}}\overset{2}{B}) , \overset{3}{C})\rangle \\
   	&- \langle \delta \overset{2}{B}, \alpha^{\ast}(\ast \overset{2}{\mathbb{F}} + k \ast \overset{2}{A}) + D \ast (\overset{3}{\mathbb{G}} - k \overset{3}{B}) + \overset{2}{A}\wedge^{\vartriangleright}\ast (D \overset{3}{B} + \overset{2}{A}\wedge^{\vartriangleright}\overset{2}{B} - \beta(\overset{4}{C}))+ \overline{\eta_2}(\ast (\overset{4}{\mathbb{H}}  \\
   	&+ k \overset{4}{C}), \overset{2}{B}) +  \overline{\eta_1}(\ast (\overset{4}{\mathbb{H}}+ k \overset{4}{C}), \overset{2}{B}) + (-1)^{n+1}\overline{\eta_1}(\ast (D \overset{4}{C} - \overset{2}{A}\wedge^{\vartriangleright}\overset{3}{C} + \overset{2}{B}\wedge^{\{, \}}\overset{3}{B} + \overset{3}{B}\wedge^{\{, \}}\overset{2}{B}) , \\
   	&\overset{3}{B}) + (-1)^{n+1}\overline{\eta_2}(\ast (D \overset{4}{C} - \overset{2}{A}\wedge^{\vartriangleright}\overset{3}{C}  + \overset{2}{B}\wedge^{\{, \}}\overset{3}{B} + \overset{3}{B}\wedge^{\{, \}}\overset{2}{B}) , \overset{3}{B}) \rangle \\
   	&- \langle \delta \overset{3}{B}, \alpha^{\ast}(\ast D \overset{2}{A} - \ast \alpha(\overset{3}{B})) + k \ast(\overset{3}{\mathbb{G}}- k \overset{3}{B}) - D\ast(D \overset{3}{B} + \overset{2}{A}\wedge^{\vartriangleright}\overset{2}{B} - \beta(\overset{4}{C})) + \overline{\eta_1}(\ast (D \overset{4}{C} -\\
   	&\overset{2}{A}\wedge^{\vartriangleright}\overset{3}{C} + \overset{2}{B}\wedge^{\{, \}}\overset{3}{B} + \overset{3}{B}\wedge^{\{, \}}\overset{2}{B}) , \overset{2}{B})+ \overline{\eta_2}(\ast (D \overset{4}{C}  - \overset{2}{A}\wedge^{\vartriangleright}\overset{3}{C} + \overset{2}{B}\wedge^{\{, \}}\overset{3}{B} + \overset{3}{B}\wedge^{\{, \}}\overset{2}{B}) , \overset{2}{B}) \rangle\\
   	&+ \langle \delta \overset{3}{C}, D \ast(\overset{4}{\mathbb{H}} + k \overset{4}{C}) + \overset{2}{A}\wedge^{\vartriangleright}(\ast (D \overset{4}{C} - \overset{2}{A}\wedge^{\vartriangleright}\overset{3}{C} + \overset{2}{B}\wedge^{\{, \}}\overset{3}{B} + \overset{3}{B}\wedge^{\{, \}}\overset{2}{B}) ) - \beta^{\ast}(\ast \overset{3}{\mathbb{G}} \\
   	&- k \ast \overset{3}{B})\rangle+\langle \delta\overset{4}{C}, k \ast (\overset{4}{\mathbb{H}} + k \overset{4}{C})- D \ast (D \overset{4}{C} - \overset{2}{A}\wedge^{\vartriangleright}\overset{3}{C} + \overset{2}{B}\wedge^{\{, \}}\overset{3}{B} + \overset{3}{B}\wedge^{\{, \}}\overset{2}{B})  -  \beta^{\ast}\ast (D \overset{3}{B} \\
   	&+ \overset{2}{A}\wedge^{\vartriangleright} \overset{2}{B} - \beta(\overset{4}{C}))\rangle,
   \end{align*}
by using \eqref{C10}, \eqref{C11}, \eqref{C_1AC_2}, \eqref{C13}, \eqref{C15}, \eqref{C16}, \eqref{C17} and \eqref{C18}.
	\end{enumerate}
	    \end{appendices}
    \section*{Acknowledgment}
    This work is supported by the National Natural Science Foundation of China (Nos.11871350, NSFC no. 11971322).
    
		\bibliographystyle{plain}
	
\end{document}